\documentclass[a4paper,UKenglish,cleveref, autoref, thm-restate]{lipics-v2021}

\hideLIPIcs
\nolinenumbers

\usepackage{amssymb}
\usepackage{amsmath}
\usepackage{stmaryrd}
\usepackage{xcolor}
\usepackage{multicol}
\usepackage{graphicx}
\usepackage{float}
\usepackage{booktabs}
\usepackage{makecell}
\usepackage{hyperref}
\usepackage{tikz-cd}
\usepackage[inline]{enumitem}
\usepackage{braket}

\newcommand{\Nat}{{\mathbb N}}

\newcommand{\vcen}[1]{\begingroup
	\setbox0=\hbox{\includegraphics[scale=0.1]{#1}}%
	\parbox{\wd0}{\box0}\endgroup}

\newcommand{\Bool}{\mathbb{B}}

\newcommand{\Dist}{\mathcal{D}}

\newcommand{\ol}[1]{\mathbf{#1}}
\newcommand{\den}[1]{\llbracket #1 \rrbracket}

\newcommand{\C}{\mathcal{C}}
\newcommand{\tom}{\twoheadrightarrow}
\newcommand{\tv}{\mathsf{tv}}
\newcommand{\MP}{\mathsf{PBCirc}}
\newcommand {\PI}[1]{\textsf{PI}(#1)}

\newcommand{\sort}{\mathcal{S}}
\newcommand{\sign}{\Sigma}

\newcommand{\id}[1]{id_{#1}}

\newcommand{\discharger}[1]{!_{#1}}
\newcommand{\copier}[1]{\blacktriangleleft_{#1}}

\newcommand{\gen}{s}
\newcommand{\ari}{\mathsf{ar}}
\newcommand{\coar}{\mathsf{coar}}
\newcommand{\Cat}[1]{\mathsf{#1}}

\newcommand{\FStoch}{\Cat{FinStoch}}

\makeatletter
\newcommand{\mylabel}[2]{
	\begingroup
    \protected@edef\@currentlabel{#2}%
    \label{#1}%
  \endgroup
}
\makeatother

\newcommand{\br}[1]{\llparenthesis #1 \rrparenthesis}

\title{Parametric Iteration in Resource Theories} %

\author{Alessandro Di Giorgio}{Tallinn University of Technology, Estonia}{alessd@taltech.ee}{https://orcid.org/0000-0002-6428-6461}{}

\author{Pawel Sobocinski}{Tallinn University of Technology, Estonia}{}{https://orcid.org/0000-0002-7992-9685}{}

\author{Niels Voorneveld}{Cybernetica, Estonia}{}{https://orcid.org/0000-0001-6650-3493}{}

\authorrunning{A. Di Giorgio, P. Sobocinski and N. Voorneveld} %

\Copyright{Alessandro Di Giorgio, Pawel Sobocinski and Niels Voorneveld}



 \acknowledgements{Di Giorgio, Voorneveld and Sobocinski were supported by the European Union under Grant Agreement No. 101087529. Sobocinski was additionally supported by the
Estonian Research Council via grants PRG1215 and TEM-TA5 and the Estonian Center of Excellence in Artificial Intelligence (EXAI).}%

\EventEditors{John Q. Open and Joan R. Access}
\EventNoEds{2}
\EventLongTitle{42nd Conference on Very Important Topics (CVIT 2016)}
\EventShortTitle{CVIT 2016}
\EventAcronym{CVIT}
\EventYear{2016}
\EventDate{December 24--27, 2016}
\EventLocation{Little Whinging, United Kingdom}
\EventLogo{}
\SeriesVolume{42}
\ArticleNo{23}

\begin{document}

\maketitle

\begin{abstract}
Many algorithms are specified with respect to a fixed but unspecified parameter.
Examples of this are especially common in cryptography, where protocols often feature a security parameter such as the bit length of a secret key.

Our aim is to capture this phenomenon in a more abstract setting. We focus on resource theories --- general calculi of processes with a string diagrammatic syntax --- introducing a general parametric iteration construction. 
By instantiating this construction within the Markov category of probabilistic Boolean circuits and equipping it with a suitable metric, we are able to capture the notion of negligibility via asymptotic equivalence, in a compositional way.
This allows us to use diagrammatic reasoning to prove simple cryptographic theorems -- for instance, proving that guessing a randomly generated key has negligible success.
\end{abstract}

\section{Introduction}\label{sec:intro}

Cryptographic protocols are designed to resist attacks by adversaries with bounded computational power.
A fundamental feature is the use of a security parameter $\sigma$, which governs the time complexity of potential attacks. As $\sigma$ increases, cryptographic schemes become more secure: for example, in public-key encryption, key lengths are chosen in relation to $\sigma$ to ensure that brute-force attacks require time exponential in $\sigma$ to succeed. More generally, security is often defined in asymptotic terms, stating that certain probabilities—such as the success rate of an adversary—become \emph{negligible} as $\sigma$ grows.
This notion of negligibility from cryptography~\cite{Modern_Cryptography} also appears in other computational settings, such as approximation schemes in probabilistic algorithms. In both cases, reasoning about security or approximation guarantees requires a formal understanding of how computations depend on the parameter~$\sigma$.

Traditionally, the definition of a cryptographic system follows a bottom-up approach: one chooses a model of computation, an algorithm and its notion of complexity, efficiency and security.
This usually entails several technicalities that tend to hinder clarity and modularity. In particular, security proofs often become bogged down with low-level details of the computational model, making it difficult to generalise results, compare different cryptographic frameworks, or reason about their composition in a structured way.

This situation motivated the study of the subject at a higher level of abstraction, starting with the work of Canetti on Universal Composable Security~\cite{canetti2001universally} and followed by the work of Maurer and Renner on Abstract Cryptography~\cite{MaurerR11,maurer2011constructive}.
Motivated by similar intents, Broadbent and Karvonen~\cite{Broadbent_2023} propose an abstract model of composable security definitions in terms of \emph{resource theories}, formalised as symmetric monoidal categories.

The resource theoretic approach comes with the visual and intuitive language of string diagrams~\cite{joyal1991geometry}. Consider the following example:
\begin{equation}\label{ex:resth}
    \vcen{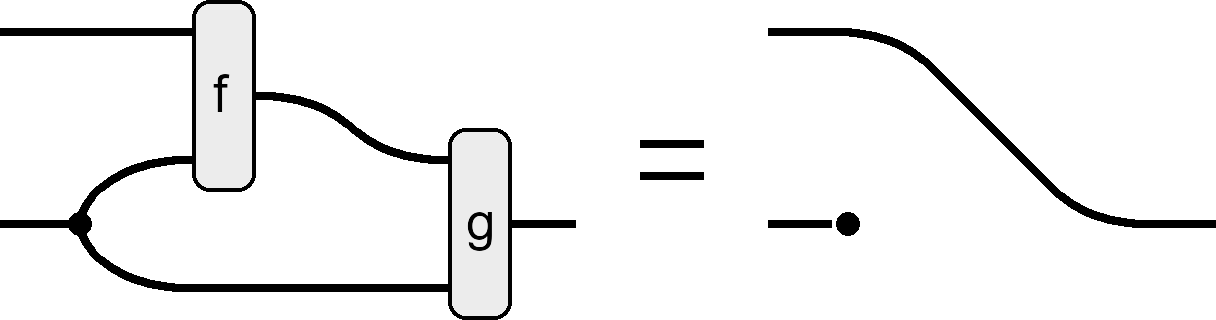}
\end{equation}
In the diagram above -- read from left to right -- we think of $f$ and $g$ as processes that transform the resources carried on the wires. The explicit splitting of wires via a \emph{copier}, indicates that $f$ and $g$ share the second resource. If a resource is unused, as in the diagram on the right hand side, we simply \emph{discard} it by closing the wire.

In cryptographic terms, we interpret the first input wire in~\eqref{ex:resth} as carrying a message and the second as carrying a key. The processes $ f $ and $ g $ correspond to encryption and decryption, respectively. The equation postulates the correctness of the protocol: the original message -- transformed into a cipher text by $f$ -- corresponds to the one decrypted by $g$. Equations between string diagrams such as these lead to proofs via \emph{diagrammatic reasoning}, a powerful axiomatic technique that allows compositional reasoning about  compound systems.

\medskip
Our aim in this paper is to capture the notion of parametric algorithms within the abstract framework of resource theories.
After a brief introduction to resource theories and their string diagrammatic language in Section~\ref{sec:string diagrams}, in Section~\ref{sec:iteration} we define a notion of \emph{iteration} bounded by an unspecified parameter, which we call $\star$-iteration. Perhaps unexpectedly, we do not study traces (as, e.g., in monoidal streams~\cite{Monoidal_Streams}) but instead define a functorial construction:  given a morphism $f \colon S \otimes A \to B \otimes S$, we define $\tau^k_{S , (A) , (B)}(f) \colon S \otimes A^k \to B^k \otimes S$ as its $k$-fold iteration. Below we give a diagrammatic representation of $\tau^k_{S,(A),(B)}(f)$, when $k = 3$.
\begin{center}
$\vcen{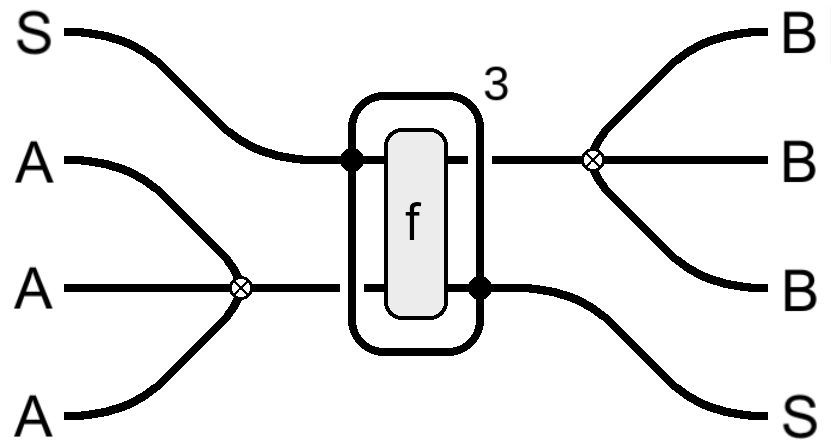} = \vcen{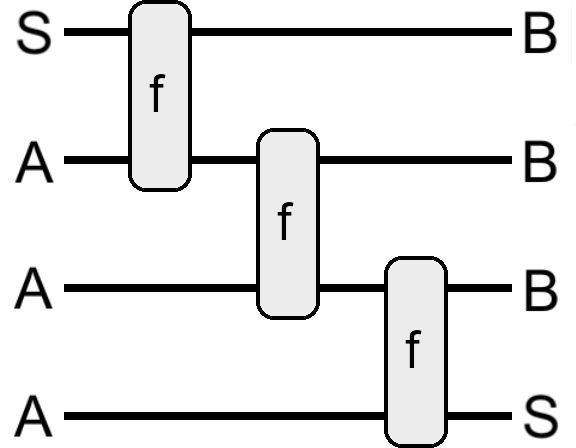}$
\end{center}
Iterated morphism consumes and produces \emph{lists} of resources of type $A$ and $B$, respectively. The resource of type $S$ can be thought of as a \emph{state} carried over by the computation. A similar construction is found in probabilistic programming languages~\cite{holtzen2020scaling,torres2024iteration}, where iteration is usually bounded. Interestingly, useful results can be proved about iteration \emph{without} instantiating the parameter. To this end, we introduce a technique we dub the \emph{diagonal method} that allows us to prove useful, general results  that add to the arsenal of diagrammatic reasoning about abstract resource theories with parametric iteration.

\medskip
Parametric iteration is a feature that extends the expressivity of resource theories.
Our main application in this paper is motivated by cryptography, and the notion of \emph{negligibility} in particular. To do this, we  work within a \emph{probabilistic} resource theoretical context.
This is the category $\FStoch_{\{1,0\}}$ of Boolean stochastic maps. It is both a Markov category and one that possesses a \emph{probabilistic choice structure}, which we introduce in Section~\ref{sec:preliminaries}.
As a notable property, the probabilistic case structure ensures a unique (up to isomorphism) normal form for probabilistic Boolean functions. This, in turn, allows us to derive a completeness result (Theorem~\ref{thm:complete}). While not essential to the main developments of the paper, it provides an interesting alternative proof to the one found in~\cite{piedeleu2025boolean}.

 Formally, we first move from symmetric monoidal categories to metric-enriched ones, so that morphisms, i.e. processes, are equipped with a notion of \emph{distance} between them (see e.g. \cite{dallago_et_al:LIPIcs.FSCD.2022.4,quantitativemonoidalalgebra}).

We reconcile metric-enrichment with the parameterized iteration construction in Section~\ref{sec:asyeq}, where we introduce a notion of \emph{asymptotic equivalence} on morphisms. This captures negligibility—also referred to as \emph{indistinguishability}—between processes. Intuitively, two processes $f$ and $g$, potentially involving parametric iterations, are considered indistinguishable if their distance approaches $0$ as the parameter $\sigma$ tends to $\infty$.
In related work \cite{Lago-Indistinct}, such notions of indistinguishability are studied in a bounded linear $\lambda$-calculus.

Guided by cryptographic intuitions, we conclude with some examples. Notably, we formally prove that a randomly generated key is unlikely to be any specific key (Lemma~\ref{lem:keyguess}), and show we can asymptotically approximate fair choice with weighted choice with the iterated von Neumann trick (Lemma~\ref{eq:vn trick}). We also discuss how we can perform arithmetic over the cyclic groups of sizes $2^\sigma$ and $2^\sigma-1$, which is useful for instance in Diffie-Hellman key exchange and Zero knowledge proofs.

The paper presents two theoretical contributions. We define the parametrized iteration construction on symmetric monoidal categories, prove that it forms a category, preserves several categorical structures, and admits an endofuctor $\tau^\star$ for parametrized iteration. We moreover define the notion of asymptotic equivalence on the parametrized iteration over a metric-enriched symmetric monoidal category, and show that it defines a congruence and is preserved under the application of $\tau^\star$.

The paper includes appendices containing additional material and omitted proofs.

\section{Symmetric Monoidal Categories as Resource Theories}\label{sec:string diagrams}%
In this section we give a brief introduction to symmetric monoidal categories (SMCs) and their associated graphical calculus~\cite{Joyal1991,Selinger2009,piedeleu2023introduction} in terms of resource theories, following~\cite{coecke2016mathematical}.

The basic data consists of objects $A,B,C,\ldots$ that we think of as types of resources, and morphisms $f \colon A \to B$, which we think of as transformations that convert a resource of type $A$ into one of type $B$. Morphisms can be composed sequentially: given $f \colon A \to B$ and $g \colon B \to C$, their composite $f ; g \colon A \to C$ represents performing $f$ followed by $g$.

The monoidal product $\otimes$ lets us express having two resources at once—so $A \otimes C$ means having both $A$ and $C$—and it extends to morphisms too: given $f \colon A \to B$ and $g \colon C \to D$, the morphism $f \otimes g \colon A \otimes C \to B \otimes D$ describes applying $f$ and $g$ in parallel. There is also a void resource $I$ that acts neutrally under combination: $A \otimes I = A$.

These ideas are captured visually using the graphical calculus of string diagrams. A resource type $A$ is drawn as a wire \vcen{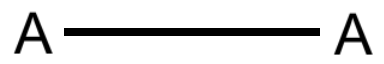}, and a transformation $f \colon A_1 \otimes \ldots \otimes A_n \to B_1 \otimes \ldots \otimes B_m$ is represented by a box \vcen{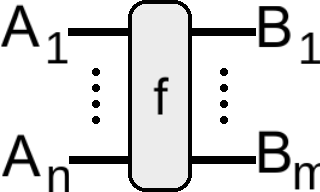} with $n$ input wires on the left and $m$ output wires on the right. Composition of morphisms is shown by connecting boxes in sequence \vcen{pics/S2/seq\_comp}, while the monoidal product corresponds to placing them side-by-side \vcen{pics/S2/par\_comp}. The void resource $I$ corresponds to no wires at all—i.e., the empty diagram.

Resources can also be swapped according to the symmetry $\chi_{A,B} \colon A \otimes B \to B \otimes A$, which is represented by a crossing of wires \vcen{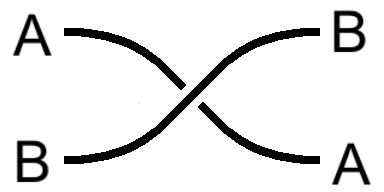}.

Diagrams are subject to the axioms of SMCs, listed in Appendix~\ref{sec:monoidal}. However, note that equational reasoning in this setting amounts to deformation of diagrams without changing their topology. For example, naturality of symmetry amounts to sliding a box past a crossing of wires, illustrated as \vcen{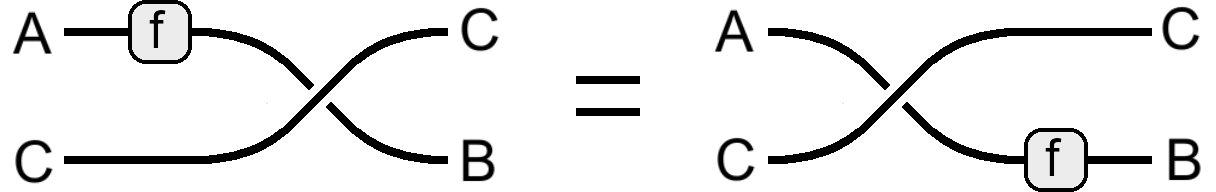}.

In the next section, where we introduce the parametric iteration construction, it will be convenient to explicitly track both the combination of resources via the monoidal product and the void resource. Therefore, following~\cite{carette2019szx}, we extend the graphical calculus with operations that introduce \vcen{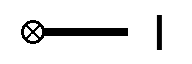} and discard  \vcen{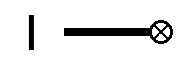} the void resource, as well as split \vcen{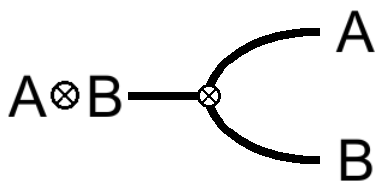} and combine \vcen{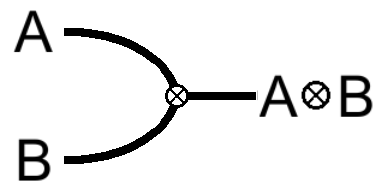} resources.
These operations are subject to the expected equations %
determined by the properties of the monoidal product.

Finally, note that all monoidal categories considered in this paper are taken to be strict, as justified by Mac Lane's strictification theorem~\cite{mclane}.

\section{Parametric Iteration}\label{sec:iteration}

We define a category of parametric iterations over a symmetric monoidal category. We do this by introducing the concept of $\star$-iteration, which expresses the idea of iterating an operation a fixed yet unspecified number of times. We do not a priori choose how many times to iterate, the operations are \emph{parametrized} over this yet to be specified choice of number of iterations.

\newcommand{\FI}[1]{\mathsf{FI}(#1)}

We define parametric iteration over $\C$ in two stages. First we define the \emph{free iteration construction} $\FI{\C}$ which adds the relevant objects and operations which will be used to express iteration. We then quotient this category to create the \emph{parametric iteration construction} $\PI{\C}$ over $\C$, giving an interpretation to the added objects and operations.

We define the objects of $\FI{\C}$ inductively as follows:
\[
\frac{A \in \C}{\br{A} \in \FI{\C}}
\qquad \qquad
\frac{A , B \in \FI{\C}}{A \otimes B \in \FI{\C}}
\qquad \qquad
\frac{A \in \FI{\C}}{A^\star \in \FI{\C}}
\]
satisfying the equations $\br{I} = I, \br{I} \otimes A = A = A \otimes \br{I}$, $\br{A} \otimes \br{B} = \br{A \otimes B}$, and $(A \otimes B) \otimes C = A \otimes (B \otimes C)$.
Here $A^\star$ represents a $\star$-tuple, that is a tuple of $A$'s, with the exact number of elements bounded by an external parameter.

Given a tuple of objects $\ol{A} = (A_1 , \dots , A_n)$, we write $\ol{A} \cdot \star = A_1^\star \otimes \dots \otimes A_n^\star$, and for each $k \in \Nat$, $\ol{A} \cdot k = A_1^k \otimes \dots \otimes A_n^k$, following notation from \cite{soton422947}, where $A^k$ is the $k$-th power of $A$ with respect to $\otimes$. In particular, $\ol{A} \cdot 1 = A_1 \otimes \dots \otimes A_n$.

Morphisms of $\FI{\C}$ are given inductively as follows:
\[
\frac{f : A \to B \in \C}{\br f : \br A \to \br B}
\qquad
\frac{A \in \FI{\C}}{\textit{id}_A : A \to A}
\qquad
\frac{f : A \to B \quad g : B \to C}{f ; g : A \to C}
\qquad
\frac{f : A \to B \quad g : C \to D}{f \otimes g : A \otimes B \to C \otimes D}
\]
\[
\qquad
\frac{A , B \in \FI{\C}}{\chi_{A,B} : A \otimes B \to B \otimes A}
\qquad
\frac{f : S \otimes \ol{A} \cdot 1 \to \ol{B} \cdot 1 \otimes S}{\tau^\star_{S, \ol{A} , \ol{B}}(f) : S \otimes \ol{A} \cdot \star \to \ol{B} \cdot \star \otimes S}
\]
satisfying the relevant equations making $\FI{\C}$ into a symmetric monoidal category,
together with the equations $\br{\textit{id}_A} = \textit{id}_{\br A}$, $\br{f;g} = \br{f};\br{g}$, $\br{\chi_{A,B}} = \chi_{\br{A},\br{B}}$, $\br{f \otimes g} = \br{f} \otimes \br{g}$.
The latter equations make $\br{-} : \C \to \FI{\C}$ sending $A$ to $\br{A}$ and $f$ to $\br{f}$ into a symmetric monoidal functor.

Here $\tau^\star$ is the operation which represents $\star$-iteration, transforming a list of $\star$-tuples $\ol{A} \cdot \star$ into a list of $\star$-tuples $\ol{B} \cdot \star$ by going through the elements one by one, carrying over the state $S$ between iterations.
We annotate $\tau^\star$ by lists telling us exactly how the input and output of $f$ is divided into different tuples. For instance, if $f : S \otimes A^2 \to B \otimes S$, we could either create $\tau^\star_{S,(A^2),(B)}(f) : S \otimes (A^2)^\star \to B^\star \otimes S$ or $\tau^\star_{S,(A,A),(B)}(f) : S \otimes A^\star \otimes A^\star \to B^\star \otimes S$, where $(A^2)^\star$ is a $\star$-tuple of elements of $A^2$, and $A^\star \otimes A^\star$ are two $\star$-tuples of elements of $A$.

String diagrammatically, we denote $\tau^\star_{S,\ol{A},\ol{B}}(f)$ as follows:
\begin{center}
	$
	\vcen{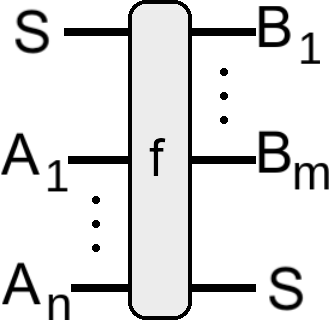} \quad \mapsto \quad \vcen{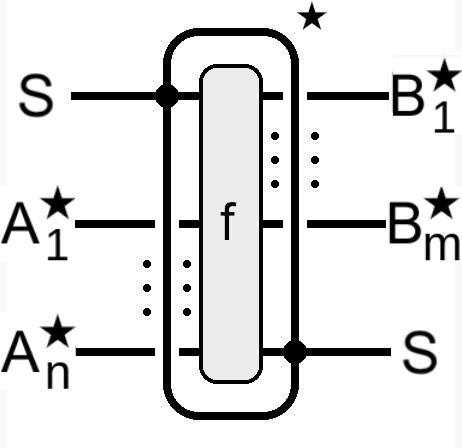}
	$
\end{center}
Here we mark \emph{state wires} with a $\bullet$, and the $\star$-tuples with an interrupted wire. Each non-state wire is part of a different tuple. We may use multiple or no state wires as well.

We give an interpretation of iteration by defining $\tau^k$ as an operation on symmetric monoidal categories describing a constant number of $k$ iterations.
Given a symmetric monoidal category $\C$ , we can inductively define operations $\textsf{push}^k_{\ol{A}} : \ol{A} \cdot 1 \otimes \ol{A} \cdot k \to \ol{A} \cdot (k+1)$ and $\textsf{pop}^k_{\ol{A}} : \ol{A} \cdot (k+1) \to \ol{A} \cdot 1 \otimes \ol{A} \cdot k$ which respectively pushes and pops the first element of each tuple of objects in $\ol{A} \cdot (k+1)$.
For instance, if $\ol{A} = (X , Y)$, then $\textsf{push}^3_\ol{A} : X \otimes Y \otimes X^2 \otimes Y^2 \to X^3 \otimes Y^3$ simply swaps the second and third element.

Given a morphism $f : S \otimes \ol{A} \cdot 1 \to \ol{B} \cdot 1 \otimes S$ in $\C$, we inductively define $\tau^k_{S,\ol{A},\ol{B}}(f) : S \otimes \ol{A} \cdot k \to \ol{B} \cdot k \otimes S$ for each $k \in \Nat$ as:
\begin{itemize}
	\item $\tau^0_{S,\ol{A},\ol{B}}(f) = \textit{id}_S$,
	\item $\tau^{k+1}_{S,\ol{A},\ol{B}}(f) = (\textit{id}_S \otimes \textsf{pop}^k_\ol{A}) ; (f \otimes \textit{id}_{\ol{A} \cdot k}) ; (\textit{id}_{\ol{B} \cdot 1} \otimes \tau^k_{S,\ol{A},\ol{B}}(f)) ; (\textsf{push}^k_\ol{B} \otimes \textit{id}_S)$.
\end{itemize}

\begin{center}
$\vcen{pics/S3/starop_a} \quad \mapsto \quad \vcen{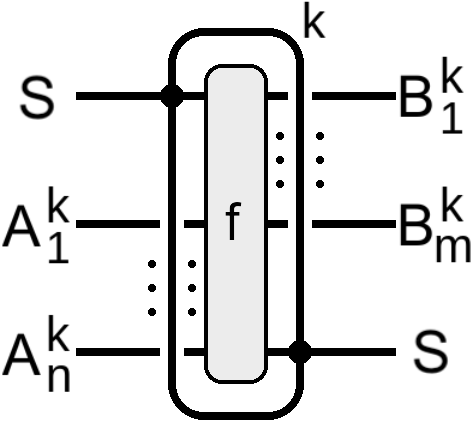}$
\end{center}
\begin{center}
$\vcen{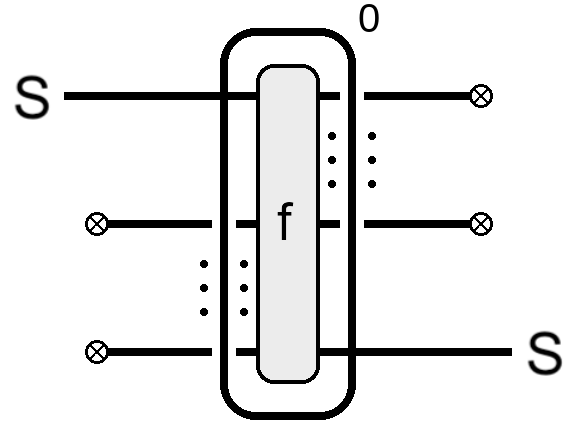} \quad := \quad \vcen{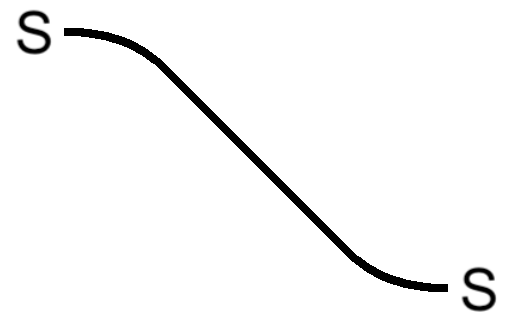}
\qquad \quad
\vcen{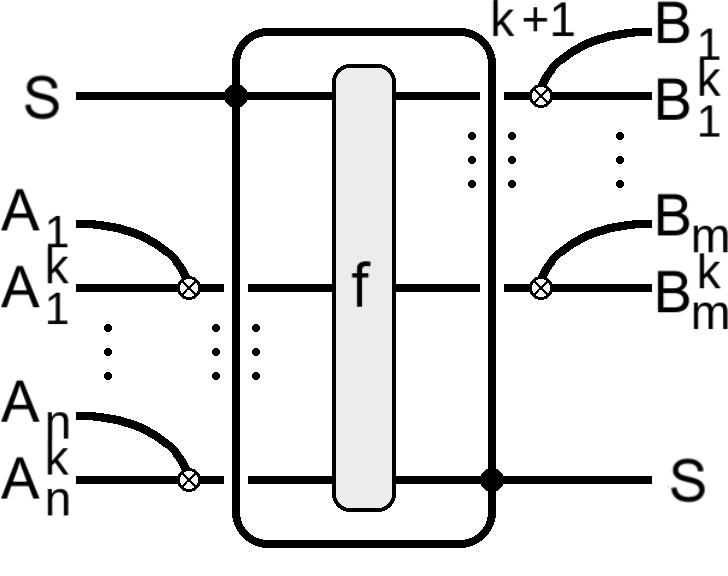} \quad := \quad \vcen{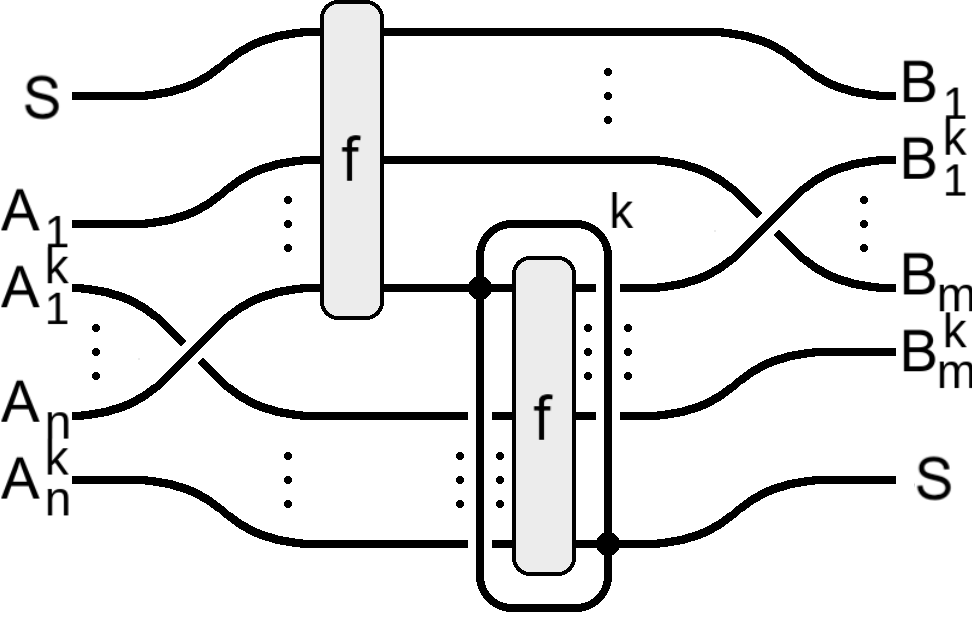}$
\end{center}

Contrary to $\tau^\star$, the annotations for $\tau^k$ serve a further clarifying role, as they are also necessary for removing ambiguity even in those instances where the type of the produced morphism is known. Without this specification, the input may be wrongly distributed into separate tuples. For instance, a $\tau^3$ iteration having $A^6$ as input, could consider this either as two triples of $A$, or one triple of $A^2$, which lead to different interpretations.
In terms of string diagrams, we can keep the convention that each string represents a different tuple, which avoids the need for putting annotations in the string diagrams.
\begin{example}
	Suppose $f : A \otimes A \to B$, then $\tau^2_{1,(A,A),(B)}(f)$ and $\tau^2_{1,(A^2),(B)}(f)$ have the same type but yield two different results, represented respectively as string diagrams as follows:
	\begin{center}
		$\vcen{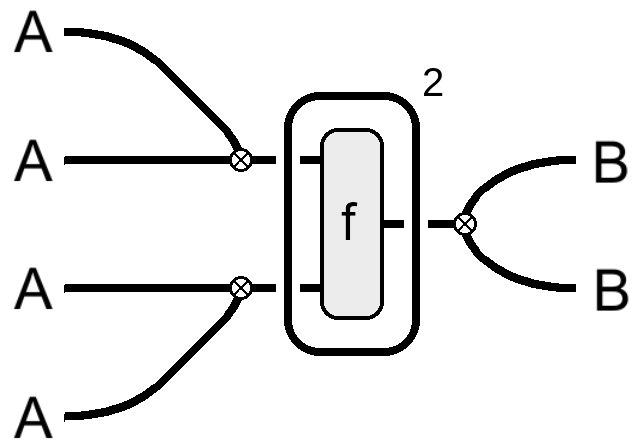} = \vcen{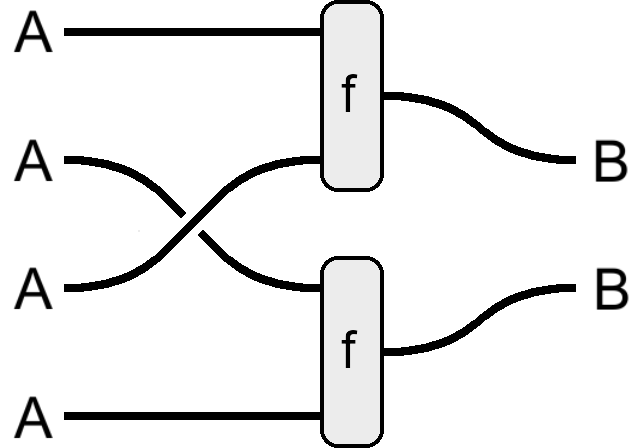}$
		\qquad \qquad
		$\vcen{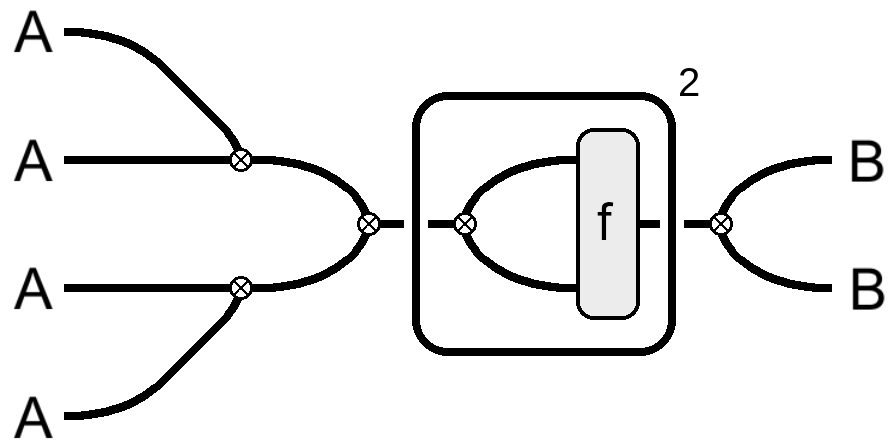} = \vcen{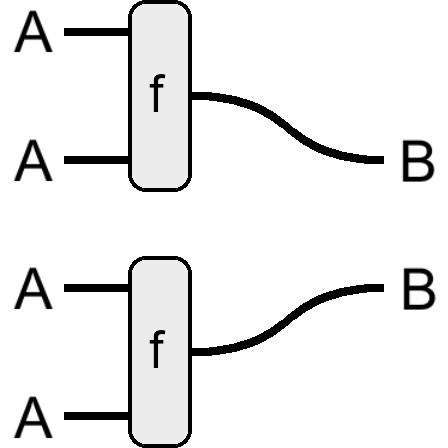}$
	\end{center}
\end{example}
For any $k \in \Nat$, we define a functor $S_k : \FI{\C} \to \C$ inductively as follows, first on objects:
\begin{itemize}
	\item $S_k(\br A) = A$,
	\item $S_k(A \otimes B) = S_k(A) \otimes S_k(B)$,
	\item $S_k(A^\star) = (S_k(A))^k$.
\end{itemize}
For $\ol{A} = (A_1, \dots , A_n)$, then $S_k(\ol{A}) = (S_k(A_1), \dots , S_k(A_n))$. $S_k$ is defined on morphisms as:
\begin{multicols}{2}
\begin{itemize}
	\item $S_k(\textit{id}_A) = \textit{id}_{S_k(A)}$,
	\item $S_k(\br{f}) = f$,
	\item $S_k(f ; g) = S_k(f) ; S_k(g)$,
	\item $S_k(f \otimes g) = S_k(f) \otimes S_k(g)$,
	\item $S_k(\chi_{A,B}) = \chi_{S_k(A) , S_k(B)}$.
\end{itemize}
\end{multicols}
\noindent For $f : S \otimes \ol{A} \cdot 1 \to \ol{B} \cdot 1 \otimes S$, then $S_k(\tau^\star_{S,\ol{A},\ol{B}}(f)) = \tau^k_{S,S_k(\ol{A}),S_k(\ol{B})}(S_k(f))$.

\begin{definition}
	Two morphisms $f,g : A \to B$ in $\FI{\C}$ are $\star$-equivalent if for any $k \in \Nat$, $S_k(f) = S_k(g)$.
	The parametric iteration category $\PI{\C}$ of $\C$ is defined as $\FI{\C}$ quotiented over this notion of $\star$-equivalence.
\end{definition}

We have a symmetric monoidal functor $\br{-} : \C \to \PI{\C}$ and for each $k \in \Nat$ a symmetric monoidal functor $S_k : \PI{\C} \to \C$. We consider the image of $\br{-}$ in $\PI{\C}$ as the constant fragment of $\PI{\C}$. It holds that $S_k \circ \br{-}$ is the identity functor on $\C$ for each $k \in \Nat$.

Before we continue on and prove results regarding $\star$-equivalence, we introduce a powerful proof technique we call the \emph{diagonal method}. Consider $\PI{\PI{\C}}$, which has two nested parametric iteration constructions, the second we label with $\star'$.
We construct a symmetric monoidal functor $M : \PI{\PI{\C}} \to \PI{\C}$ that unifies the two iterations, i.e.\ it satisfies:
\begin{itemize}
	\item $M (A^{\star'}) = M (A^\star) = (M A)^\star$,
	\item $M (\tau^{\star'}_{S,\ol{A},\ol{B}}(f)) = M (\tau^\star_{S,\ol{A},\ol{B}}(f)) = \tau^\star_{M S,M \ol{A},M \ol{B}}(M f)$.
\end{itemize}

\begin{lemma}
	There is a functor $M : \PI{\PI{\C}} \to \PI{\C}$ satisfying the properties above.
\end{lemma}
\begin{proof}
	Let us consider the obvious functor $M' : \FI{\PI{\C}} \to \PI{\C}$.
	We prove that $M'$ preserves $\star$-equivalence. Let $f$ and $g$ be $\star$-equivalent morphisms in $\mathsf{FI}(\PI{\C})$, then for any $k \in \Nat$, $S_k(f) = S_k(g)$ in $\PI{C}$. So for any $r \in \Nat$, $S_r(S_k(f)) = S_r(S_k(g))$ in $\C$. Note that the $S_k$ instantiates the $\star'$, and the $S_r$ the $\star$.

	Now, for any $k \in \Nat$, $S_k(M' f) = S_k(S_k(f)) = S_k(S_k(g)) = S_k(M' g)$, hence $M' f = M' g$. So $M'$ preserves $\star$-equivalence and can be factored through $\PI{\PI{\C}}$ creating $M$.
\end{proof}

In practice, this allows us to prove $\star$-equivalence between morphisms by instantiating some of the iteration operations but not necessarily all. We can do this as long as we can find a separation into two iteration types, marking a choice of top level $\star$-iterations as $\star'$ in a way that type checks.
We now have sufficient tools to prove some useful properties.

\mylabel{def:tid}{\ensuremath{\tau}\text{-id}}
\mylabel{def:tswap}{\ensuremath{\tau}\text{-swap}}
\mylabel{def:tcomp}{\ensuremath{\tau}\text{-comp}}
\mylabel{def:tmono}{\ensuremath{\tau}\text{-mono}}
\begin{lemma}\label{lem:tau-prop}
	The following equations hold:
	\[
	\begin{array}{r@{\;}c@{\;}l@{\qquad}r@{\;}c@{\;}l}
		\vcen{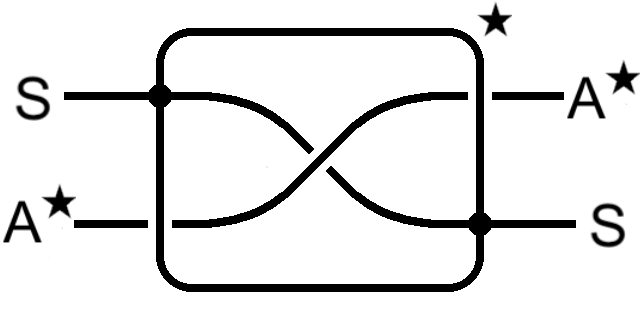}
		&
		\stackrel{(\ref*{def:tid})}{=}
		&
		\vcen{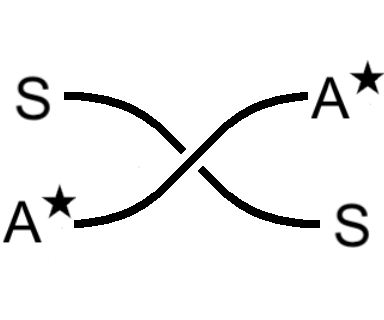}
		&
		\vcen{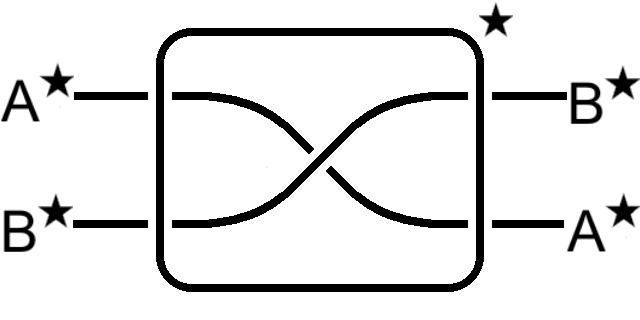}
		&
		\stackrel{(\ref*{def:tswap})}{=}
		&
		\vcen{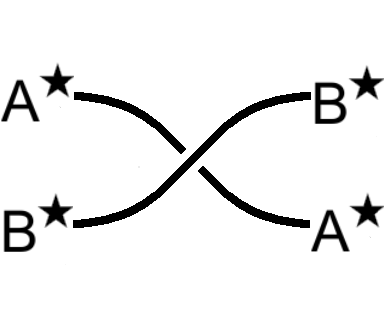}
		\\
		\vcen{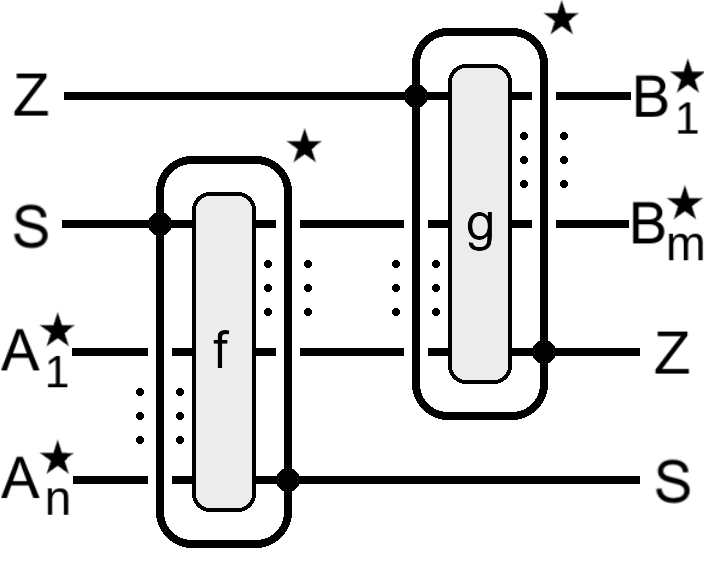}
		&
		\stackrel{(\ref*{def:tcomp})}{=}
		&
		\vcen{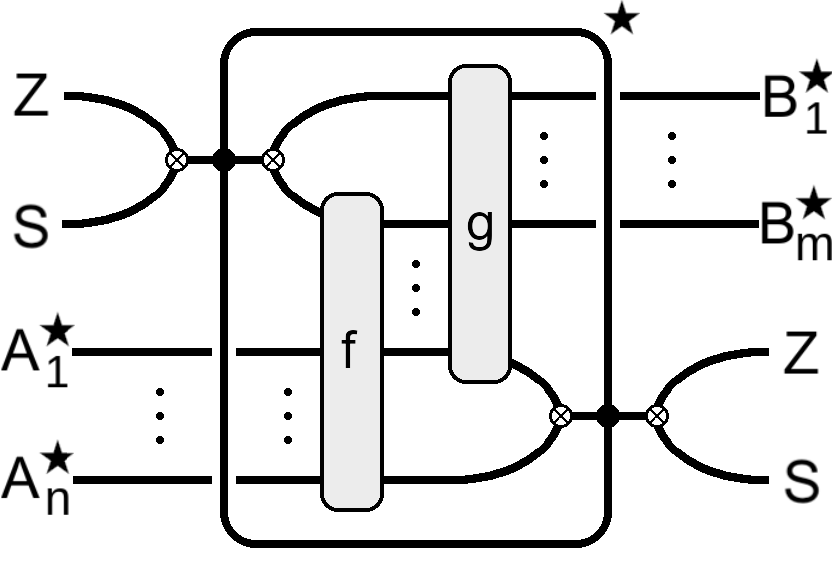}
		&
		\vcen{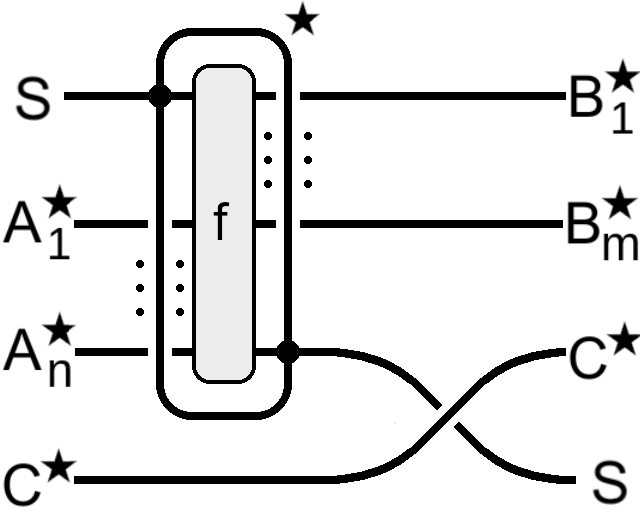}
		&
		\stackrel{(\ref*{def:tmono})}{=}
		&
		\vcen{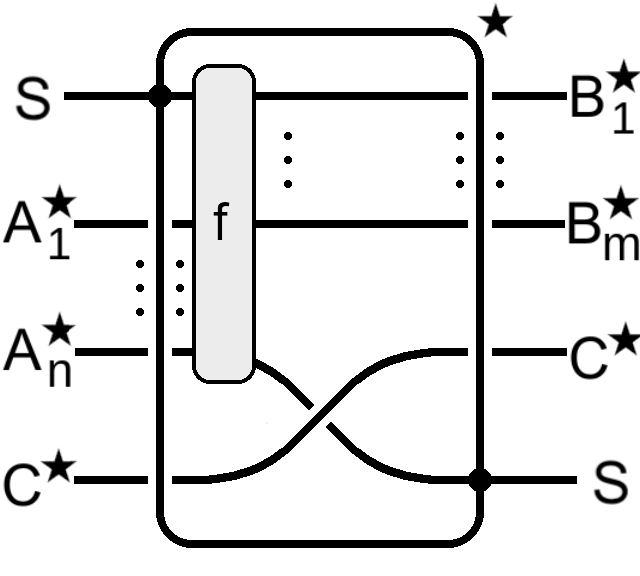}
	\end{array}
	\]
\end{lemma}
\begin{proof}
	Let us prove (\ref{def:tcomp}), the other equations are shown in Appendix \ref{ap:4}. We use the diagonal method, instantiating only the outer iteration operations. We perform induction on those instantiations:

	Induction basis:

	\begin{center}
		$
		\vcen{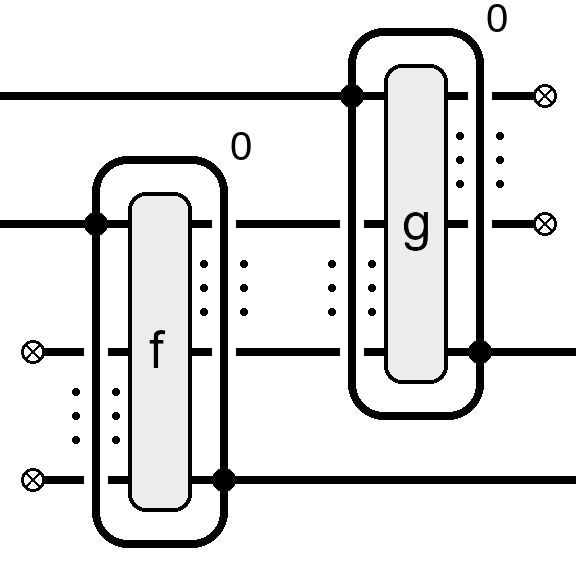}
		\stackrel{(\text{def})}{=}
		\vcen{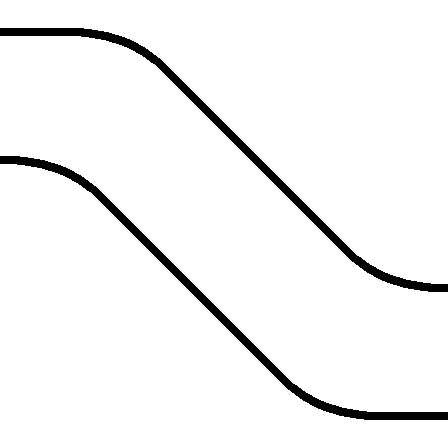}
		\stackrel{(\text{def})}{=}
		\vcen{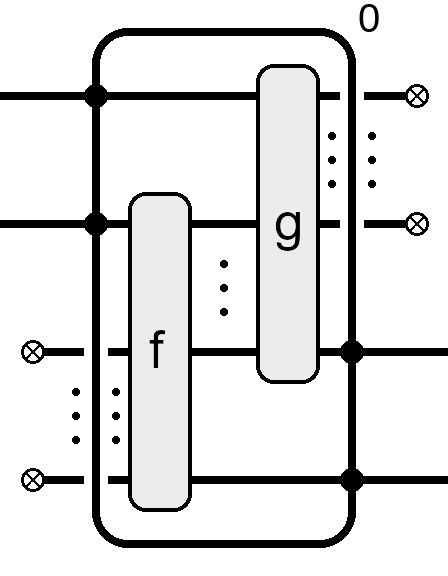}
		$
	\end{center}

	Induction step:

	\begin{center}
		$\
		\vcen{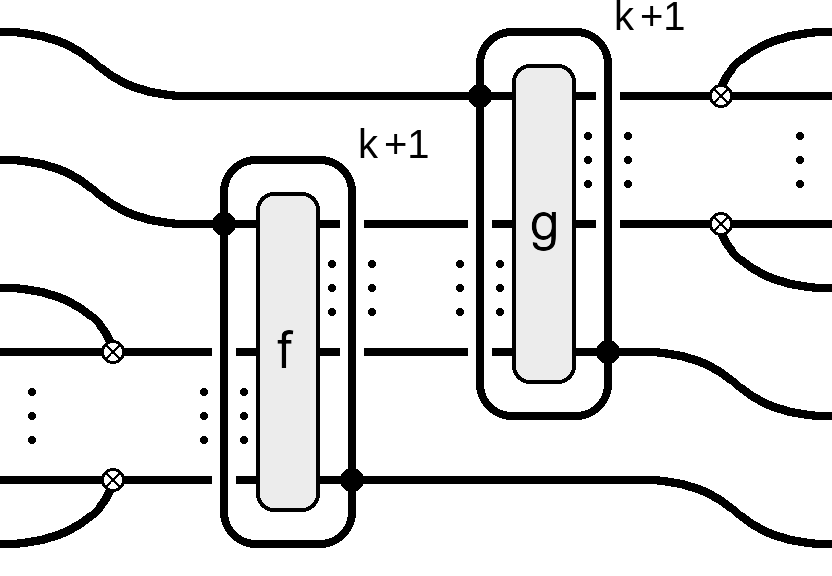}
		\stackrel{(\text{def})}{=}
		\vcen{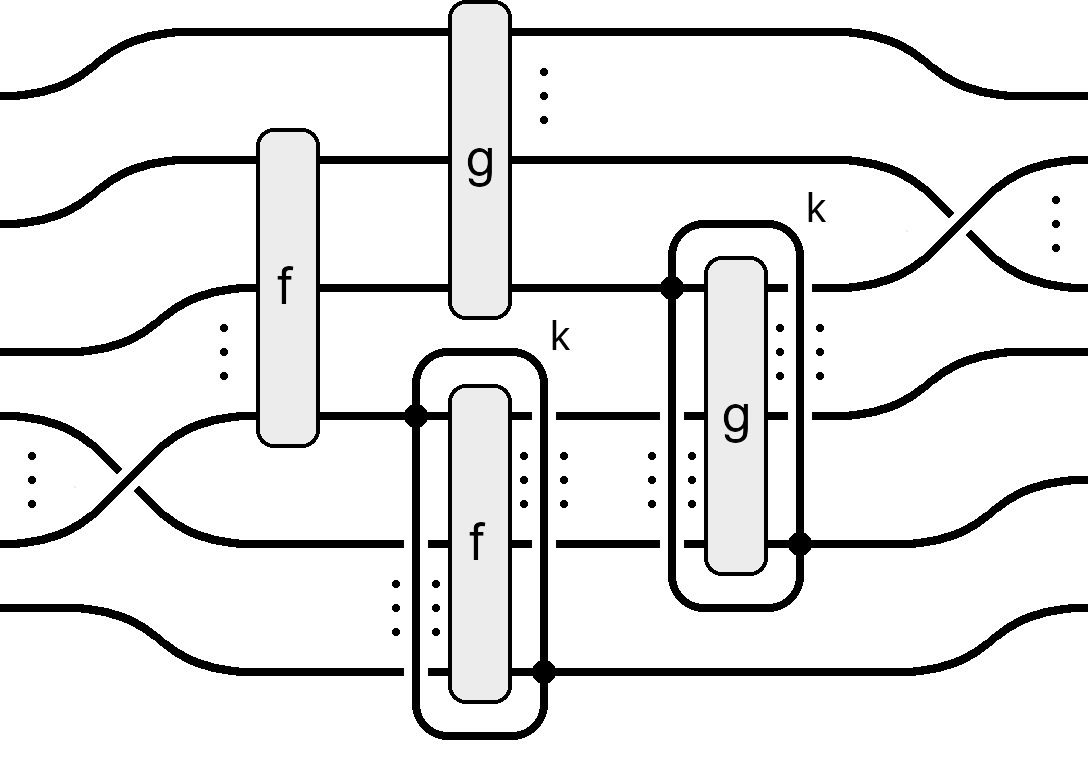}
		\stackrel{(\text{hypo})}{=}
		\vcen{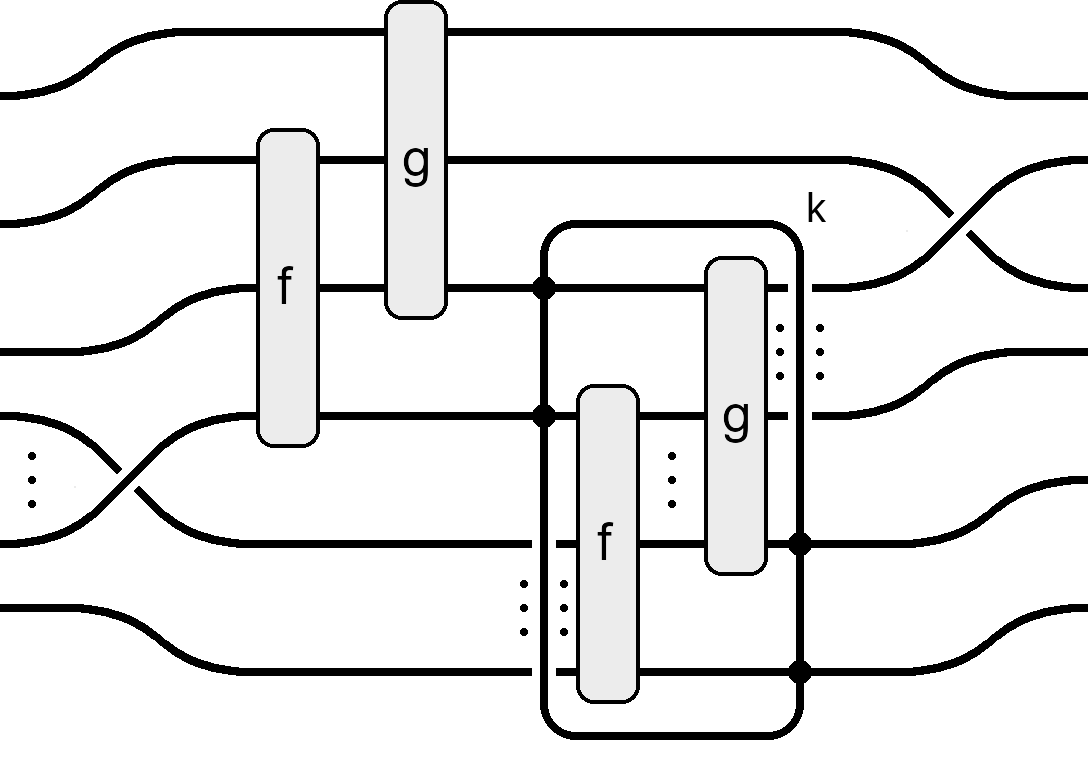}
		\stackrel{(\text{def})}{=}
		\vcen{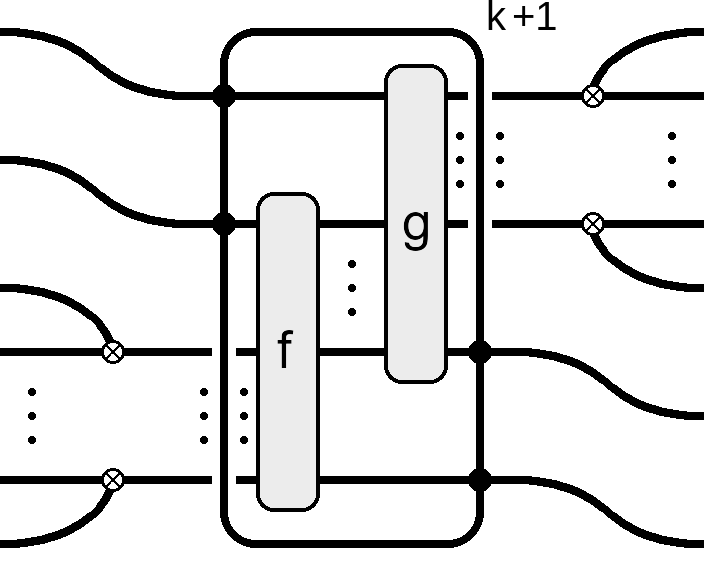}
		$
	\end{center}
\end{proof}

\begin{corollary}
	There is a symmetric monoidal endofunctor on $\PI{\C}$ given by $A \mapsto A^\star$ and $f : A \to B \mapsto \tau^\star_{I,(A),(B)}(f)$.
\end{corollary}

The following property is particularly useful in forthcoming examples:

\begin{lemma}[Newton's cradle]\label{lem:Newton0}
	For any $f \colon S \to Z, g \colon Z \otimes \ol{A} \cdot 1 \to \ol{B} \cdot 1 \otimes Z, h \colon S \otimes \ol{A} \cdot 1 \to \ol{B} \cdot 1 \otimes S$, if $\;\;\vcen{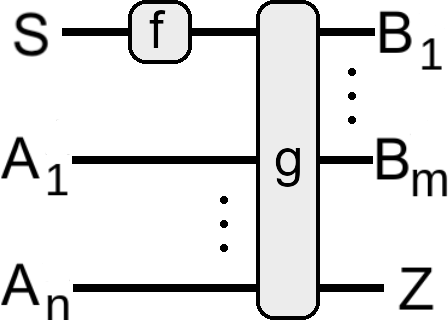} \stackrel{(\text{I})}{=} \vcen{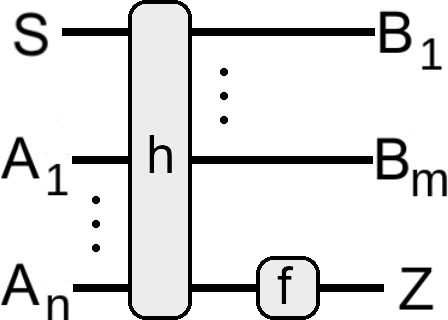}\;\;$ then $\;\;\vcen{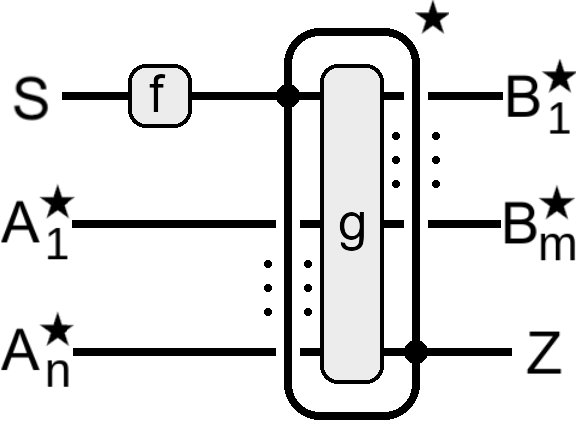} = \vcen{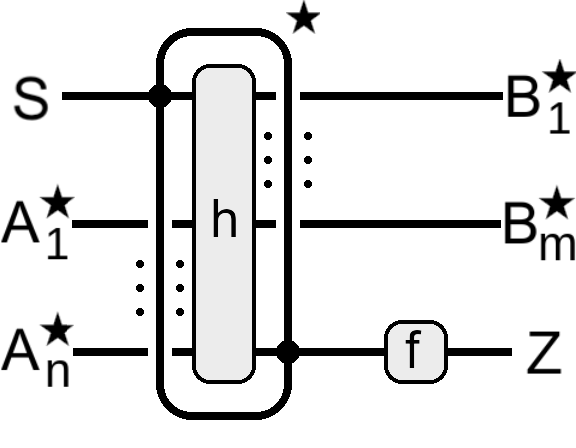}$.
\end{lemma}
\begin{proof}
	We use the diagonal method, and choose to instantiate only the top level iterations.
	We show by induction on $k \in \Nat$ that for each instantiation by $k$ the equation holds.

	Induction basis:

	\begin{center}
		$
		\vcen{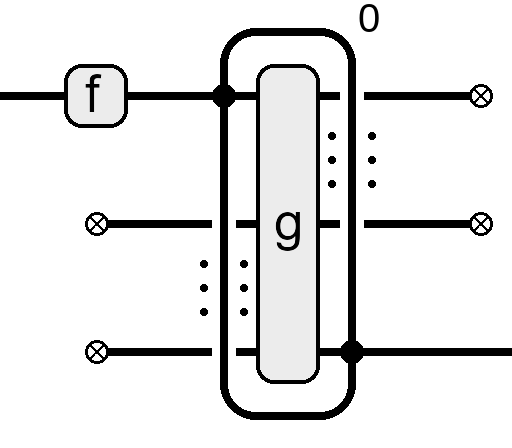}
		\stackrel{(\text{def})}{=}
		\vcen{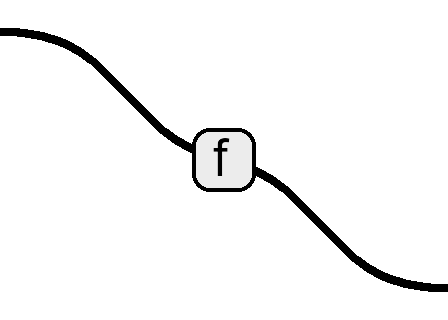}
		\stackrel{(\text{def})}{=}
		\vcen{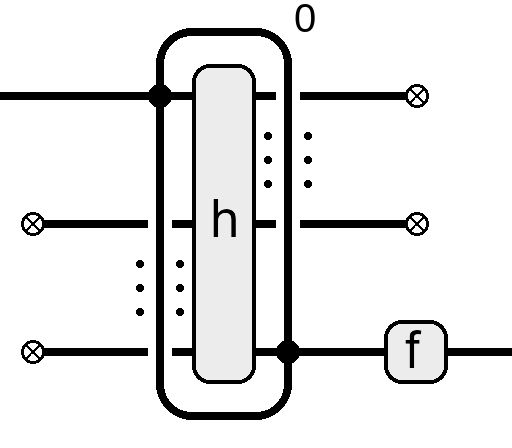}
		$
	\end{center}

	Induction step:

	\begin{center}
		$\
		\vcen{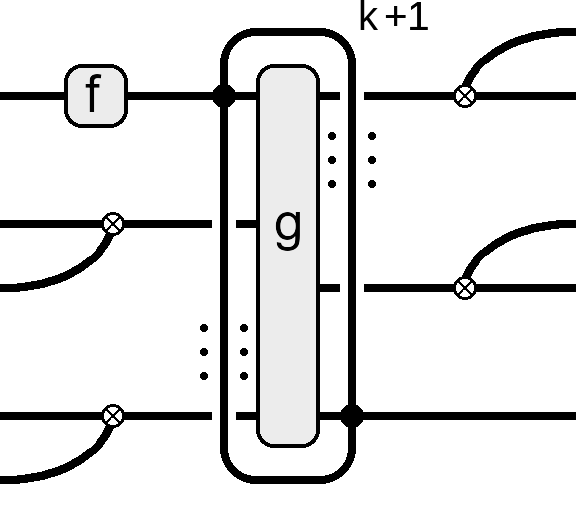}
		\stackrel{(\text{def})}{=}
		\vcen{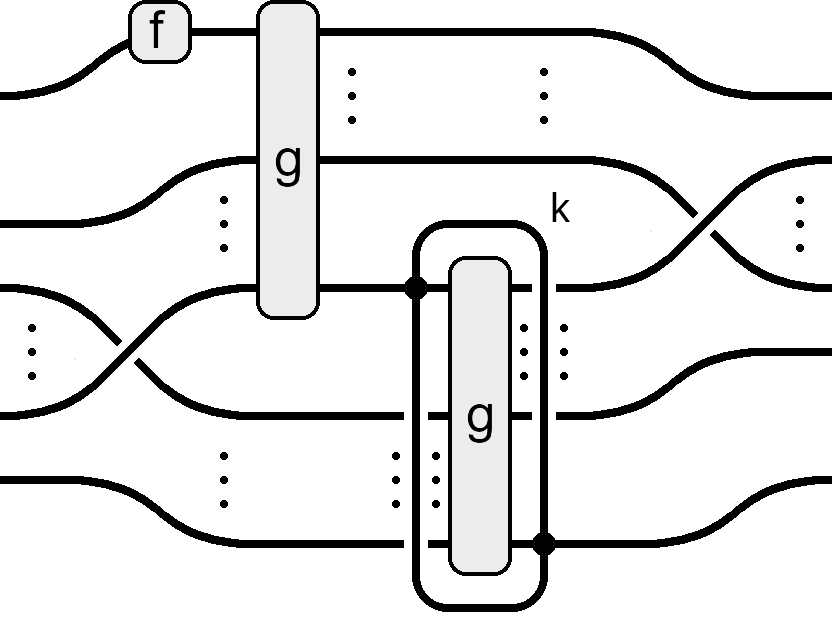}
		\stackrel{(\text{I})}{=}
		\vcen{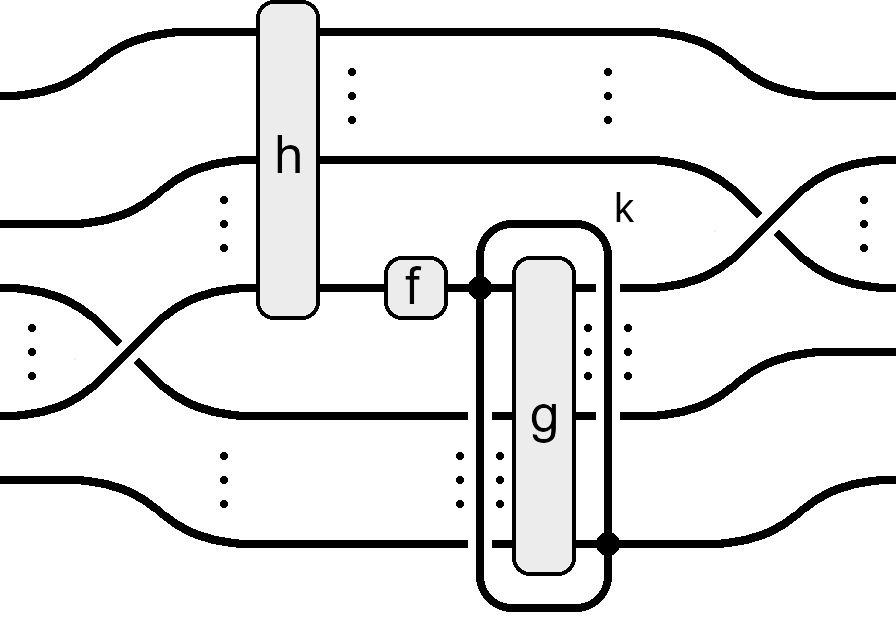}
		\stackrel{(\text{hypo})}{=}
		\vcen{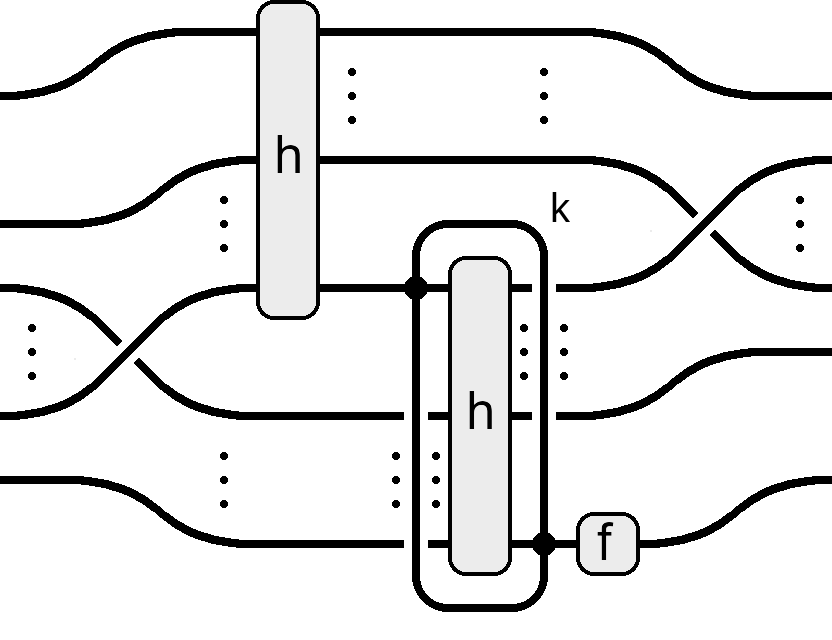}
		\stackrel{(\text{def})}{=}
		\vcen{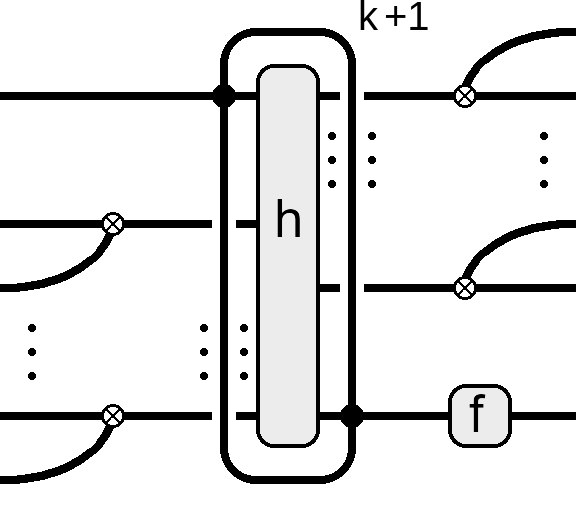}
		$
	\end{center}

\end{proof}

Let us look at two examples which occur in $\PI{\C}$ for any symmetric monoidal category $\C$.

\begin{example}
	Consider \emph{zipping} two $\star$-tuples using $\tau_{I,(A,B),(A \otimes B)}(\textit{id}_{A \otimes B})$:

	\begin{center}
		\vcen{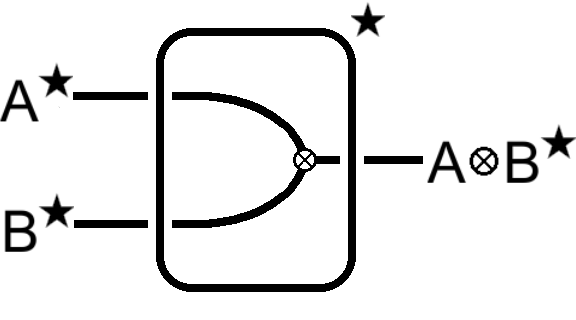}
		\qquad \qquad
		\vcen{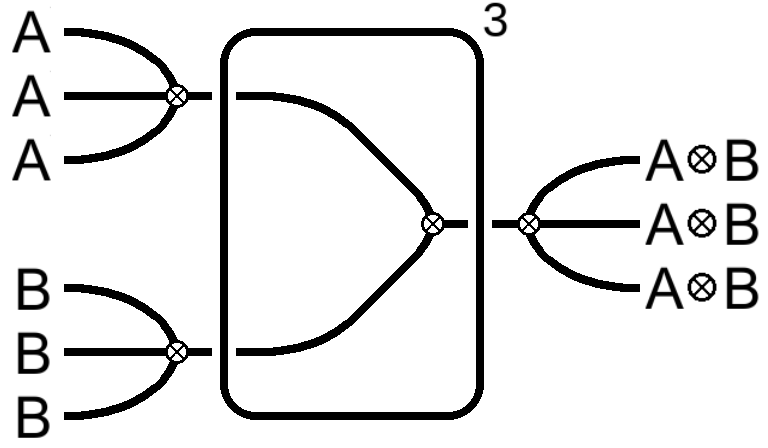} = \vcen{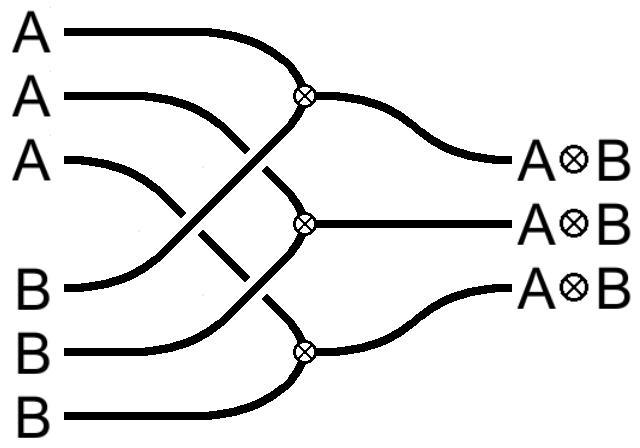}
	\end{center}

	On the right we see an example instantiation performing three iterations.
	The zip operations gives one direction of an isomorphism between $A^\star \otimes B^\star$ and $(A \otimes B)^\star$.
\end{example}

\begin{example}
	The morphism $\tau^\star_{A,(A),(A)}(\textit{id}_{A^2}) : A^\star \otimes A \to A^\star \otimes A$ pushes a new element to the front of a $\star$-tuple, and pops the last element from the $\star$-tuple to make room.

	\begin{center}
		$\tau^\star_{A,(A),(A)}(\textit{id}_{A^2}) = \vcen{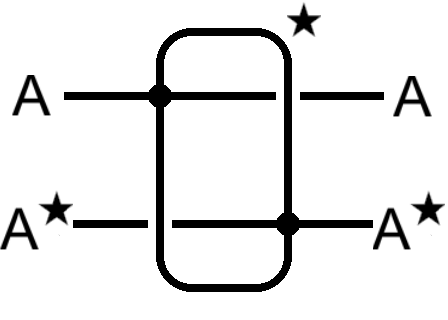}
		\qquad \qquad
		\tau^k_{A,(A),(A)}(\textit{id}_{A^2}) =
		\vcen{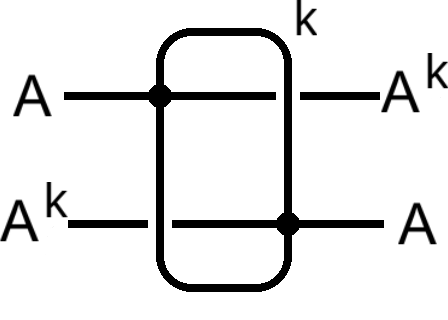} =
		\vcen{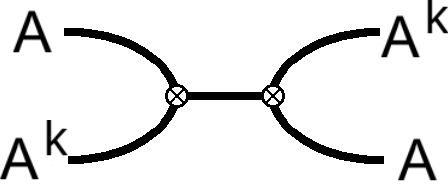}$
	\end{center}

	We can use this to cycle forward and backward through the elements of a $\star$-tuple.
	We take $\textsf{cycle}_A = \tau^\star_{A,(A),(A)}(\textit{id}_{A^2}) ; \chi_{A^\star , A}$ and $\textsf{cycle-back}_A = \tau^\star_{A \otimes A^\star , () , ()}(\textsf{cycle}_A)$,

	\begin{center}
		$\textsf{cycle}_A  = \vcen{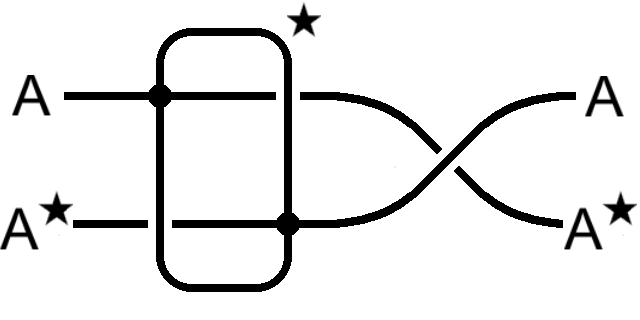}
		\qquad \qquad
		\textsf{cycle-back}_A  =
		\vcen{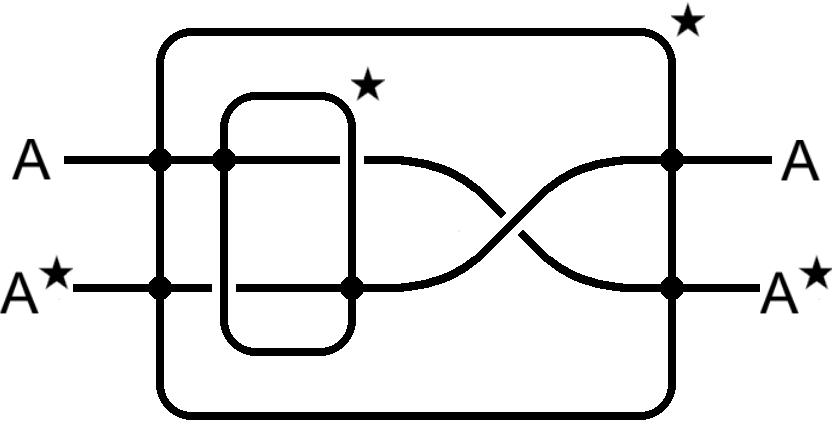}$
	\end{center}

	Here we cycle back by cycling forward one less than what is needed to get back to where we started.
	Concretely, note that $S_k(\textsf{cycle}_A) = \chi_{A^k,A}$, and for any $a , b \in \Nat$, $\chi_{A^{a+1},A^b} ; \chi_{A^{a+b},A} = \chi_{A^b,A^{b+1}}$,
	so by induction, $S_k(\textsf{cycle-back}_A) = (\chi_{A^k,A})^k = \chi_{A,A^k}$, so $S_k(\textsf{cycle-back}_A ; \textsf{cycle}_A) = \chi_{A,A^k} ; \chi_{A^k,A} = \textit{id}_{A^{k+1}}$.
	We conclude that:

	\begin{center}
		$\vcen{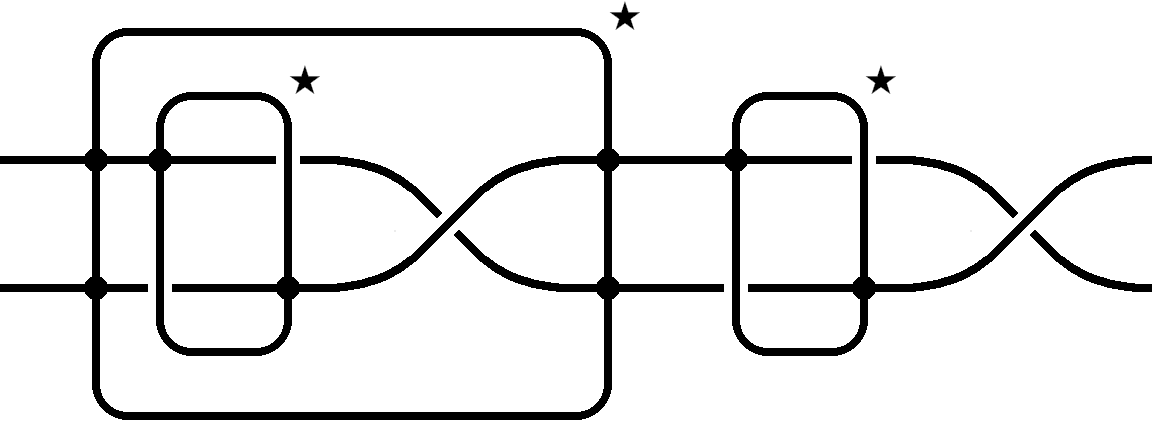} = \vcen{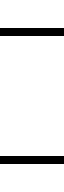}$
	\end{center}
\end{example}

\section{Reasoning with Probabilities and Distances}\label{sec:preliminaries}

To model richer computational behaviors, we consider symmetric monoidal categories with additional properties and structure. In particular, for probabilistic computations, we focus on metric-enriched symmetric monoidal categories and Markov categories~\cite{fritz_2020}.

\begin{figure}
    \mylabel{ax:copyas}{\ensuremath{\copier{}}\text{-as}}
    \mylabel{ax:copyun}{\ensuremath{\copier{}}\text{-un}}
    \mylabel{ax:copyco}{\ensuremath{\copier{}}\text{-co}}
    \mylabel{ax:discardnat}{\ensuremath{\discharger{}}\text{-nat}}

    \mylabel{ax:onedet}{1\text{-det}}
    \mylabel{ax:zerodet}{0\text{-det}}

    \mylabel{ax:phiI}{\ensuremath{\phi_I}}
    \mylabel{ax:phiTimes}{\ensuremath{\phi_\otimes}}
    \mylabel{ax:phiOne}{\ensuremath{\phi}\text{-1}}
    \mylabel{ax:phico}{\ensuremath{\phi}\text{-co}}
    \mylabel{ax:phias}{\ensuremath{\phi}\text{-as}}
    \mylabel{ax:phinat}{\ensuremath{\phi}\text{-nat}}
    \mylabel{ax:phiBool}{\ensuremath{\phi\text{-}\Bool}}
    \mylabel{ax:phiIdemp}{\ensuremath{\phi}\text{-id}}
    \[
        \begin{array}{c@{\quad}c@{\quad}c}
            \vcen{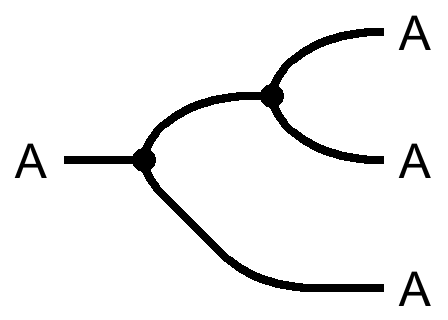} \stackrel{(\ref*{ax:copyas})}{=} \vcen{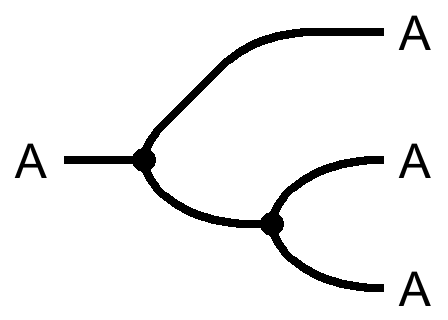}
            &
            \vcen{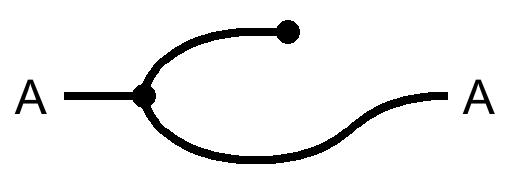} \stackrel{(\ref*{ax:copyun})}{=} \vcen{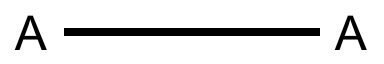}
            &
            \vcen{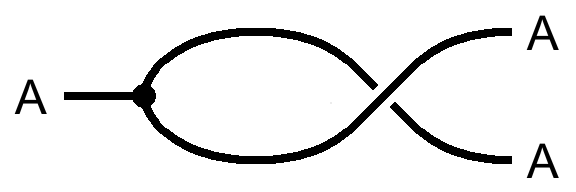} \stackrel{(\ref*{ax:copyco})}{=} \vcen{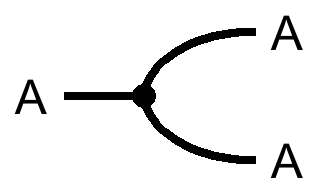}
            \\
            &
            \vcen{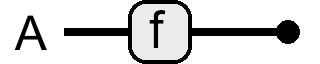} \stackrel{(\ref*{ax:discardnat})}{=} \vcen{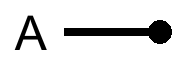}
            &
            \\
            \midrule
            \multicolumn{3}{c}{\begin{array}{c@{\quad}c@{\quad}c@{\quad}c}
                \vcen{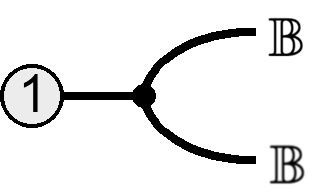} \stackrel{(\ref*{ax:onedet})}{=} \vcen{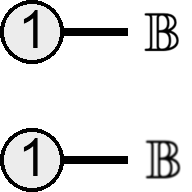}
                &
                \vcen{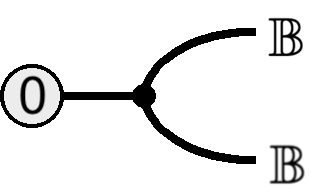} \stackrel{(\ref*{ax:zerodet})}{=} \vcen{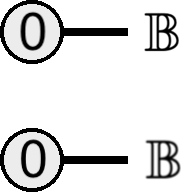}
                &
                \vcen{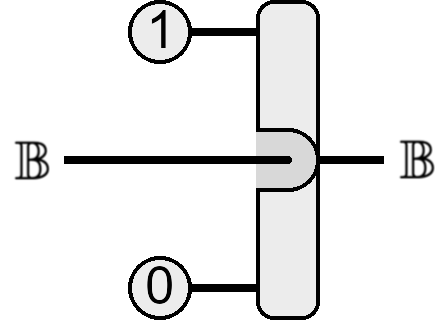} \stackrel{(\ref*{ax:phiBool})}{=} \vcen{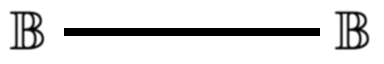}
                &
                \vcen{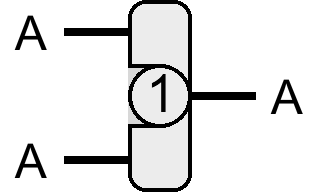} \stackrel{(\ref*{ax:phiOne})}{=} \vcen{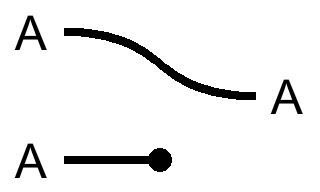}
            \end{array}}
            \\[15pt]
            \vcen{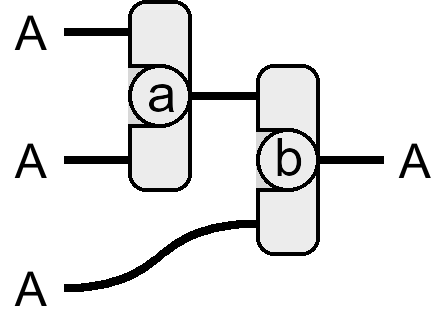} \stackrel{(\ref*{ax:phias})}{=} \vcen{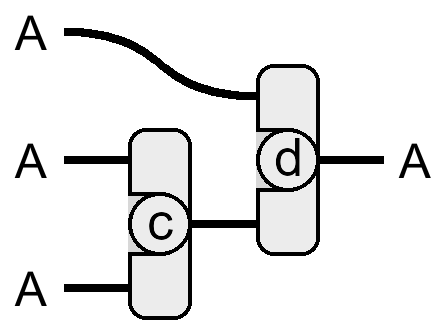}
            &
            \vcen{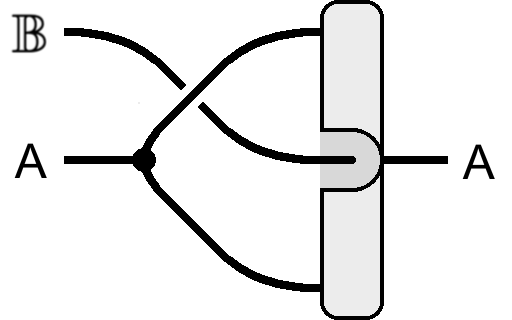} \stackrel{(\ref*{ax:phiIdemp})}{=} \vcen{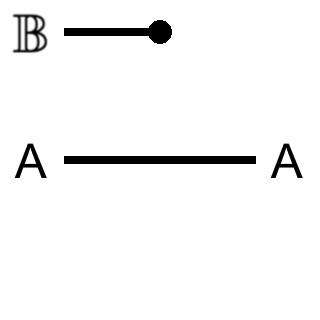}
            &
            \vcen{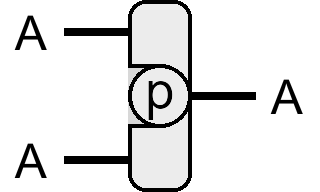} \stackrel{(\ref*{ax:phico})}{=} \vcen{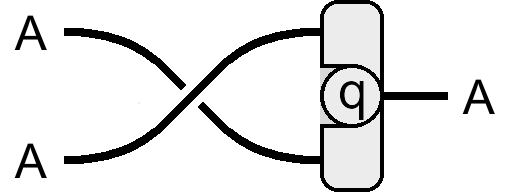}
            \\[15pt]
            \vcen{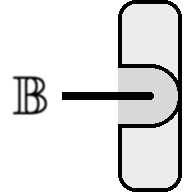} \stackrel{(\ref*{ax:phiI})}{=} \vcen{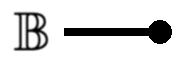}
            &
            \vcen{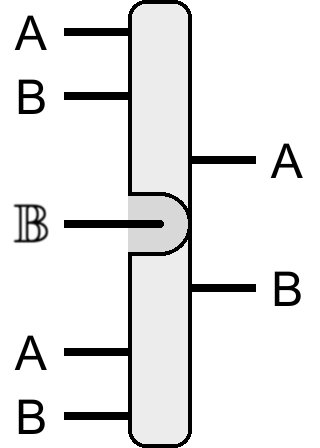} \stackrel{(\ref*{ax:phiTimes})}{=} \vcen{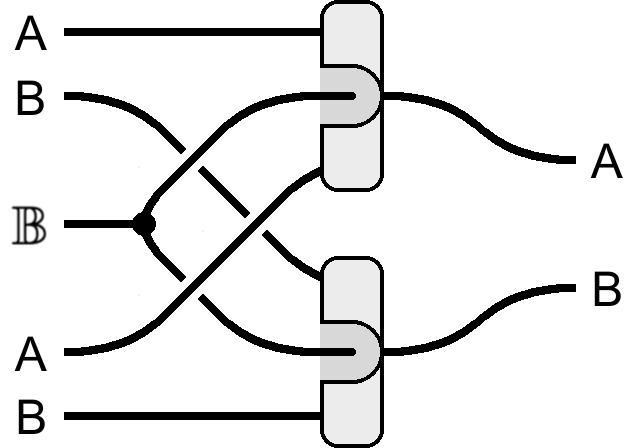}
            &
            \vcen{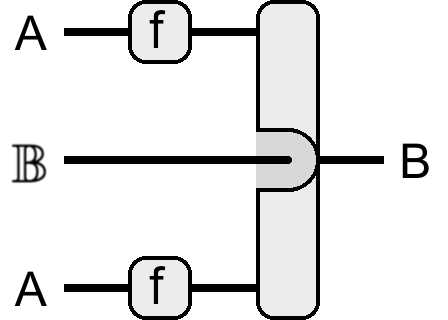} \stackrel{(\ref*{ax:phinat})}{=} \vcen{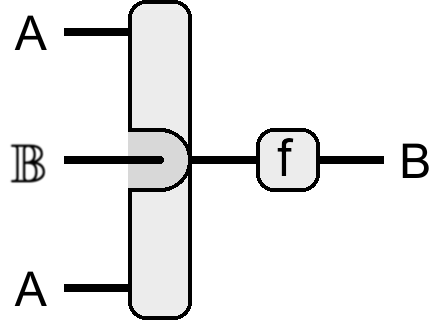}
        \end{array}
    \]
    \caption{Axioms for Markov categories (top) and probabilistic choice (bottom). Here $f \colon A \to B$, $q = 1-p$, $c = \frac{(1-a)b}{1 - ab}$ and $d = ab$.}
    \label{fig:ax markov and p choice}
\end{figure}

\begin{definition}\label{def:markov cat with p choice}
    A \emph{Markov category with probabilistic choice} is a symmetric monoidal category equipped with operations $\discharger{A} \colon A \to I$, $\copier{A} \colon A \to A \otimes A$, $\langle p \rangle \colon I \to \Bool$ and $\phi_A \colon A \otimes \Bool \otimes A \to A$, represented in string diagrams, respectively, as
    \begin{center}
        \vcen{pics/S4/discard}
        \qquad\qquad
        \vcen{pics/S4/copy}
        \qquad\qquad
        \vcen{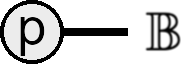}
        \qquad\qquad
        \vcen{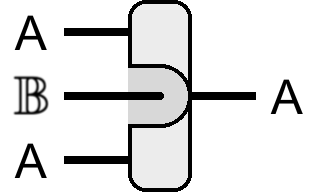}
    \end{center}
    where $p \in [0,1]$, $\Bool$ is a designated object and $A$ is any object, and such that the axioms in Figure~\ref{fig:ax markov and p choice} hold. That is:
    \begin{itemize}
        \item $(\copier{A} , \discharger{A})$ is a cocommutative comonoid, that is $\copier{A}$ is associative~\eqref{ax:copyas}, cocommutative~\eqref{ax:copyco} and $!_A$ is the unit~\eqref{ax:copyun},
        \item $\copier{A}$ and $\discharger{A}$ are coherent with the monoidal structure,
		\item $\discharger{A}$ is natural \eqref{ax:discardnat},
		\item $\langle 1 \rangle$ and $\langle 0 \rangle$ are deterministic (\eqref{ax:onedet} and \eqref{ax:zerodet}),
        \item $\phi_A$ is coherent with the monoidal structure (\eqref{ax:phiI} and \eqref{ax:phiTimes}),
        \item $\phi_A$ is natural in $A$ \eqref{ax:phinat},
        \item $\phi_A^p$--where $\phi_A^p$ is $\vcen{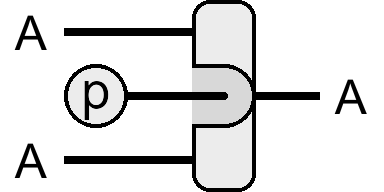}$, often simply $\vcen{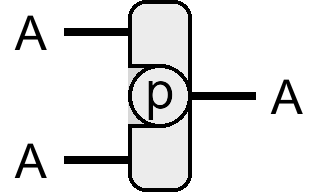}$--is a convex algebra, that is, $\phi^p_A$ is idempotent \eqref{ax:phiIdemp} and satisfies parametric associativity \eqref{ax:phias} and commutativity \eqref{ax:phico},
        \item $\phi$ and $\langle 1 \rangle$, $\langle 0 \rangle$ interact via~\eqref{ax:phiBool} and \eqref{ax:phiOne}.
    \end{itemize}
\end{definition}

The $\Bool$ represents the type of two bit values, $1$ and $0$. $\langle p \rangle$ describes a weighted coin flip, which has a chance of $p$ of producing $1$, and a chance of $1-p$ of producing $0$. The $\phi$ represents an \emph{if-gate}, choosing between two results of the same type using a bit as an argument.

\begin{figure}
    \mylabel{ax:phiZero}{\ensuremath{\phi}\text{-0}}
    \mylabel{ax:phidet}{\ensuremath{\phi}\text{-det}}
    \mylabel{ax:boolsplit}{\ensuremath{\Bool}\text{-split}}
    \mylabel{ax:phidistr}{\ensuremath{\phi}\text{-distr}}
    \[
    \begin{array}{c@{\qquad\qquad}c}
        \vcen{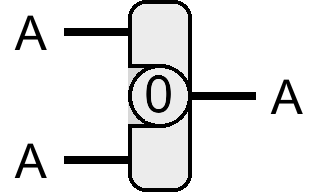} \stackrel{(\ref*{ax:phiZero})}{=} \vcen{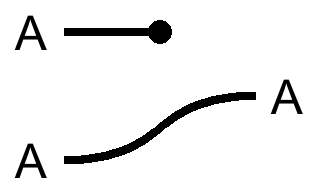}
        &
        \vcen{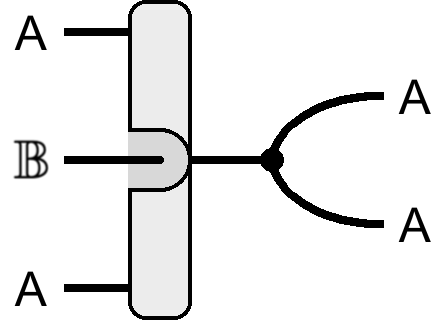} \stackrel{(\ref*{ax:phidet})}{=} \vcen{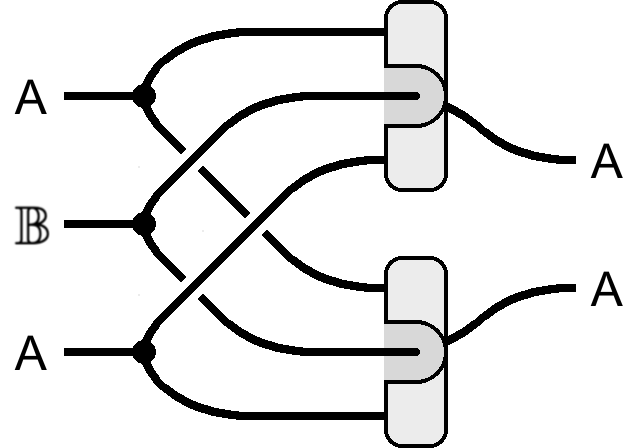}
        \\
        \vcen{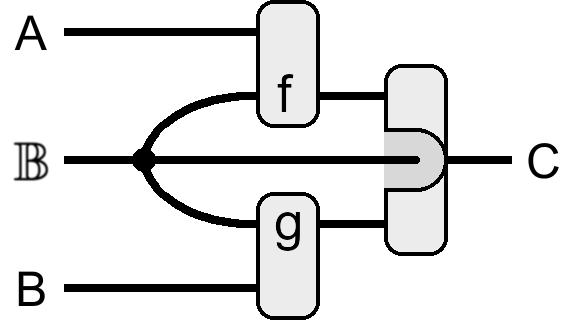} \stackrel{(\ref*{ax:boolsplit})}{=} \vcen{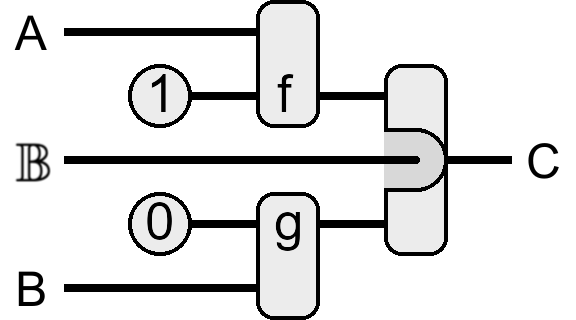}
        &
        \vcen{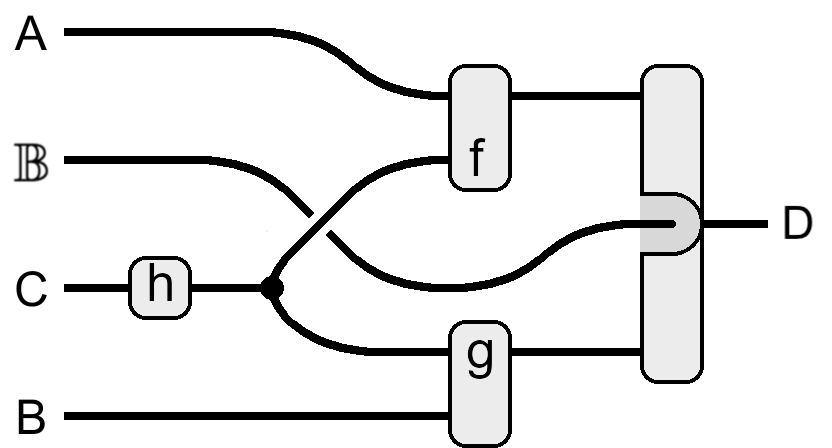} \stackrel{(\ref*{ax:phidistr})}{=} \vcen{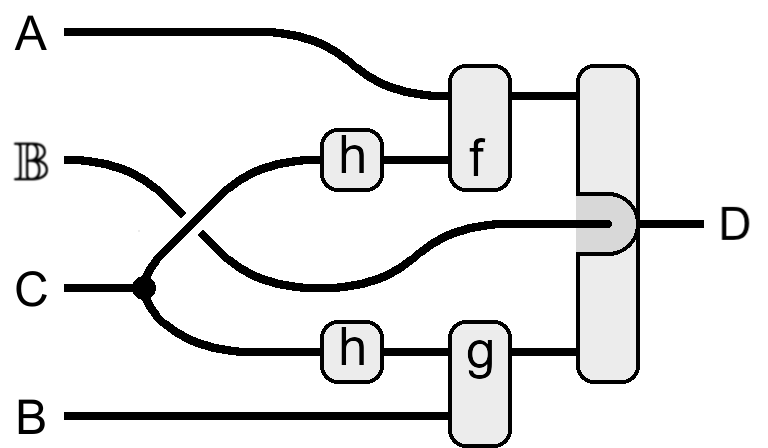}
    \end{array}
    \]
    \caption{Derived laws for Markov categories with probabilistic choice.}
    \label{fig:derived laws markov p choice}
\end{figure}

\begin{lemma}\label{lem:split}
    The laws in Figure~\ref{fig:derived laws markov p choice} hold in any Markov category with probabilistic choice.
\end{lemma}
\begin{proof}
    We give a proof only for~\eqref{ax:phidet}, the rest can be found in Appendix~\ref{ap:4}.
    \[
        \vcen{pics/S4/dlaws/phidet/lhs}
        \stackrel{(\ref*{ax:phinat})}{=}
        \vcen{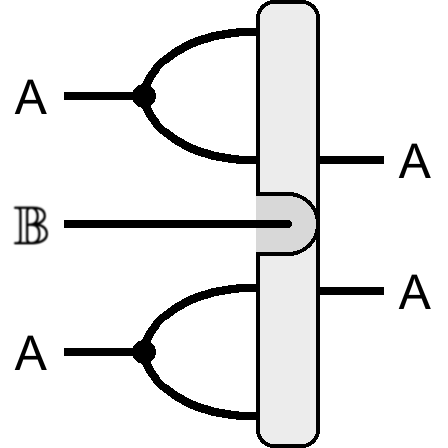}
        \stackrel{(\ref*{ax:phiTimes})}{=}
        \vcen{pics/S4/dlaws/phidet/rhs}
    \]
\end{proof}
    Our leading example of Markov category with probabilistic choice is $\FStoch$, the category of sets and stochastic maps. Objects are sets and morphisms $f \colon A \to B$ are functions $A \to \Dist B$ mapping each $a \in A$ to a finitely-supported discrete probability distribution over $B$. We represent a distribution as a function $v : A \to [0,1]$ with finite support $\texttt{supp}(v) := \{a \in A \mid v(a) > 0\}$ and such that $\sum_{a \in A} v(a) = 1$, or, alternatively, as a formal sum  $\sum_i{p_i \ket{a_i}}$, where the $\ket{\cdot}$ notation is used to separate the elements $a_i \in A$ from their associated probability $p_i \in [0,1]$.

    We write $f(b | a)$ for $f(a)(b)$, representing the probability that $f$ returns $b$, given $a$. For two stochastic maps $f \colon A \to B, g \colon B \to C$, their composition is defined by summing over the middle variable, i.e. $(f;g)(c \mid a) := \sum_{b \in B} g(c \mid b)f(b \mid a)$. Identities $id_A \colon A \to A$ map each $a \in A$ to the Dirac distribution at $a$, i.e. $\id{A}(a) := \ket{a}$.

    Given two distributions $v \in \Dist A$ and $u \in \Dist B$, their product $v \otimes w \in \Dist{(A \times B)}$ is given by $v \otimes u(a,b) := v(a)u(b)$.
	This yields a symmetric monoidal structure, with the monoidal product $\otimes$ being the cartesian product of sets $\times$ on objects and with monoidal unit the singleton set $I := \{ \bullet \}$. For arrows $f \colon A \to B$ and $g \colon C \to D$, $f \otimes g \colon A \otimes C \to B \otimes D$ is defined as $f \otimes g (b,d | a,c) := f(b | a)g(d|c)$, and symmetries $\chi_{A,B} \colon A \otimes B \to B \otimes A$ as $\chi_{A,B}(a,b) := \ket{b} \otimes \ket{a}$.

    The comonoid structure is given by $\copier{A}(a) := \ket{a} \otimes \ket{a}$ and $\discharger{A}(a) := \ket{\bullet}$, while for the probabilistic choice structure we take as choice object the set of Booleans $\{1,0\}$ and, for any $p \in [0,1]$, we define $\langle p \rangle \colon \{ \bullet \} \to \{1,0\}$ as $\langle p \rangle(\bullet) := p\ket{1} + (1-p)\ket{0}$ and $\phi_A \colon A \otimes \{1,0\} \otimes A \to A$ as $\phi_A(a,1,b) := \ket{a}$ and $\phi_A(a,0,b) := \ket{b}$.
			Note that $\phi^p_A(a,b) = p \ket{a} + (1-p) \ket{b}$.

    Throughout the rest of the paper we will work in $\FStoch_{\{1,0\}}$, the subcategory of $\FStoch$ consisting of stochastic maps between sets that are powers of the booleans.
    The structure of Markov category with probabilistic choice, introduced here, is novel and it offers an alternative presentation of $\FStoch_{\{1,0\}}$ to the one in~\cite{piedeleu2025boolean}. In particular: 
   \begin{theorem}\label{thm:complete}
   the axioms in Figure~\ref{fig:ax markov and p choice} are \emph{complete} for $\FStoch_{\{1,0\}}$.
   \end{theorem} 
	The interested reader can find the details in Appendix~\ref{app:completeness}.

As an example of diagrammatic reasoning in $\FStoch_{\{1,0\}}$ we look at the \emph{one-time-pad} protocol on a single bit. We regard $\langle \frac{1}{2} \rangle \colon I \to \Bool$ as an operation that produces a randomly generated bit, and the xor $\oplus \colon \Bool \otimes \Bool \to \Bool$, defined below on the left, as an operation that encrypts and decrypts a bit.
\begin{center}
    \vcen{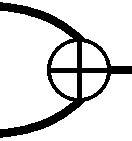} := \vcen{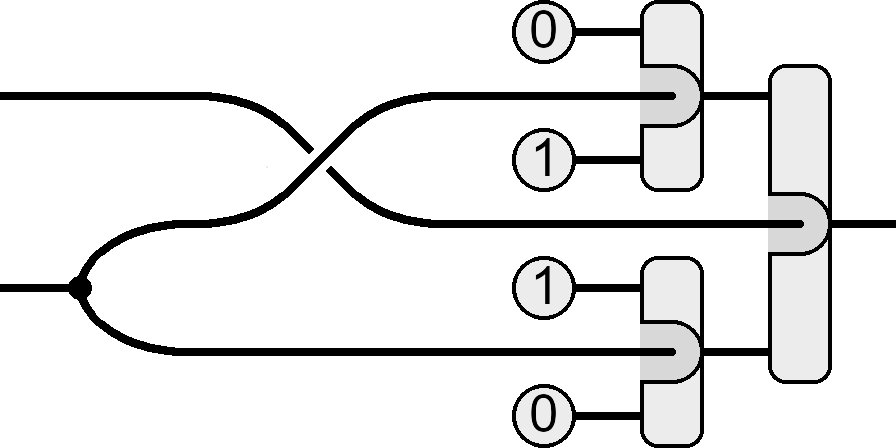}
    \qquad \qquad
    \vcen{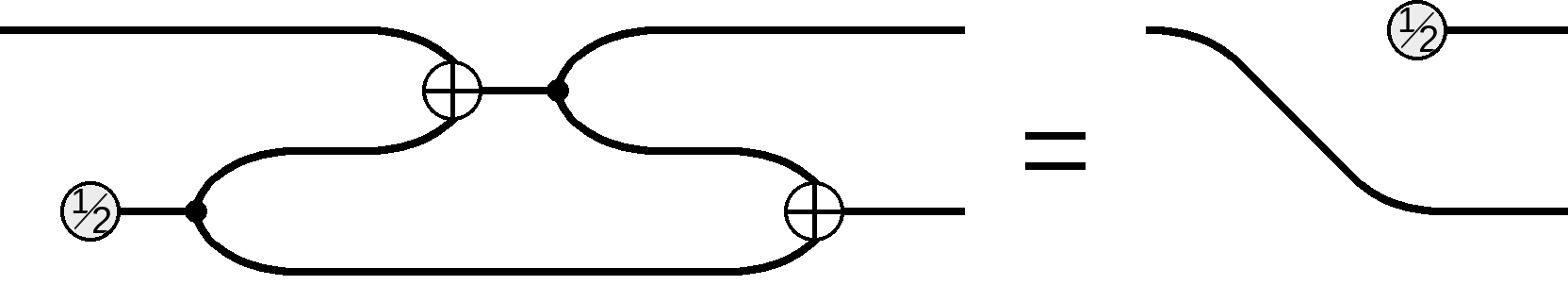}
    \end{center}
The equation on the right expresses the correctness of the protocol: encoding and decoding with the same randomly generated bit gets the same result, and that to an adversary the encrypted message looks completely random. We prove correcteness by equational reasoning:
\begin{center}
    \vcen{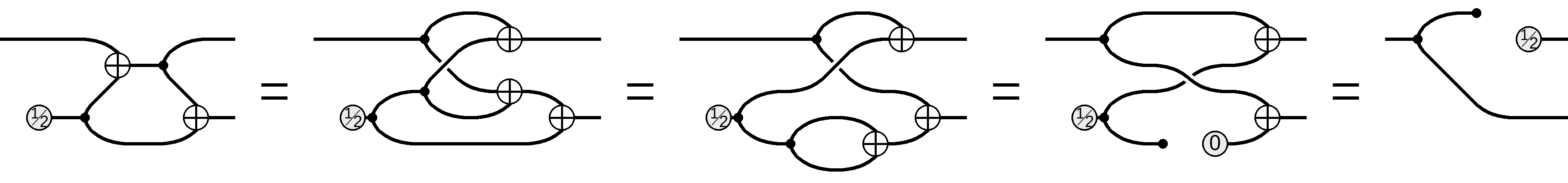}
\end{center}
The equations use the fact that $\oplus$ is deterministic, associative, idempotent, unitary with respect to $\langle 0 \rangle$ and nullary with respect to $\langle \frac{1}{2} \rangle$. All of these properties can be proved via the axioms of Markov categories with probabilistic choice. For example the fact that $\langle \frac{1}{2} \rangle$ is absorbing for $\oplus$ is proved using the laws in Figure~\ref{fig:derived laws markov p choice} as follows:
\begin{center}
    \vcen{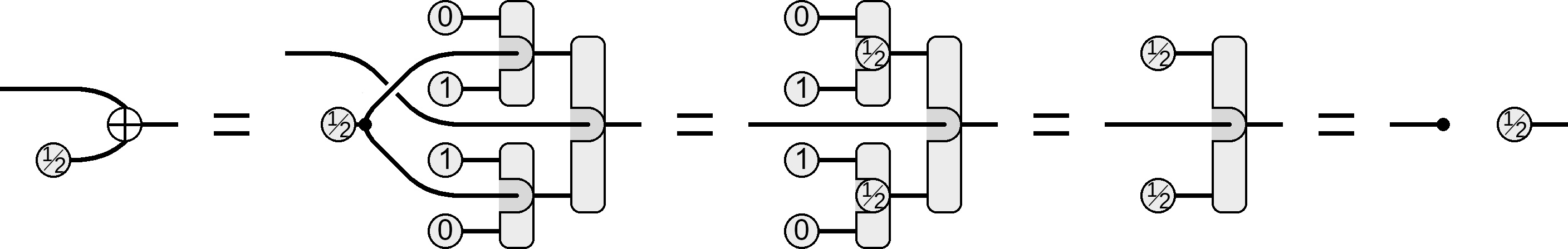}
\end{center}

Looking at $\PI{\FStoch_{\{1,0\}}}$, the above example can be extended to operations on $\Bool^\star$, encrypting and decrypting $\star$-tuples of bits.
We now show that the parametric iteration construction preserves the structure of Markov categories with probabilistic choice.

\begin{lemma}\label{lemma:PI on markov with p choice}%
    If $\C$ is a Markov category with probabilistic choice, then so is $\PI{\C}$.%
\end{lemma}
\begin{proof}
	For $\Bool$ we can take $\br{\Bool}$, and for $\langle p \rangle$ we take $\br{\langle p \rangle}$.
    We define discard $\discharger{A\star}$, copy $\copier{A\star}$ and case $\phi_{A^\star}$ for $\PI{\C}$ inductively on objects, with their definition on $A^\star$ being:

    \begin{center}
        \vcen{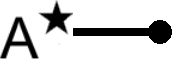}
        :=
        \vcen{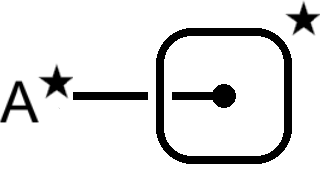}
        \qquad
        \vcen{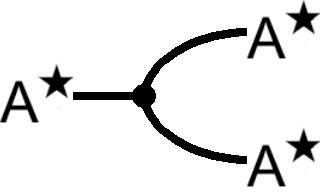}
        :=
        \vcen{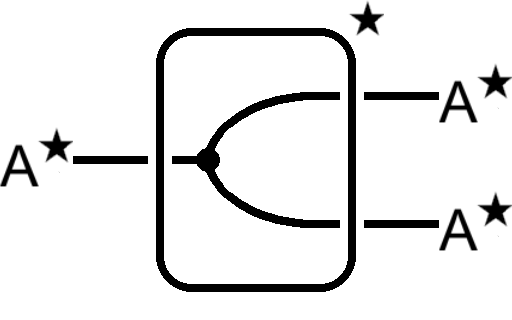}
        \qquad
        \vcen{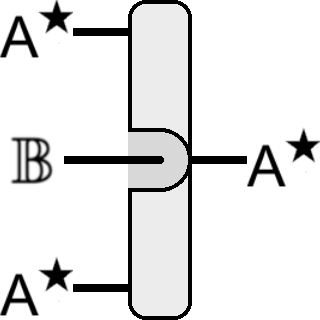}
        :=
        \vcen{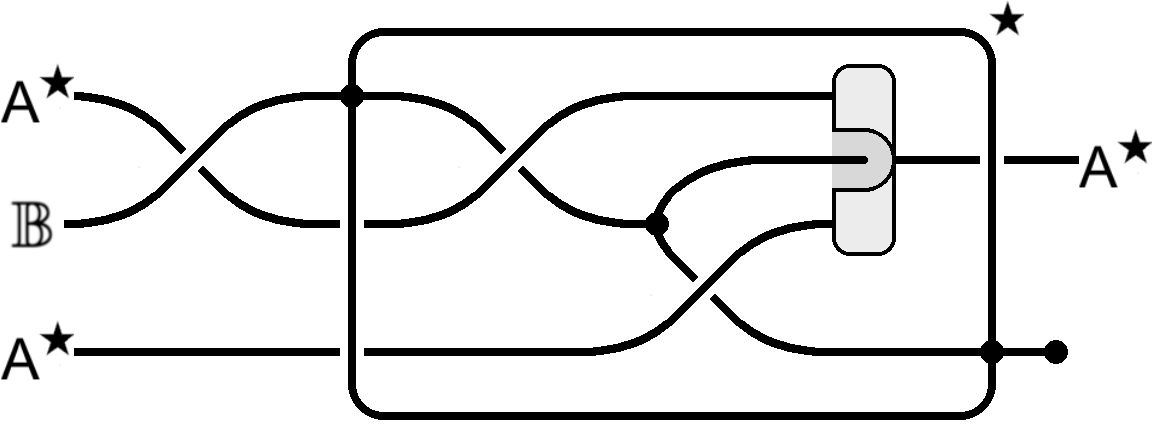}
    \end{center}
    The relevant properties can be shown using the diagonal method, with the naturality equations in particular being instances of Newton's cradle.
\end{proof}

	\subsection{Quantitative Reasoning}\label{ssec:quantitative}
    The theory developed in the previous section allows for exact reasoning about boolean stochastic maps, enabling us to derive precise equalities between them.

    However, in some cases—especially when parametric iteration is involved—we are interested in reasoning about approximate equivalence instead. To do so, we must move to metric-enriched symmetric monoidal categories. We begin by recalling some basic definitions.
\begin{definition}
    A \emph{metric space} $(A, d)$ consists of a set $A$ and a distance function $d : A \times A \to [0, \infty)$ such that for all $a,b,c\in A$:
    \begin{equation}\label{eq:distance}
		d(a,b) = 0 \Leftrightarrow a = b,
		\qquad
		d(a,b) = d(b,a),
		\qquad
		d(a,c) \leq d(a,b) + d(b,c).
	\end{equation}
    A \emph{morphism between metric spaces} $f : (A , d_A) \tom (B , d_B)$ is given by a non-expansive function $f : A \to B$, that is a function satisfying $d_B(f(a), f(b)) \leq d_A(a,b)$ for all $a,b\in A$.
\end{definition}

    \begin{definition}
        The category $\mathbb{M}$ has metric spaces as objects and non-expansive functions as arrows.
        This is a monoidal category, with $\otimes$ defined on objects as the metric space $(A \otimes B, d_{A \otimes B})$, where $A \otimes B$ is the cartesian product of sets and $d_{A \otimes B}((a,b) , (a' , b')) := d_A(a , a') + d_B(b , b')$, and on arrows as $(f \otimes g)(a,c) := (f(a) , g(c))$, for every $f \colon (A,d_A) \tom (B,d_B)$ and $g \colon (C,d_C) \tom (D,d_D)$.
    \end{definition}

    Our objects of study in this section are \emph{metric enriched monoidal categories}, or categories enriched over $\mathbb{M}$. Unfolding the definition, a metric enriched monoidal category is a category whose %
    homs
		are metric spaces, such that:
    \begin{equation}\label{eq:composition non expansive}
		d_{A \to C}(f ; g , f' ; g') \leq d_{A \to B}(f , f') + d_{B \to C}(g , g')
	\end{equation}
		and
	\begin{equation}\label{eq:tensor non expansive}
		d_{A \otimes C \to B \otimes D}(f \otimes g , f' \otimes g') \leq d_{A \to B}(f , f') + d_{C \to D}(g , g').
	\end{equation}

	\begin{lemma}\label{lemma:fstoch is metric enriched}
		$\FStoch$ is a metric enriched symmetric monoidal category.
	\end{lemma}
	\begin{proof}
		For each set $A$, we can define a metric $\tv_A$ on $\Dist A$ called the \emph{total variation metric}~(see e.g. \cite{kolmogorov1970introductory}). We give two equivalent formulations of $\tv_A$ below, for any $v, w \in \Dist A$.
		\[\tv_A(v,w) := \sum_{a \in \texttt{supp}(v) \cup \texttt{supp}(w)} \frac{|v(a) - w(a)|}{2}  = 1- \sum_{a \in \texttt{supp}(v) \cap \texttt{supp}(w)} \texttt{min}(v(a) , w(a))\]

		We take as distance between arrows $d_{A \to B}(f,g) := \texttt{max}_{a \in A} \tv_B(f(a) , g(a))$, for every $f, g \colon A \to B$ in $\FStoch$. This choice of distance makes $\FStoch$ metric enriched (see Example 3.2.7 in~\cite{houghton2021mathematical}).

	\end{proof}

	Just as we approached exact reasoning equationally, we now aim to reason quantitatively in a similar style. To this end, we introduce a proof system that axiomatizes distances between morphisms of metric-enriched Markov categories with probabilistic choice.
	\begin{equation}\label{eq:metric}
		\begin{array}{ccccc}
			\frac{}{f \equiv_0 f} \quad \frac{}{f \equiv_1 g}
			&
			\frac{f \equiv_\delta g}{g \equiv_\delta f}
			&
			\frac{f \equiv_{\delta} g \quad g \equiv_{\delta'} h}{f \equiv_{\delta + \delta'} h}
			&
			\frac{f \equiv_{\delta} g \quad {\delta} \leq {\delta'}}{f \equiv_{\delta'} g}
			&
			\frac{f \equiv_\delta f' \quad g \equiv_\delta g'}{(f \otimes id_\Bool \otimes g) ; \phi_X \equiv_\delta (f' \otimes id_\Bool \otimes g') ; \phi_X}
			\\[10pt]
			\frac{f \equiv_{\delta} f'}{f ; g \equiv_{\delta} f' ; g}
			&
			\frac{g \equiv_{\delta} g'}{f ; g \equiv_{\delta} f ; g'}
			&
			\frac{f \equiv_{\delta} f'}{f \otimes g \equiv_{\delta} f' \otimes g}
			&
			\frac{g \equiv_\delta g'}{f \otimes g \equiv_{\delta} f \otimes g'}
			&
			\frac{f \equiv_\delta f' \quad g \equiv_\gamma g'}{(f \otimes g) ; \phi^p_X \equiv_{p \cdot \delta + (1-p) \cdot \gamma} (f' \otimes g') ; \phi^p_X}
		\end{array}
	\end{equation}

	We write $f \equiv_\delta g$ to say that $f$ and $g$, of the same type, are provably at most $\delta$ away from each other. Moreover, we say that a proof system is \emph{sound} if $f \equiv_\delta g$ implies $d(f,g) \leq \delta$, and is \emph{complete} if $f \equiv_{d(f,g)} g$ for each pair of morphisms $f,g$ of the same type.

	Note that a sound proof system allows one to prove upper bounds to the actual distance between morphisms. As such, certain proofs give better bounds than others, allowing one to improve on a proof of distance by finding different routes. In particular, in order to prove security, one would like to prove that certain distances are below a level of allowed risk $\varepsilon$.

	We now show that the proof system above is sound and complete for $\FStoch_{\{1,0\}}$, equipped with the total variation distance.

	This approach varies from the total variation treatment in \cite{quantitativemonoidalalgebra}, where the monoidal product combines different cases (disjunctive), whereas in our case the monoidal product combines resources (conjunctive).

	\begin{lemma}\label{lemma:quantitative soundness and completeness}
		The proof system in~\eqref{eq:metric} is sound and complete for $\FStoch_{\{1,0\}}$. %
	\end{lemma}
    \begin{proof}
        See Appendix~\ref{ap:4}.
    \end{proof}

\section{Asymptotic Equivalence}\label{sec:asyeq}

	For a metric enriched symmetric monoidal category $\C$, we can equip $\PI{\C}$ with a notion of asymptotic equivalence.
	A nonnegative function $f \colon \Nat \to [0 , \infty)$ \emph{approaches zero}, written as $\lim_{k \to \infty} f(k) = 0$, if for any possible threshold $\varepsilon \in (0, \infty)$ there is some point $N \in \Nat$ after which the function is below $\varepsilon$. Symbolically: $\forall \varepsilon > 0.\, \exists N > 0.\, \forall k \geq N.\, f(k) < \varepsilon$.

	From the perspective of cryptography, we consider $\varepsilon$ the risk factor we are trying to beat, or the allowable failure rate of our algorithm. Then $N > 0$ is the minimum security parameter to guarantee that the actual probability of failure is below the allowable risk.

	The following definition captures indistinguishability between morphisms in $\PI{\C}$, given that they may be investigated a polynomial number of times by the distinguisher (e.g. adversary), and this distinguisher may choose a polynomial amount of inputs and study each probabilistically generated output.

	\begin{definition}[Asymptotic Equivalence]
		For a metric enriched symmetric monoidal category $\Cat{C}$, the asymptotic equivalence relation $\equiv$ on morphisms of $\PI{\Cat{C}}$ is defined as follows:
		\[
		f \equiv g \iff \forall a \in \Nat. \; \lim_{k \to \infty} k^a \cdot d(S_k(f) , S_k(g)) = 0 \qquad \text{for any $f,g \colon X \to Y$ in $\PI{\Cat{C}}$.}
		\]
	\end{definition}

	\begin{proposition}\label{prop:asyeq is a congruence}
		$\equiv$ is preserved under sequential composition and parallel composition.
	\end{proposition}

	We can show that the Newton's cradle equation from before also holds for asymptotic equivalence. We prove this using a method similar to the diagonal method. As a consequence, asymptotic equivalence is preserved under application of the $\star$-iteration operation $\tau^\star$.

	\begin{lemma}[Newton's cradle, reprised]\label{lem:Newton2}
		For any $f \colon S \to Z, g \colon Z \otimes \ol{A} \cdot 1 \to \ol{B} \cdot 1 \otimes Z, h \colon S \otimes \ol{A} \cdot 1 \to \ol{B} \cdot 1 \otimes S$, if \;\;$\vcen{pics/S3/Newton_S_a} \equiv \vcen{pics/S3/Newton_S_b}$ \;\; then \;\; $\vcen{pics/S3/Newton_S_c} \equiv \vcen{pics/S3/Newton_S_d}$.
	\end{lemma}
	\begin{proof}

		Let $k \in \Nat$, and take $S' = S_k(S)$, $Z' =  S_k(Z)$, $\ol{C} = S_k(\ol{A})$, $\ol{D} = S_k(\ol{B})$, $f_k = S_k(f)$, $g_k = S_k(g)$, and $h_k = S_k(h)$.
		We shall define $a \in [0,1]$ as:

		\noindent
		$\;\;a := d \left( \vcen{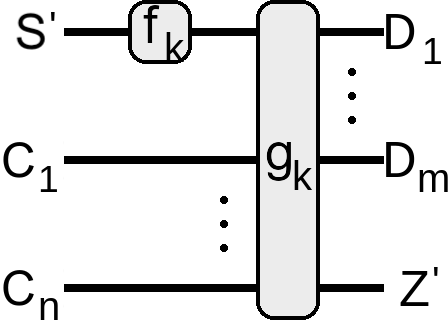} \;,\; \vcen{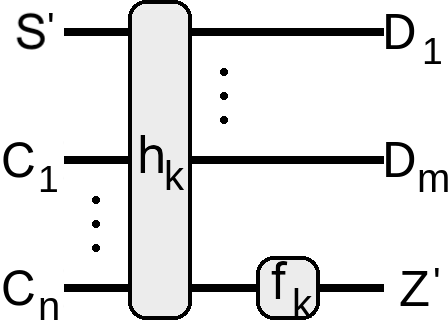} \right)$
		 and prove that \; $d \left( \vcen{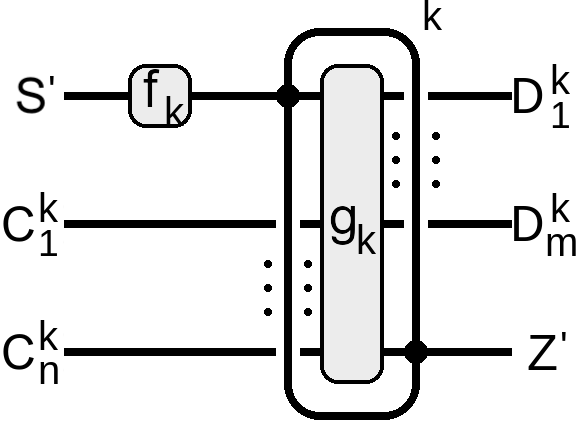} \;,\; \vcen{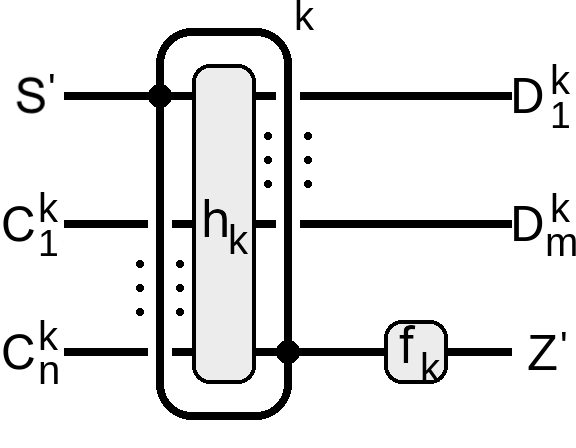} \right) \leq k a$

		We prove this, by showing by induction on $r \in \Nat$ that:
		$\;\;\vcen{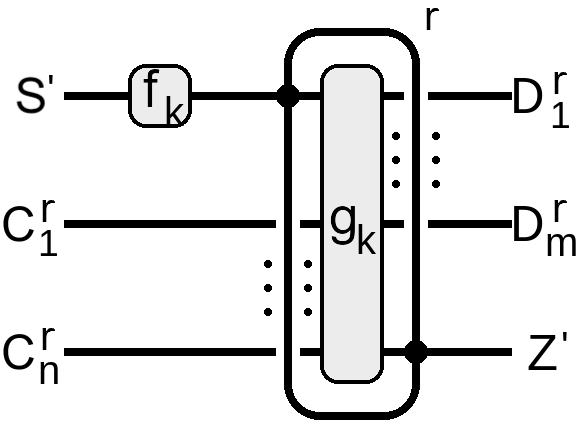} \equiv_{r a} \vcen{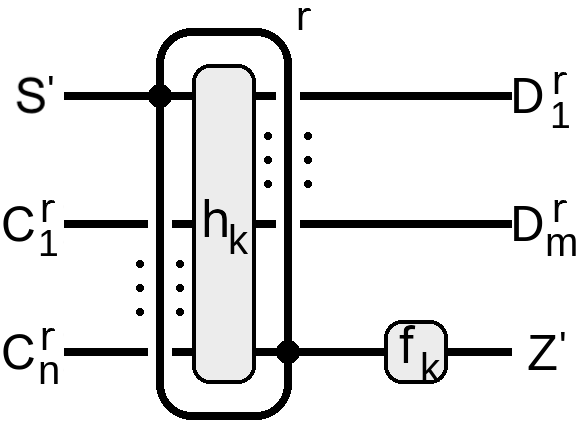}\;\;$

		Induction basis:
			$
			\vcen{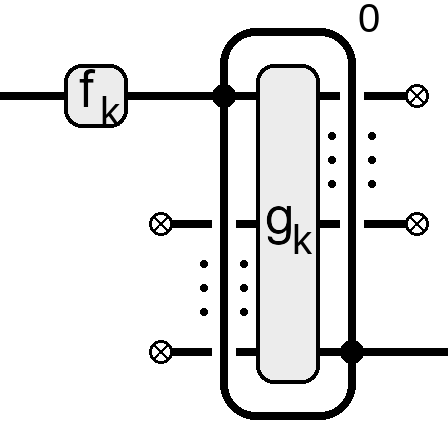}
			\stackrel{(\text{def})}{=}
			\vcen{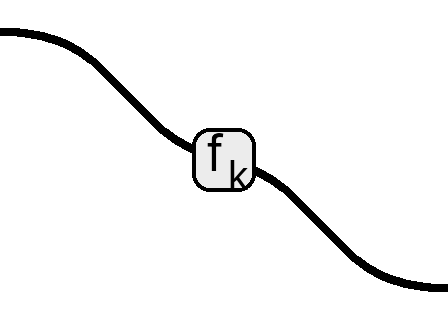}
			\stackrel{(\text{def})}{=}
			\vcen{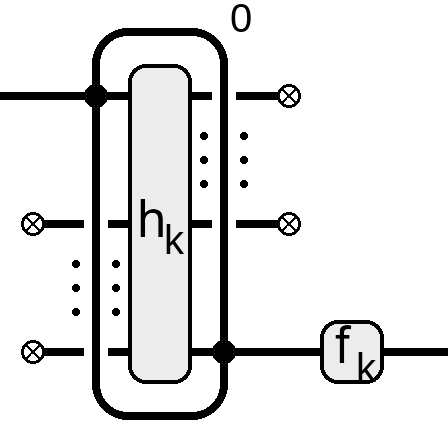}
			$

		Induction step:
			$
			\vcen{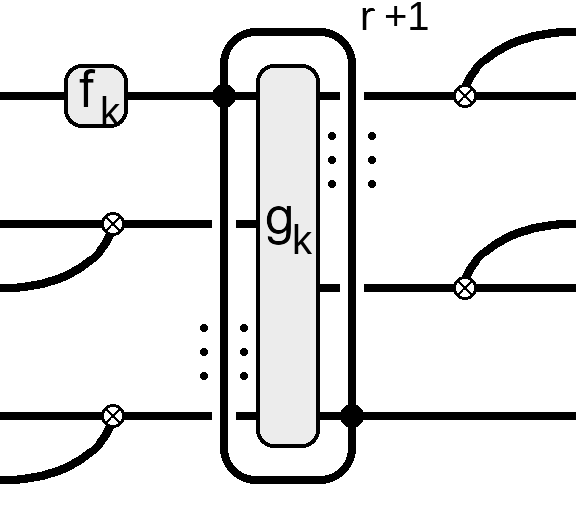}
			\stackrel{(\text{def})}{=}
			\vcen{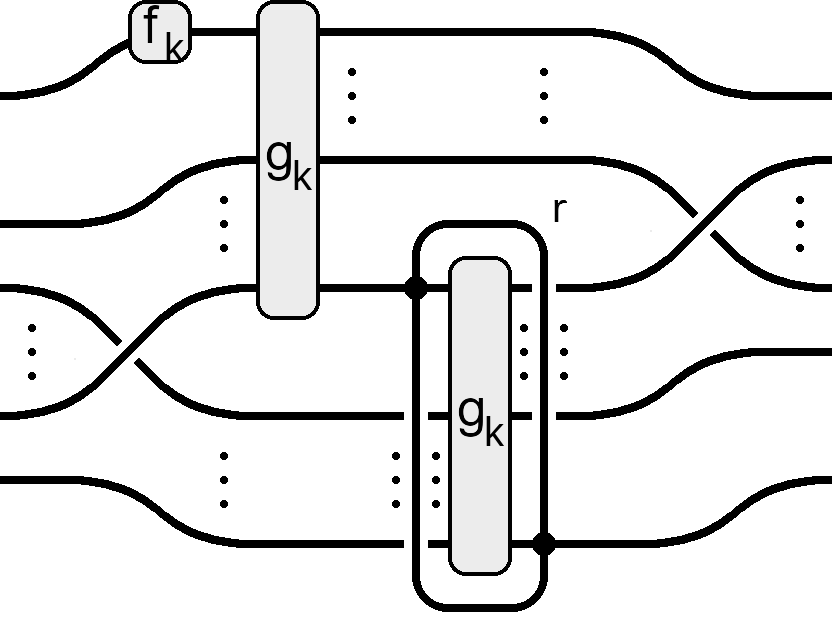}
			\stackrel{(\text{I})}{\equiv_a}
			\vcen{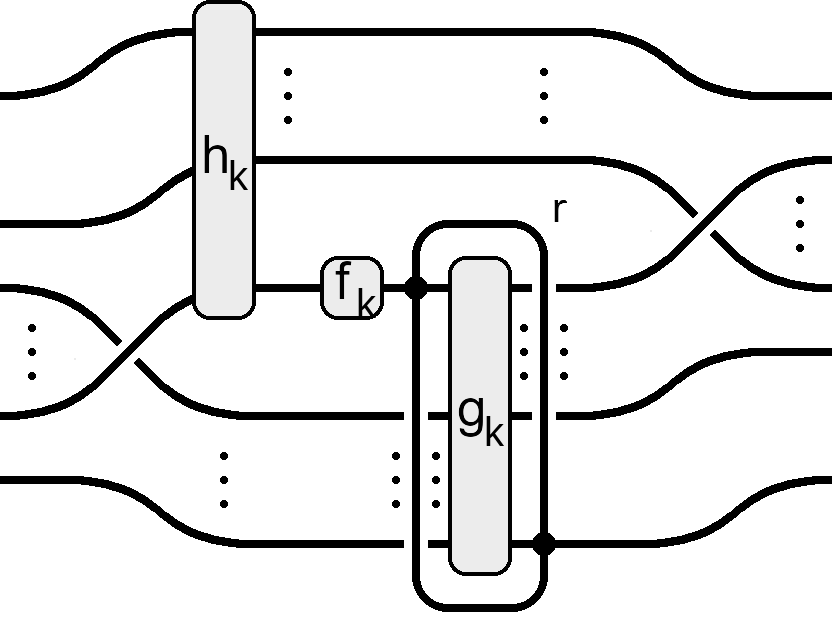}
			$
	
		\begin{center}
			$
			\stackrel{(\text{hyp})}{\equiv_{ra}}
			\vcen{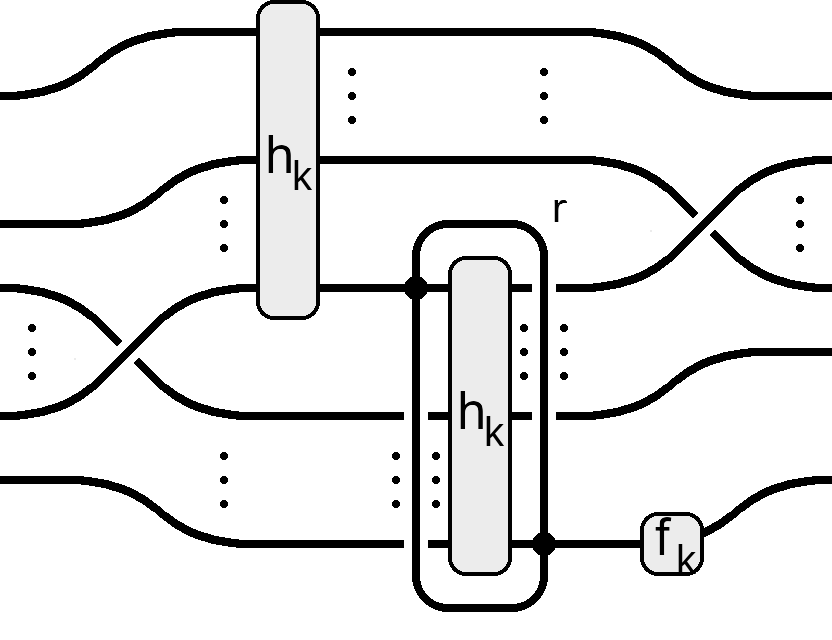}
			\stackrel{(\text{def})}{=}
			\vcen{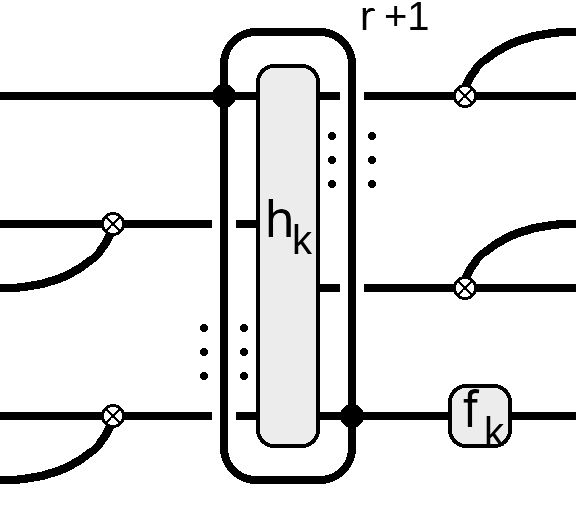}
			$
		\end{center}

		\noindent
		together getting an upper bound for the distance of $a + ra = (r+1)a$, as required.

		Taking $r = k$, we get a distance of $ka$ between the two sides of the equation.
		Hence, given any $b \in \Nat$ we get

		\noindent
		$\lim_{k \to \infty} k^b d\left( \vcen{pics/S5/Newton_I_c} , \vcen{pics/S5/Newton_I_d} \right) = \lim_{k \to \infty} k^{b+1} d\left( \vcen{pics/S5/Newton_I_a} , \vcen{pics/S5/Newton_I_b} \right) = 0$

	\end{proof}

	\subsection{Examples}

	Asymptotic equivalence is used to reason whether it is possible to increase the security parameter to significantly minimize the risk of two programs being distinguished. This allows us to formally ignore undesirable situations if their probability of occurring is asymptotically insignificant, or in other words: \emph{negligible}.
    Our first example shows that a randomly generated key of length $\sigma$ is unlikely to be any specific key.
	This has three useful consequences:
	\begin{enumerate}
		\item Two randomly generated keys of length $\sigma$ are not equal (asymptotically).
		\item If an adversary attempts to guess a key before that key is generated (or without any knowledge of the generated key), then this guess is wrong (asymptotically).
		\item We can uniformly generate an element of $\mathbb{Z}_{2^\sigma-1}$ (asymptotically),
	\end{enumerate}

	We first define $\textsf{all-1} : \Bool^\star \to \Bool$ which checks if all elements of $\Bool^\star$ are equal to $1$. Given an and-gate $\textsf{\&}: \Bool^2 \to \Bool$, we define $\textsf{all-1} := (\langle 1 \rangle \otimes id_{\Bool^\star}) ; \tau^\star_{\Bool , (\Bool) , ()}(\textsf{\&})$.

	\begin{lemma}\label{lem:all1}
		If $p < 1$, then
        $\vcen{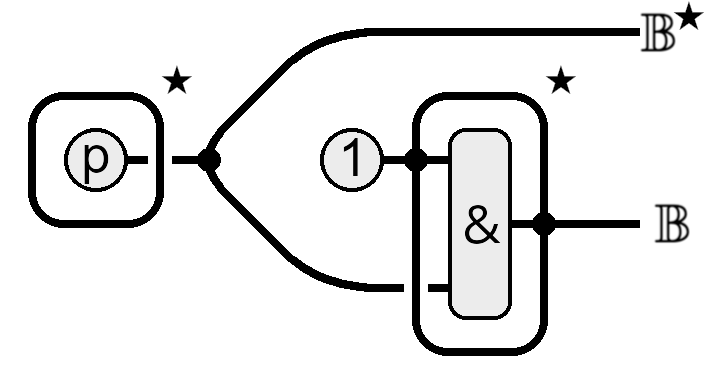} \equiv \vcen{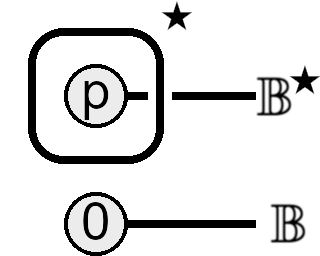}$.
	\end{lemma}
	\begin{proof}
		By induction on $k$, we can prove that the two diagrams as seen through $S_k$ in the Lemma statement are at most $p^k$ apart in the underlying metric-enriched category (to cover an edge-case we use $0^0 = 1$). Given that $p < 1$, we have that for any $a$, $\lim_{k \to \infty} k^a \cdot p^k = 0$.
		First, observe that by (\ref{def:tcomp}) and (\ref{def:tmono}),  $\vcen{pics/S5/ex/zero_a} = \vcen{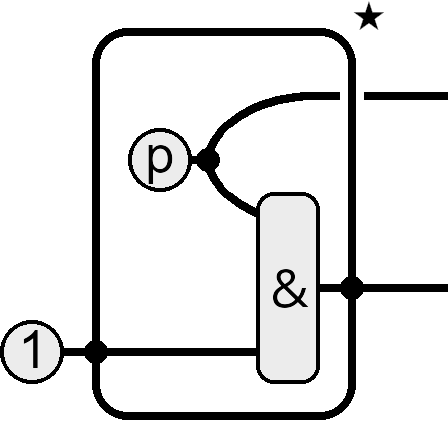}$.

		For $k = 0$, the required distance $\equiv_1$ is trivially true, i.e. $\vcen{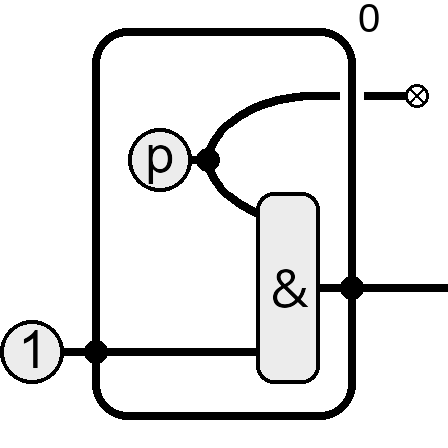} \equiv_1 \vcen{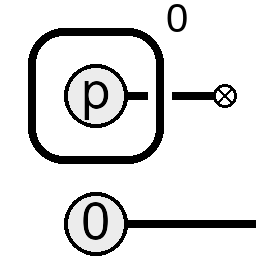}$.

		For $k' = k+1$, we do the following steps.

		\begin{center}
		$\vcen{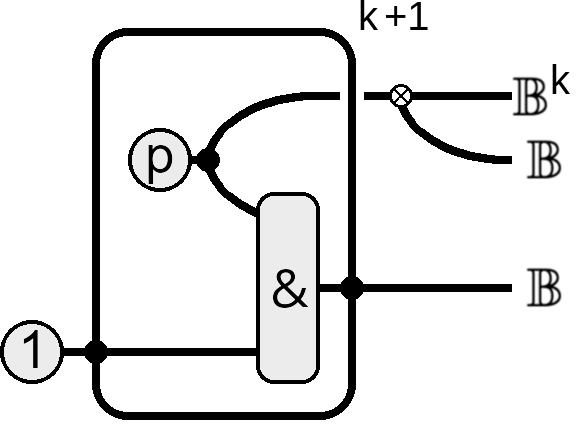}
		= \vcen{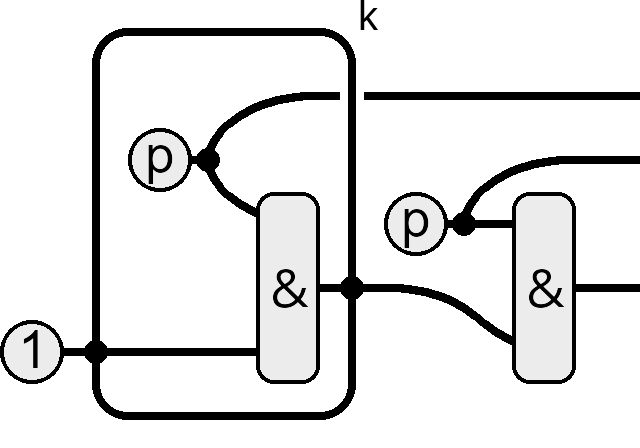}
		= \vcen{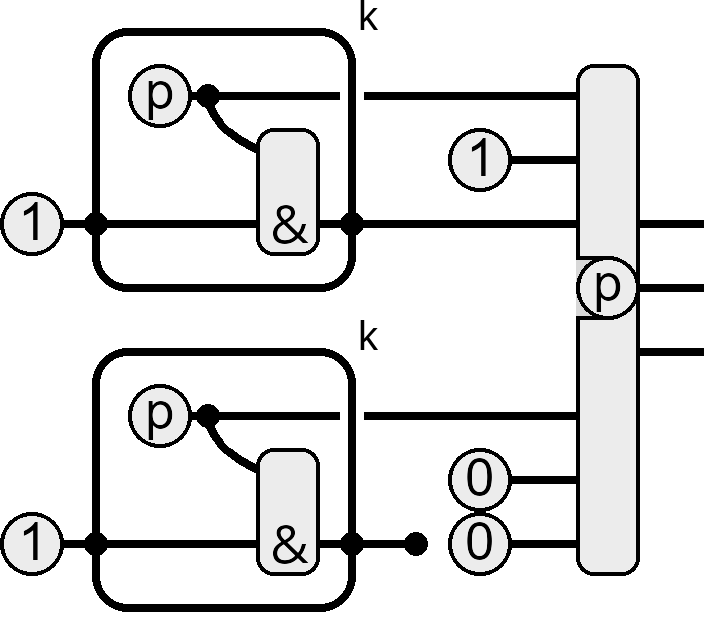}
		\equiv_{p \cdot p^k + (1-p) \cdot 0} \vcen{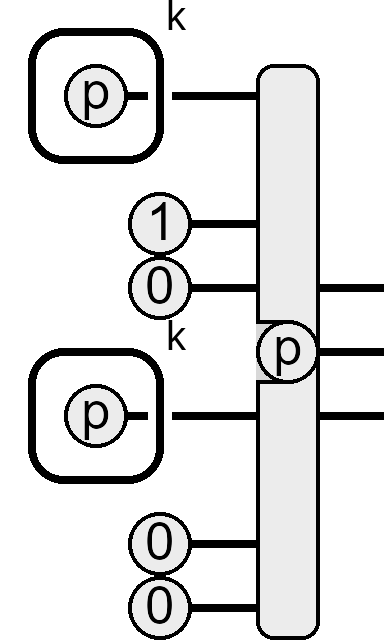}
		= \vcen{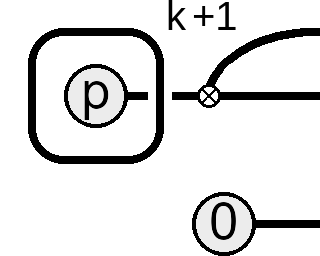}
		$
		\end{center}
	\end{proof}

	In other words, no randomly generated key with weighted choice below 1 consists of only ones. Note in particular that this is the case even if the key is used in some other context.
	We define equality on bits using a negated xor gate $(=) : \Bool^2 \to \Bool$.
	Then equality on $\star$-tuples $(=)^\star : \Bool^\star \otimes \Bool^\star \to \Bool$ is done by iterating the equality and checking that all entries are one.
	\begin{lemma}\label{lem:keyguess}
        $\vcen{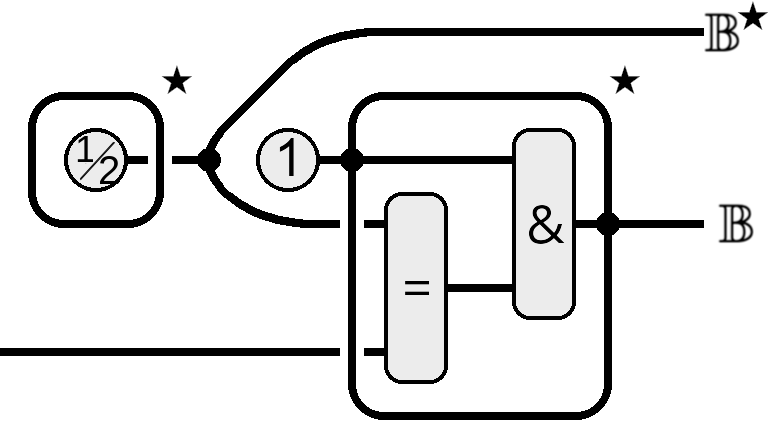} \equiv \vcen{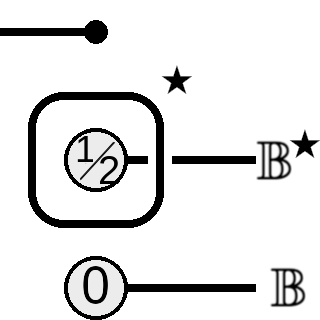}$.
	\end{lemma}
	\begin{proof}
		First, we merge the two iterations into one using (\ref{def:tcomp}) and (\ref{def:tmono}). Then we can reorder $\langle 1/2 \rangle$ and $(=)$ operations in such a way that they still give the same probabilistic stochastic map. Thirdly, we can separate the iteration into two iterations, and apply Lemma \ref{lem:all1} to the first iteration. Lastly, merge the two remaining iterations again and note that equality to a random bit gives a random bit (similar to what we observed with xor before).

		\begin{center}
			$\vcen{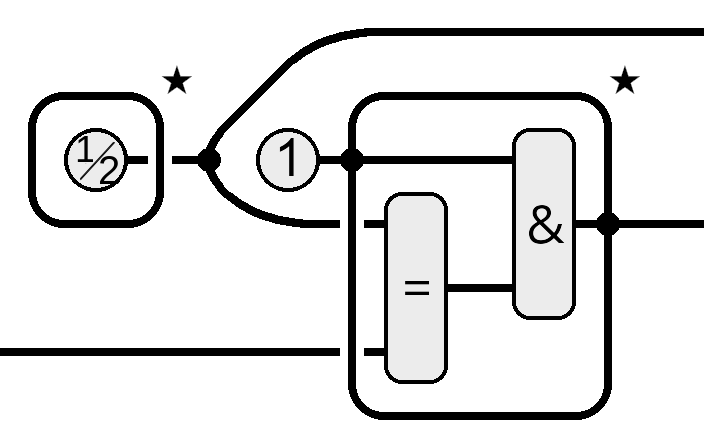}
			= \vcen{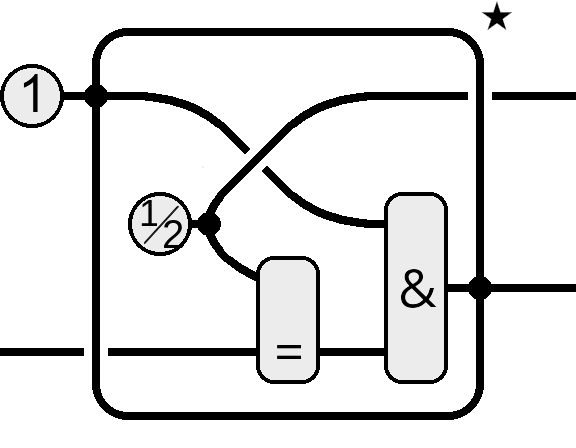}
			= \vcen{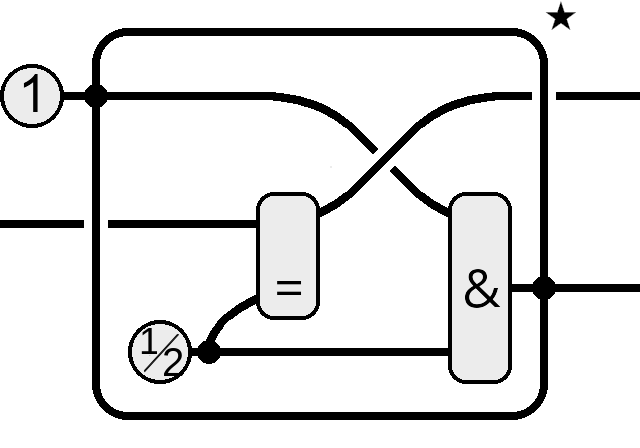}
			= \vcen{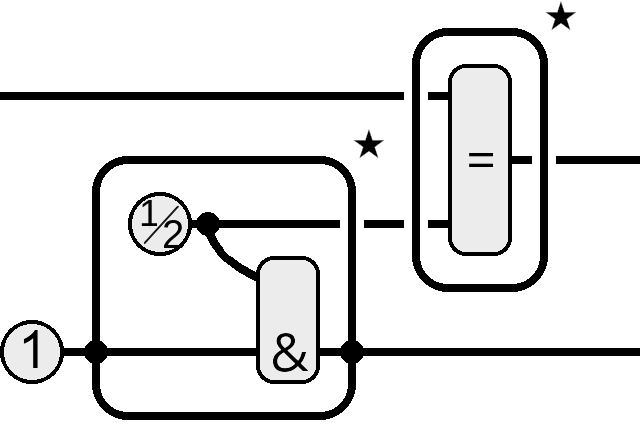}
			\equiv \vcen{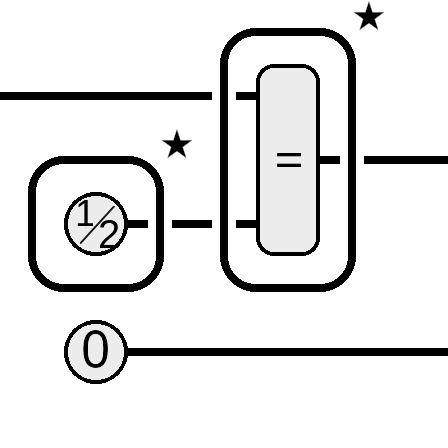}
			= \vcen{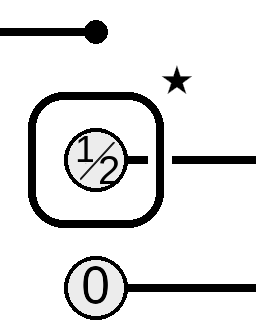}
			$
		\end{center}
	\end{proof}

    We consider one final example, the iterated von Neumann's trick. This is a method for simulating a fair coin flip using a biased coin that lands on heads with probability $p \in (0,1)$. The method involves flipping the biased coin twice and selecting the first outcome if both flips match; otherwise, the process is repeated.
    When a form of parametric iteration is available, the trick can be formalised as Example 11 in~\cite{torres2024iteration}. Here we give a categorical and diagrammatic representation as the diagram in the lemma.
    \begin{lemma}\label{eq:vn trick}
    	$\vcen{pics/S5/ex/Von_a} \equiv \vcen{pics/S5/ex/Von_b}$
    \end{lemma}

	\begin{proof}
    The equation above states that the process of flipping a bit with probability $p$ a total of $k$-times converges to a fair coin flip as $k$ increases. We prove this by showing that for a constant number $k$ of iterations, the two diagrams are $(2p-1)^k$ away if $p \geq 1/2$, and $(1-2p)^k$ away if $p \leq 1/2$. In both cases, the distance approaches zero exponentially.

    We shall look at the case in which $p \geq 1/2$, the other case goes similarly.
    We prove the result by induction on $k$, as follows. The base case holds since everything is $1$ away from everything else, which satisfies the condition.
    For the inductive step, let $k' = k + 1$ and observe that the following holds true, where we take $q = 2p-1$ and $h = q^k$ the distance from the induction hypothesis. We write $n : \Bool \to \Bool$ for the not gate, then:
    \begin{center}
    $\vcen{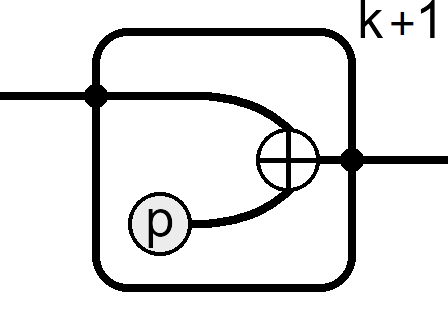}
    = \vcen{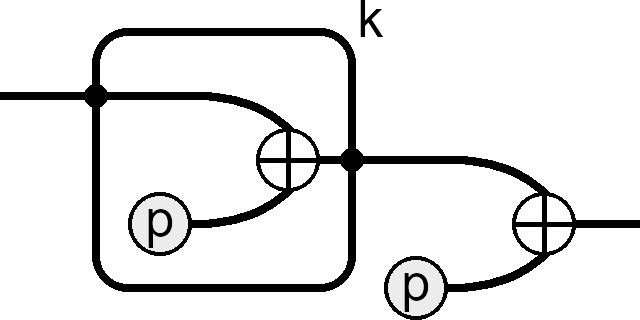}
    = \vcen{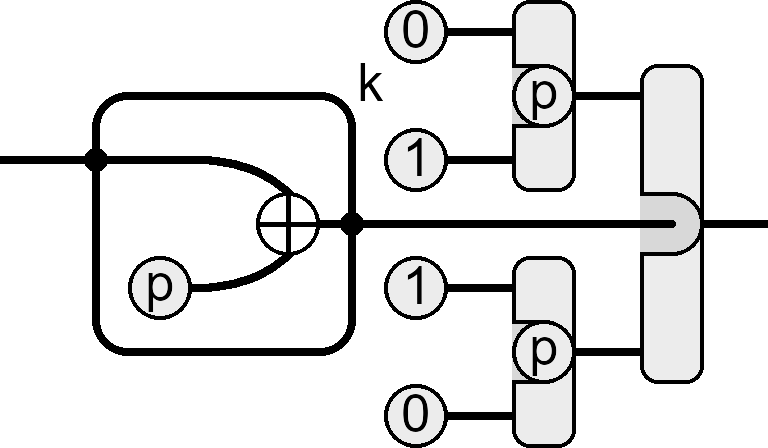}
    = \vcen{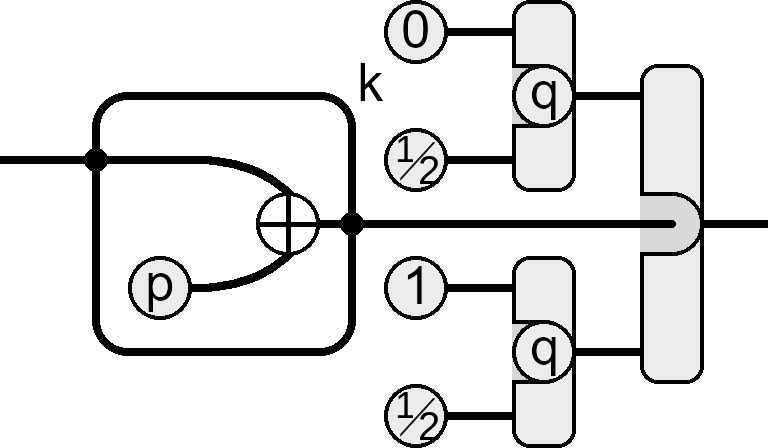}$
	\end{center}
	\begin{center}
	$
    = \vcen{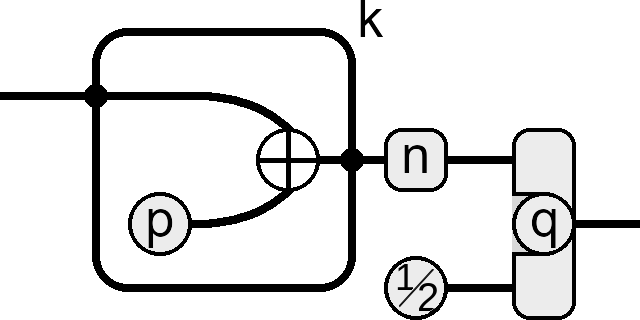}
    \stackrel{\text{induction hypothesis}}{\equiv_{q \cdot q^k + (1-q) \cdot 0}} \vcen{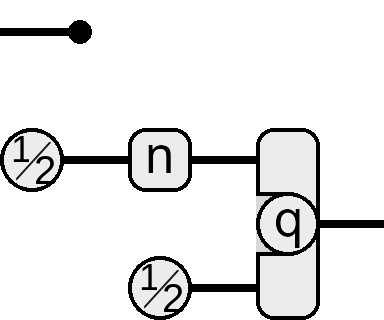}
    = \vcen{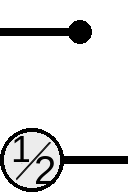}
    $
	\end{center}
	\end{proof}

\section{Conclusions and Future Work}\label{sec:conc}

This work is a step towards an algebraic approach to cryptographic proofs completely parametric in security parameters. Such proofs normally entail showing that all protocols have polynomial time complexity with respect to the parameters, and require intricate probabilistic calculations.
Binding the parameter helps avoid unsound proof pitfalls, such as accidentally using asymptotic rewriting while performing induction over number of iterations.

In order to incorporate further cryptographic protocols however, we need two further extensions. First, we need to allow for cryptographic primitives, which should be given as families of operations in $\C$ indexed over the security parameter, with accompanying assumptions regarding distances. This should be possible by extending $\C$ accordingly, and deriving further equations once one takes the parametric iteration over $\C$.

The second extension entails dealing with operations like \emph{pseudorandom number generators}, as e.g. used in stream ciphers. In the current framework, we cannot generate more randomness from less. Concretely, we theoretically cannot have a morphism $f : \Bool^\star \to \Bool^\star \otimes \Bool$ such that $\tau^\star(\langle \frac{1}{2} \rangle) ; f \equiv \tau^\star(\langle \frac{1}{2} \rangle) \otimes \langle \frac{1}{2} \rangle$.
The solution to this problem is to define a coarser relation, that of \emph{contextual equivalence} relative to contexts of morphisms which can only investigate a finite amount of data, meaning whose input and output are constant (from $\C$). Our notion of asymptotic equivalence is (trivially) sound with respect to this notion of contextual equivalence, though further investigation of this relation is for future research.

On the topic of probabilistic computations, we use the fact that all our computations are total, which is reflected in the naturality property of our if-gate operation. Extending to partial probabilistic computations, we should generalize to partial Markov categories with a condition/equality operation. Our if-gate as well as the discard operation will only be natural over all total computations (those not involving conditioning). 
In order to more accurately model an if-statement as might be used in a programming language, we should either introduce a more structured functorial if operation, or move towards distributive monoidal categories and their associated languages of tape/sheet diagrams~\cite{bonchi2025tape,comfort2020sheet}.

\bibliography{main}
\vfill

\appendix

\section{Additional Material}\label{app:additional}

\begin{figure}[H]
    \[
    \begin{array}{l@{\qquad}l@{\qquad}l@{\qquad}l}
        \toprule
        \multicolumn{4}{c}{\text{Markov categories}}
        \\
        \midrule
        \multicolumn{4}{c}{
            \begin{array}{l@{\qquad}l@{\qquad}l}
                \copier{A} ; (\discharger{A} \otimes id_A) = id_A
                &
                \copier{A} ; \chi_{A,A} = \copier{A}
                &
                \copier{A} ; (\copier{A} \otimes id_A) = \copier{A} ; (id_A \otimes \chi_{A,A})
            \end{array}
        }
        \\
        \multicolumn{4}{c}{f ; \discharger{B} = \discharger{A}}
        \\
        \midrule
        \multicolumn{4}{c}{\text{Probabilistic choice structure}}
        \\
        \midrule
        \langle 1 \rangle ; \copier{\Bool} = \langle 1 \rangle \otimes \langle 1 \rangle
        &
        \langle 0 \rangle ; \copier{\Bool} = \langle 0 \rangle \otimes \langle 0 \rangle
        &
        \phi^1_A = (id_A \otimes \discharger{A})
        &
        \phi_A ; f = (f \otimes \id{\Bool} \otimes f) ; \phi_B
        \\
        \multicolumn{4}{l}{
            \begin{array}{@{}l}
                \phi_I = \discharger{\Bool}
                \\
                \phi_{A \otimes C} = (id_{A^2} \otimes \copier{\Bool} \otimes id_{B^2}) ; (id_A \otimes \chi_{C , \Bool} \otimes \chi_{\Bool , A} \otimes id_C) ; (id_{A \otimes \Bool} \otimes \chi_{C,A} \otimes id_{\Bool \otimes C}) ; (\phi_A \otimes \phi_C)
                \\
                \phi^p_A = \chi_{A \otimes \Bool , A} ; (id_A \otimes \chi_{A , \Bool}) ; \phi^{1-p}_A
                \\
                (\phi^p_A \otimes id_A) ; \phi^q_A = (id_A \otimes \phi^r_A) ; \phi^t_A \text{ if } p q = t \text{ and } (1-r)(1-t) = 1-q
            \end{array}
        }
        \\
        \multicolumn{4}{l}{
            \begin{array}{@{}l@{\qquad}l}
                (\langle 1 \rangle \otimes id_\Bool \otimes \langle 0 \rangle) ; \phi_\Bool = id_\Bool
                &    
                (\copier{A} \otimes id_\Bool) ; (id_A \otimes \chi_{A , \Bool}) ; \phi_A = id_A \otimes \discharger{\Bool}
            \end{array}
            
        }
        \\
        \bottomrule
    \end{array}
\]
\caption{Axioms in term-based syntax.}    
\label{fig:axioms table}
\end{figure}

\section{Symmetric Monoidal Categories}\label{sec:monoidal}%

\begin{definition}
    A \emph{symmetric monoidal category} $(\Cat{C}, \otimes, I, \chi)$ consists of a category $\Cat{C}$, a bifunctor $\otimes \colon \Cat{C} \times \Cat{C} \to \Cat{C}$, an object $I$ and natural isomorphisms
    \[
    \alpha_{X,Y,Z} \colon (X \otimes Y) \otimes Z \to X \otimes (Y \otimes Z)
    \quad
    \lambda_X \colon I \otimes X \to X
    \quad
    \rho_X \colon X \otimes I \to X
    \quad
    \chi_{X,Y} \colon X \otimes Y \to Y \otimes X
     \]
     satisfying the coherence axioms below.
     \[
     \scalebox{0.7}{\begin{tikzcd}[column sep=tiny]
                                                                                                                           & (X \otimes Y) \otimes (Z \otimes W) \arrow[rd, "{\alpha_{X, Y, Z \otimes W}}"] &                                                                                 \\
((X \otimes Y) \otimes Z) \otimes W \arrow[ru, "{\alpha_{X \otimes Y, Z, W}}"] \arrow[d, "{\alpha_{X,Y,Z} \otimes \id{W}}"'] &                                                                                & X \otimes (Y \otimes (Z \otimes W))                                             \\
(X \otimes (Y \otimes Z)) \otimes W \arrow[rr, "{\alpha_{X, Y \otimes Z, W}}"']                                            &                                                                                & X \otimes ((Y \otimes Z) \otimes W) \arrow[u, "{\id{X} \otimes \alpha_{Y,Z,W}}"']
\end{tikzcd}
}
     \quad
     \scalebox{0.7}{\begin{tikzcd}[column sep=tiny]
    (X \otimes I) \otimes Y \arrow[rr, "{\alpha_{X, I, Y}}"] \arrow[rd, "\rho_X \otimes \id{Y}"'] &             & X \otimes (1 \otimes Y) \arrow[ld, "\id{X} \otimes \lambda_Y"] \\
                                                                                               & X \otimes Y &                                                            
\end{tikzcd}
}
     \]
     \[
    \scalebox{0.7}{\begin{tikzcd}[row sep=normal]
X \otimes Y \arrow[r, "{\chi_{X,Y}}"] \arrow[rd, Rightarrow, no head, shift right] & Y \otimes X \arrow[d, "{\chi_{Y,X}}"] \\
                                                                        & X \otimes Y                            
\end{tikzcd}
}
     \quad
     \scalebox{0.7}{\begin{tikzcd}[row sep=normal]
X \otimes I \arrow[r, "{\chi_{X,I}}"] \arrow[d, "\rho_X"'] & I \otimes X \arrow[d, "\lambda_X"] \\
X \arrow[r, Rightarrow, no head]                                & X                                    
\end{tikzcd}
}
     \quad
     \scalebox{0.7}{\begin{tikzcd}
                                                                                                  & X \otimes (Y \otimes Z) \arrow[r, "{\id{X} \otimes \chi_{Y,Z}}"] & X \otimes (Z \otimes Y) \arrow[rd, "{\alpha^-_{X, Z, Y}}"]       &                         \\
(X \otimes Y) \otimes Z \arrow[ru, "{\alpha_{X, Y, Z}}"] \arrow[rd, "{\chi_{X \otimes Y, Z}}"'] &                                                                  &                                                                     & (X \otimes Z) \otimes Y \\
                                                                                                  & Z \otimes (X \otimes Y) \arrow[r, "{\alpha^-_{Z, X, Y}}"']    & (Z \otimes X) \otimes Y \arrow[ru, "{\chi_{Z, X} \otimes \id{Y}}"'] &                        
\end{tikzcd}
}
     \]
\end{definition}

\begin{definition}
    Let $(\Cat{C}, \otimes^\Cat{C}, I^\Cat{C}, \alpha^\Cat{C}, \lambda^\Cat{C}, \rho^\Cat{C}, \chi^\Cat{C})$ and $(\Cat{D}, \otimes^\Cat{D}, I^\Cat{D}, \alpha^\Cat{D}, \lambda^\Cat{D}, \rho^\Cat{D}, \chi^\Cat{D})$ be two symmetric monoidal categories. A \emph{symmetric monoidal functor} is a \emph{strong monoidal functor}, that is a triple $(F, \epsilon, \mu)$ consisting of a functor $F \colon \Cat{C} \to \Cat{D}$ and two natural isomorphisms
    \[ \epsilon \colon I^\Cat{D} \to F(I^\Cat{C}) \qquad \mu_{X,Y} \colon F(X) \otimes^\Cat{D} F(Y) \to F(X \otimes^\Cat{C} Y) \]
    such that the following diagrams commute
    \[
        \scalebox{0.7}{
\begin{tikzcd}
        {(F(X) \otimes^\Cat F(Y)) \otimes^\Cat{D} F(Z) } && {F(X) \otimes^\Cat{D} (F(Y) \otimes^\Cat{D} F(Z))} \\
        {F(X \otimes^\Cat{C} Y) \otimes^\Cat{D} F(Z) } && {F(X) \otimes^\Cat{D} F(Y \otimes^\Cat{C} Z)} \\
        {F((X \otimes^\Cat{C} Y) \otimes^\Cat{C} Z) } && {F(X \otimes^\Cat{C} (Y \otimes^\Cat{C} Z) )}
        \arrow["{\alpha^{\Cat{D}}_{F(X), F(Y), F(Z)}}", from=1-1, to=1-3]
        \arrow["{F(\alpha^{\Cat{C}}_{X,Y,Z})}"', from=3-1, to=3-3]
        \arrow["{\mu_{X, Y \otimes^\Cat{C} Z}}", from=2-3, to=3-3]
        \arrow["{\id{F(X)} \otimes^\Cat{D} \mu_{Y,Z}}", from=1-3, to=2-3]
        \arrow["{\mu_{X,Y} \otimes^\Cat{D} \id{F(Z)}}"', from=1-1, to=2-1]
        \arrow["{\mu_{X \otimes^\Cat{C} Y, Z}}"', from=2-1, to=3-1]
    \end{tikzcd}
}
    \]
    \[
        \scalebox{0.7}{
\begin{tikzcd}
        {I^\Cat{D} \otimes^\Cat{D} F(X)} & {F(I^\Cat{C}) \otimes^\Cat{D} F(X)} \\
        {F(X)} & {F(I^\Cat{C} \otimes^\Cat{C} X)}
        \arrow["{\epsilon \otimes^\Cat{D} \id{F(X)}}", from=1-1, to=1-2]
        \arrow["{\mu_{I^\Cat{C},X}}", from=1-2, to=2-2]
        \arrow["{F(\lambda^{\Cat{C}}_X)}", from=2-2, to=2-1]
        \arrow["{\lambda^{\Cat{D}}_{F(x)}}"', from=1-1, to=2-1]
    \end{tikzcd}
}
        \scalebox{0.7}{
\begin{tikzcd}
        {F(X) \otimes^\Cat{D} J} & {F(X) \otimes^\Cat{D} F(I^\Cat{C})} \\
        {F(X)} & {F(X \otimes^\Cat{C} I^\Cat{C})}
        \arrow["{\id{F(X)} \otimes^\Cat{D} \epsilon}", from=1-1, to=1-2]
        \arrow["{\mu_{X,I^\Cat{C}}}", from=1-2, to=2-2]
        \arrow["{F(\rho^{\Cat{C}}_X)}", from=2-2, to=2-1]
        \arrow["{\rho^{\Cat{D}}_{F(x)}}"', from=1-1, to=2-1]
    \end{tikzcd}
}
    \]
    and such that $F$ preserves symmetries on both sides, that is $\chi^\Cat{D}_{F(X), F(Y)} ; \mu_{Y,X} = \mu_{X,Y} ; F(\chi^\Cat{C}_{Y,X})$.
\end{definition}

\begin{definition}
    A symmetric monoidal category $\Cat{C}$ is said to be \emph{strict} when $\alpha, \lambda$ and $\rho$ are identity natural isomorphisms. A strict symmetric monoidal functor is a functor $F \colon \Cat{C} \to \Cat{D}$, preserving $\otimes, I$ and $\chi$.
\end{definition}

By MacLane's strictification theorem~\cite{mclane}, every symmetric monoidal category is monoidally equivalent to a strict symmetric monoidal category.

\section{String Diagrams for Symmetric Monoidal Categories}\label{sec:app string diagrams}%

In this appendix we give a brief exposition of string diagrams as terms of a typed language.

Given a set $\sort$ of basic sorts, types are elements of $\sort^\ast$, i.e. words over $\sort$. Terms are generated by the following context-free grammar:
\begin{equation}\label{eq:syntaxsymmetricstrict}
\begin{array}{rcl}
f,g & :=& \; \id{I} \; \mid \; \id{A} \; \mid \; \gen \; \mid \; \chi_{A,B} \; \mid \;   f ; g   \; \mid \;  f \otimes g \\
\end{array}
\end{equation}
where $s$ belongs to a set of generators $\sign$. For any pair of words $X$ and $Y$, we regard $X \otimes Y$ as their concatenation. Each $s \in \sign$, comes with an arity $\ari(s) \in \sort^\ast$ and a coarity $\coar(s) \in \sort^\ast$. Altogether $(\sort, \sign, \ari, \coar)$ forms what is called a \emph{monoidal signature}.

Amongst the terms generated by~\eqref{eq:syntaxsymmetricstrict}, we consider only those that can be typed according to the following inference rules,
\begin{equation}
    \begin{array}{ccc}
        \frac{}{\id{A} \colon A \to A}
        &
        \frac{}{\id{I} \colon I \to I}
        &
        \frac{}{\chi_{A, B} \colon A \otimes B \to B \otimes A}
        \\[10pt]
        \frac{\gen \in \sign \quad \ari(\gen) = X \quad \coar(\gen) = Y}{\gen \colon X \to Y}
        &
        \frac{f \colon X \to Y \quad g \colon Y \to Z}{f ; g \colon X \to Z}
        &
        \frac{f \colon X_1 \to Y_1 \quad g \colon X_2 \to Y_2}{f \otimes g \colon X_1 \otimes X_2 \to Y_1 \otimes Y_2}
    \end{array}
\end{equation}
where $I$ is the empty word.
String diagrams are such terms modulo the axioms of symmetric monoidal categories, listed below.
\begin{equation}\label{eq:smcaxioms}
    \begin{array}{c@{\qquad}c@{\qquad}c}
        (f;g);h=f;(g;h) & \id{X};f=f=f;\id{Y} & (f_1\otimes f_2) ; (g_1 \otimes g_2) = (f_1;g_1) \otimes (f_2;g_2)
        \\
        \multicolumn{3}{c}{\begin{array}{c@{\qquad}c}
            \id{I}\otimes f = f = f \otimes \id{I} & (f \otimes g)\, \otimes h = f \otimes \,(g \otimes h)
        \\
        \chi_{A, B} ; \chi_{B, A} = \id{A \otimes B} & (f \otimes \id{Z}) ; \chi_{Y, Z} = \chi_{X,Z} ; (\id{Z} \otimes f)
        \end{array}}
    \end{array}
\end{equation}

String diagrams enjoy an elegant graphical representation: the identity $id_I$ is the empty diagram, while $id_A$ is a single wire labelled with $A$; symmetries $\chi_{A,B}$ are depicted as crossing of wires; a generator $s  \in \sign$, with $\ari(s) = A_1 \otimes \dots \otimes A_n$ and $\coar(s) = B_1 \otimes \dots \otimes B_m$, corresponds to a box having $n$ labelled wires on the left and $m$ labelled wires on the right; sequential and parallel composition are represented by concatenation and juxtaposition of diagrams, respectively. Indeed, the grammar in \eqref{eq:syntaxsymmetricstrict} can be given graphically as:
\begin{equation*}%
    \begin{array}{rcl}
    f & :=& \; \quad \; \mid \; \vcen{pics/S2/id} \; \mid \; \vcen{pics/S2/gen3} \; \mid \; \vcen{pics/S2/symm} \; \mid \;   \vcen{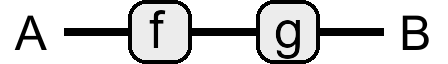}   \; \mid \;  \vcen{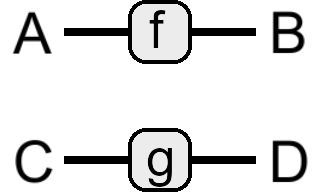} \\
    \end{array}
\end{equation*}

The first two rows of axioms in~\eqref{eq:smcaxioms}
are implicit in the  graphical representation while the axioms in the last row  are displayed as
\[ \vcen{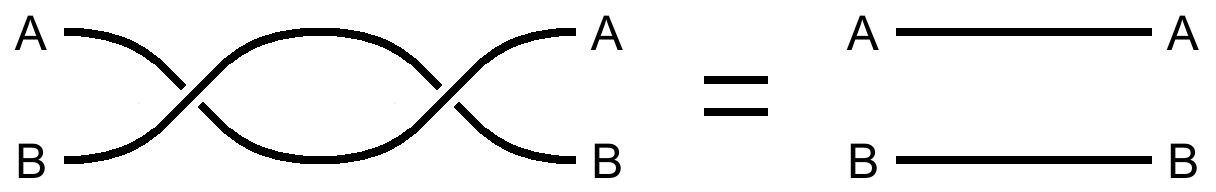} \qquad\qquad \vcen{pics/S2/symm_nat} \]

The category $\Cat{C}_\sign$ has words in $\sort^\ast$ as objects and string diagrams as arrows. Theorem 2.3 in~\cite{joyal1991geometry} states that $\Cat{C}_\sign$ is the free strict symmetric monoidal category generated by $(\sort,\sign)$. %

Wires in string diagrams may represent any object in the category. As such, it is possible to have a wire for the identity object and we can merge and split wires to combine and separate objects.
These operations can be represented in diagrammatic form as follows:

\begin{center}
	$
	\vcen{pics/S2/op_1}
	\qquad
	\vcen{pics/S2/op_2}
	\qquad
	\vcen{pics/S2/op_3}
	\qquad
	\vcen{pics/S2/op_4}
	$
\end{center}

Note that as morphisms in a category, all these operations are simply given by the identity morphism. They are purely a tool of convenience for keeping track of resources in string diagrams.
These operations satisfy the following equations:

\begin{figure}[H]
    \begin{center}
        $
        \vcen{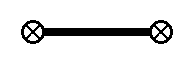}
        =
        \vcen{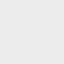}
        \qquad
        \vcen{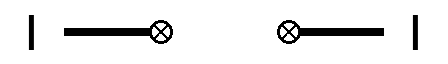}
        =
        \vcen{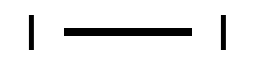}
        \qquad
        \vcen{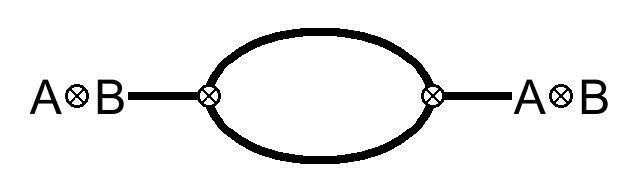}
        =
        \vcen{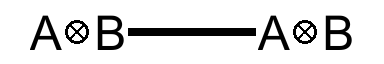}
    $
    \end{center}

    \begin{center}
        $
        \vcen{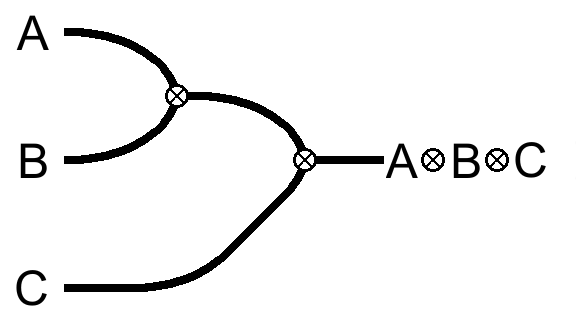}
        =
        \vcen{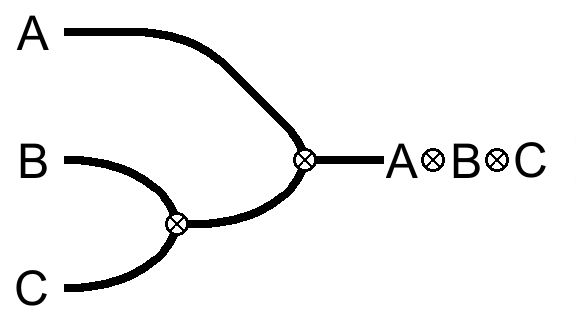}
        \qquad
        \vcen{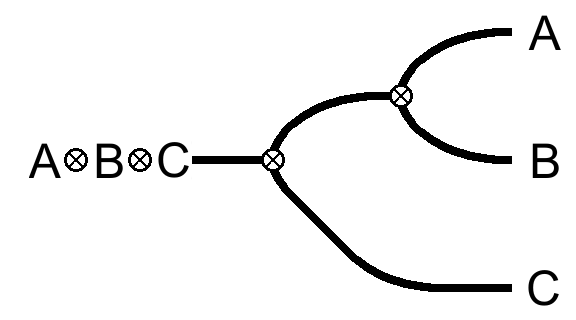}
        =
        \vcen{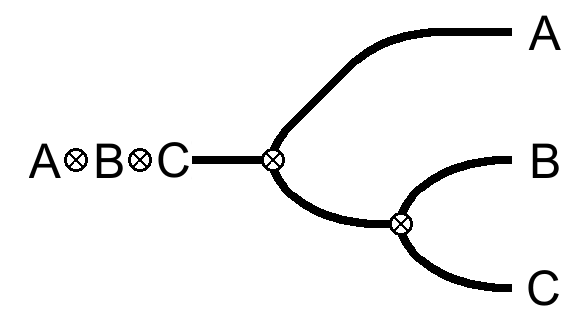}
        $
    \end{center}

    \begin{center}
        $
        \vcen{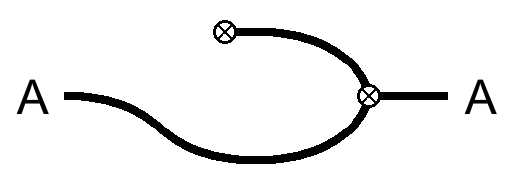}
        =
        \vcen{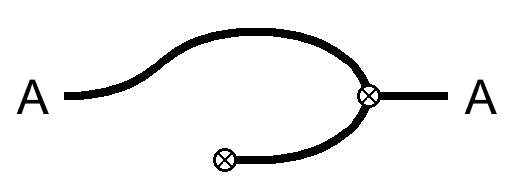}
        =
        \vcen{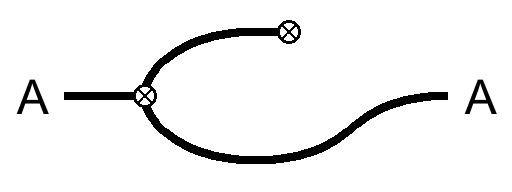}
        =
        \vcen{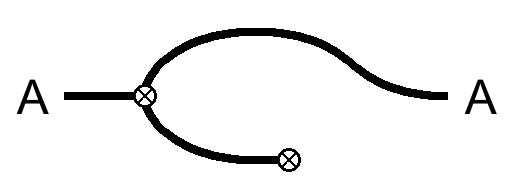}
        =
        \vcen{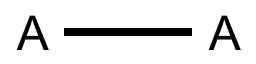}
        $
    \end{center}

    \begin{center}
        $
        \vcen{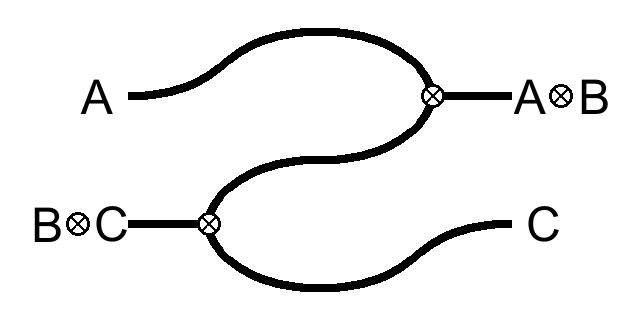}
        =
        \vcen{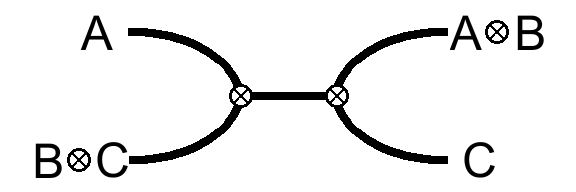}
        \qquad
        \vcen{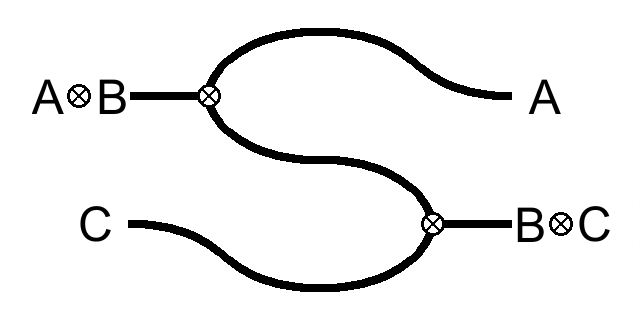}
        =
        \vcen{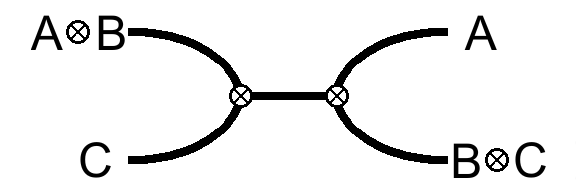}
        $
    \end{center}
    \label{fig:monoidal operations}
    \caption{}
\end{figure}

\section{Appendix to Section~\ref{sec:iteration}}

\begin{proof}[Extended proof of Lemma \ref{lem:tau-prop}]
	We prove each property using the diagonalisation method, instantiating the top level iterations by $k$. We prove the relevant equation by induction on $k$. The base case is relatively simple. We prove the induction steps below, unfolding once, then using the induction hypothesis, and finishing with an optional refolding.

	Induction step for (\ref{def:tid}):
	\begin{center}
		$\vcen{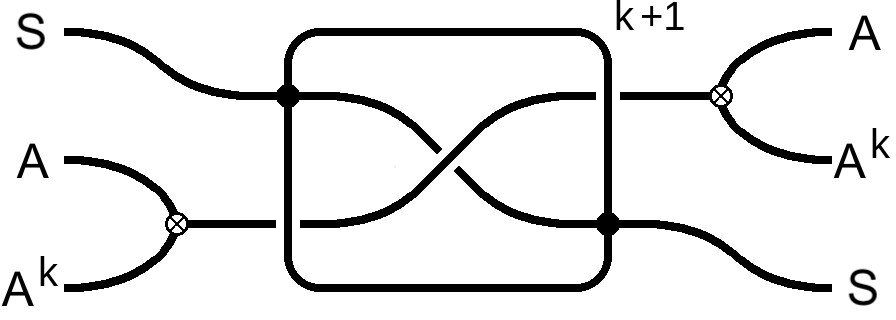} = \vcen{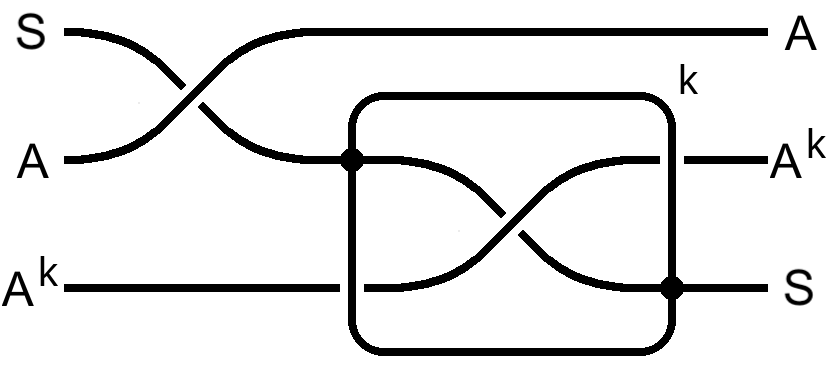} = \vcen{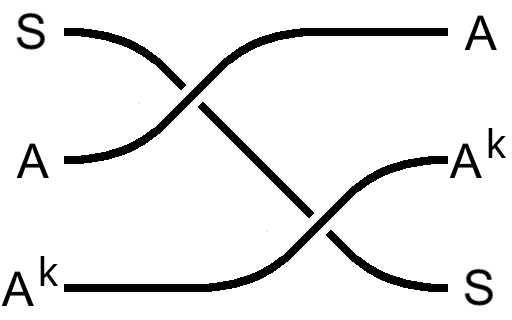} = \vcen{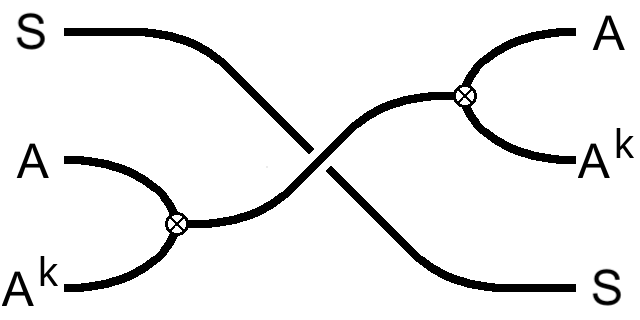}$
	\end{center}

	Induction step for (\ref{def:tswap}):
	\begin{center}
		$\vcen{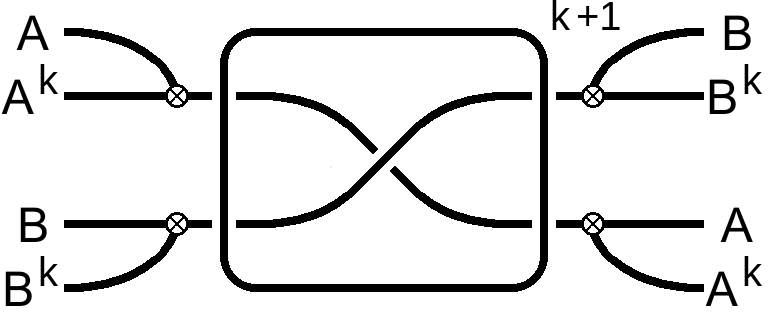} = \vcen{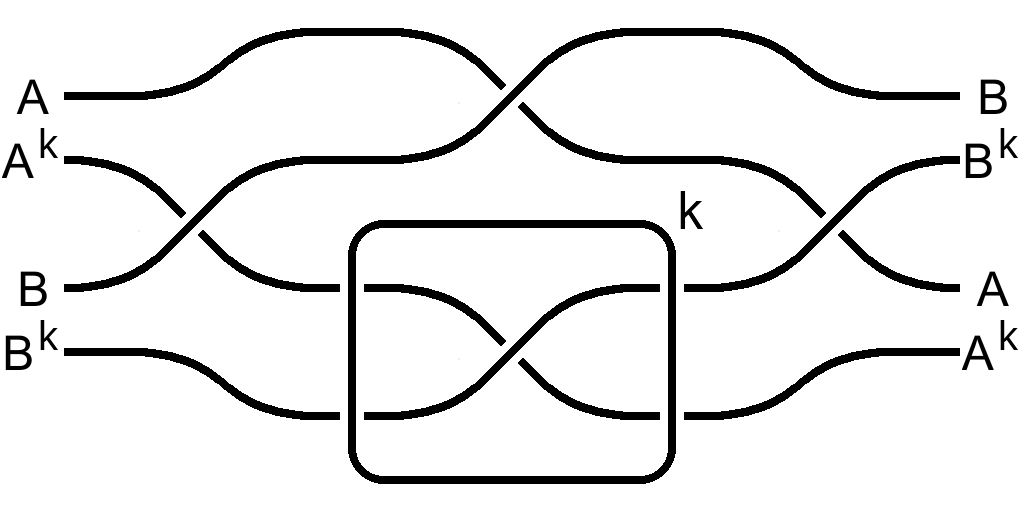} = \vcen{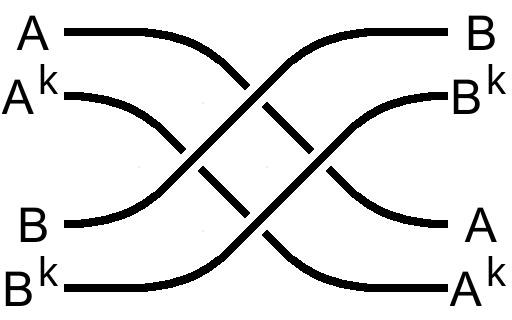} = \vcen{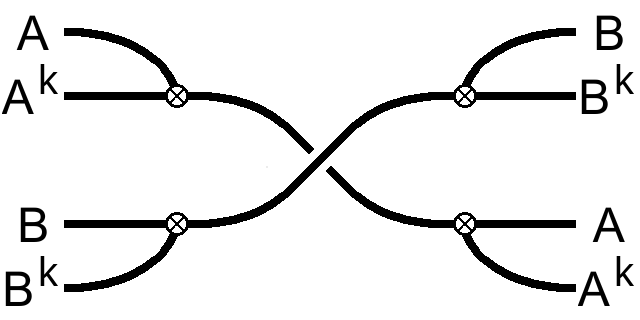}$
	\end{center}
	
	Induction step for (\ref{def:tcomp}) is done in the Section~\ref{sec:iteration}.
	
	Induction step for (\ref{def:tmono}):
	\begin{center}
		$\vcen{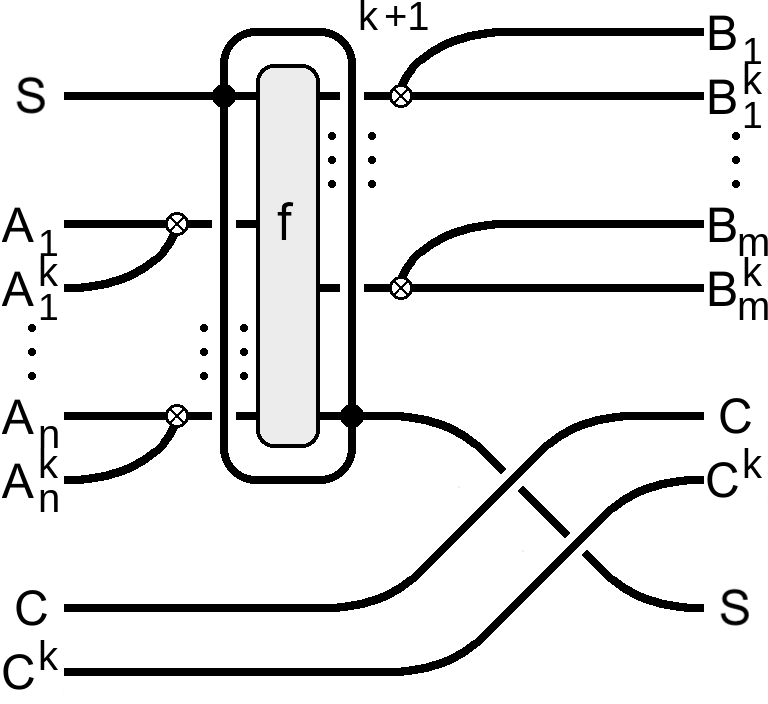} = \vcen{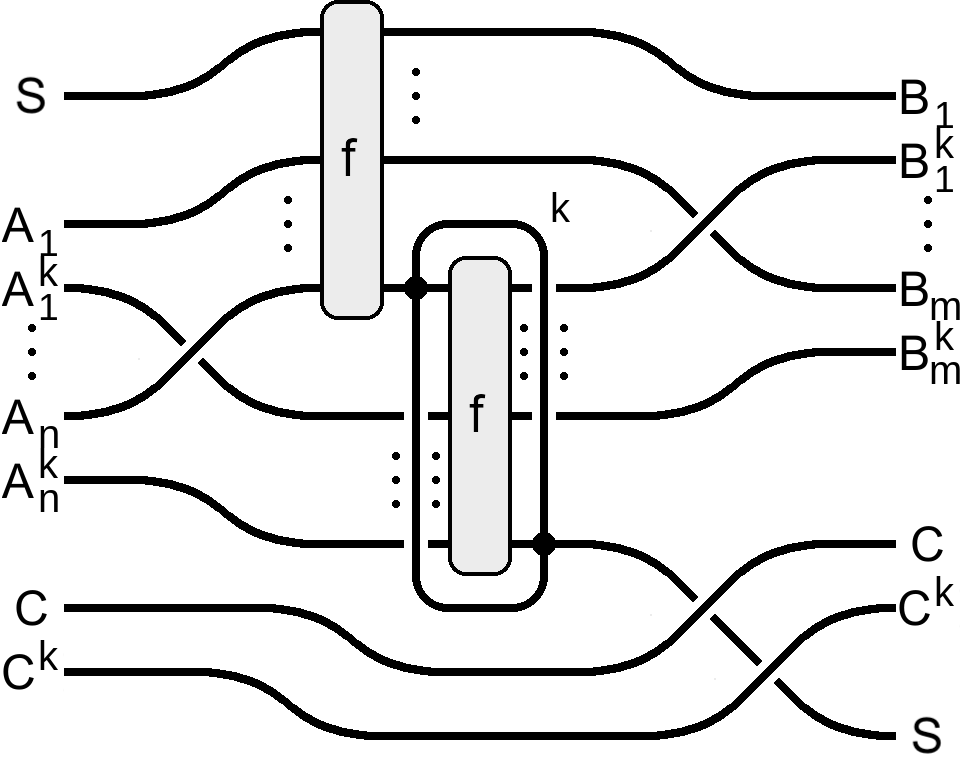} = \vcen{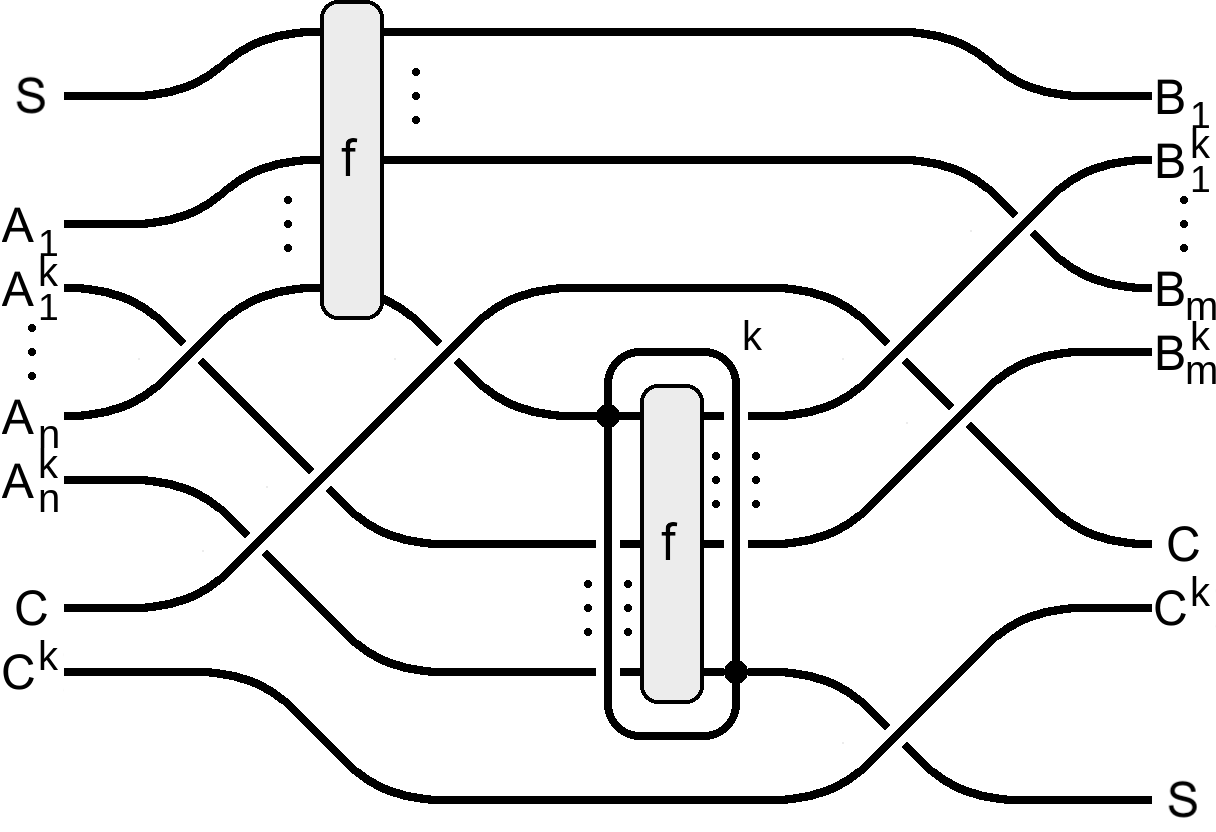} = \vcen{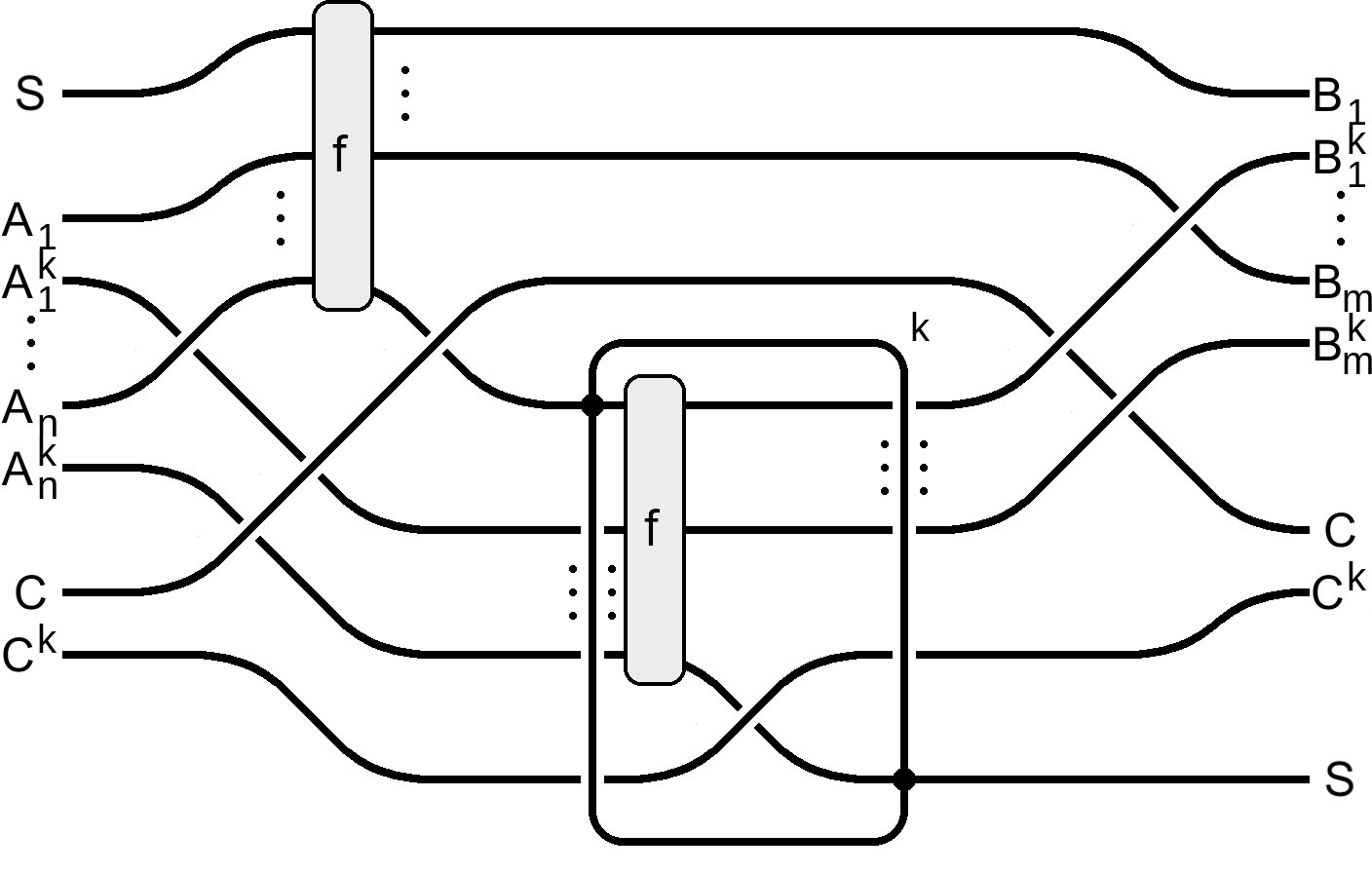} = \vcen{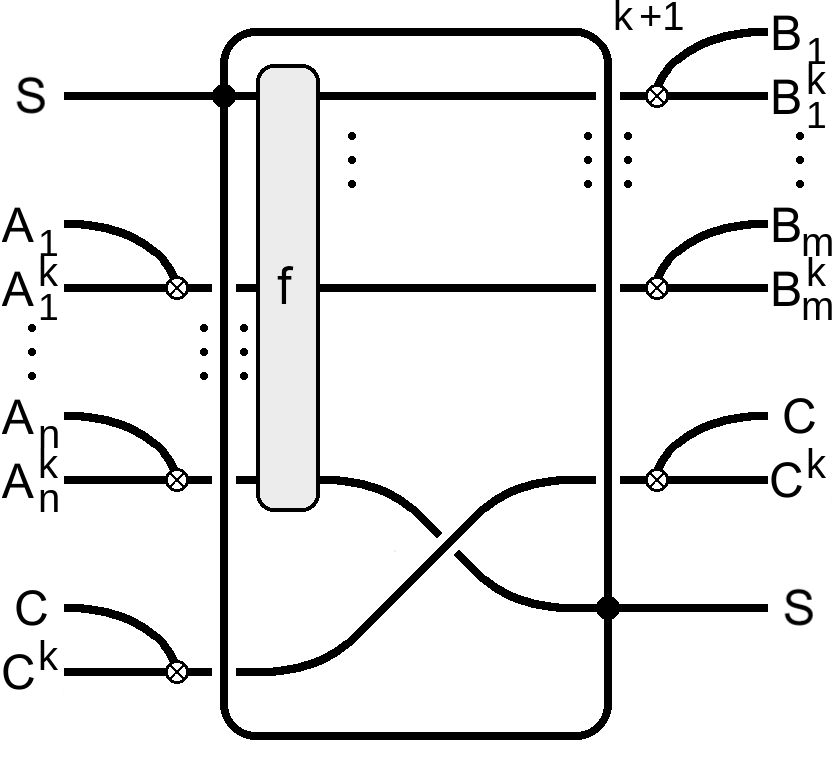}$
	\end{center}
	
\end{proof}

\section{Appendix to Section~\ref{sec:preliminaries}}\label{ap:4}

\subsection{A Complete Equational Theory of Boolean Stochastic Maps}\label{app:completeness}
The structure of Markov categories with probabilistic choice, introduced in Section~\ref{sec:preliminaries}, is novel. This makes it particularly interesting to study whether the axioms are \emph{complete} with respect to the category $\FStoch_{\{1,0\}}$. In this context, completeness means that any equality between Boolean stochastic maps that holds in $\FStoch_{\{1,0\}}$ can also be derived syntactically using the axioms listed in Figure~\ref{fig:ax markov and p choice}. 

In order to formally state a completeness theorem, we introduce the syntactic category $\MP$ of probabilistic boolean circuits.
This is the category of string diagrams freely generated by the single-sorted signature $(\Bool, \sign)$, where $\sign := \{ \vcen{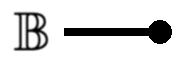}, \vcen{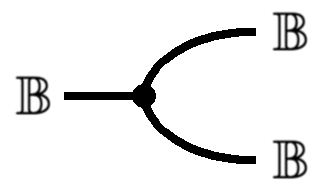}, \vcen{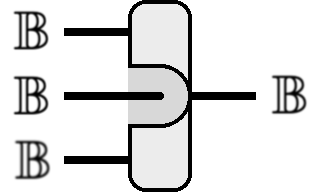}, \vcen{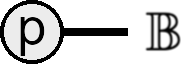} \}$ and such that the morphisms are additionally quotiented by the axioms in Figure~\ref*{fig:ax markov and p choice}.

Diagrams in $\MP$ take semantics in $\FStoch_{\{1,0\}}$ via the interpretation functor $\den{-} \colon$ ${\MP \to \FStoch_{\{1,0\}}}$, defined inductively as follows:
\[
		\den{\Bool} := \{1,0\}
		\quad
		\den{\vcen{pics/S4/A/discardB}} := \discharger{\{1,0\}}
		\quad
		\den{\vcen{pics/S4/A/copyB}} := \copier{\{1,0\}}
		\quad
		\den{\vcen{pics/S4/A/phiB}} := \phi_{\{1,0\}}
		\quad
		\den{\vcen{pics/S4/A/p}} := \langle p \rangle
\]
Then we say that the equational theory of $\MP$ is complete for $\FStoch_{\{1,0\}}$ if, for any two diagrams $c, d \in \MP$ of the same type, if $\den{c} = \den{d}$ then their equality can be derived using the axioms in Figure~\ref{fig:ax markov and p choice}. We rephrase Theorem \ref{thm:complete}.

\begin{theorem}
	The equational theory of $\MP$ is \emph{complete} for $\FStoch_{\{1,0\}}$.
\end{theorem}

\begin{proof}
	The proof shows that each element of $\FStoch_{\{1,0\}}$ has a unique normal form representation as a diagram in $\MP$, and each diagram in $\MP$ is provably equal to a diagram in normal form.
	
	We first equip each set $\Bool^n$ with a linear order.
	We then say that a formal sum $p_1 \cdot x_1 + \dots + p_m \cdot x_m$ over $\Bool^n$ is in normal form if all $p_i > 0$ and the list $x_1, \dots , x_m$ forms a sorted list of elements without repetition.
	As such, a formal sum in normal form is either of the form $1 \cdot x$, or $p \cdot x + (1-p) \cdot s$ with $p \in (0,1)$, and $s$ a formal sum in normal form, containing values $x'$ all higher than $x$ according to the chosen order.
	Importantly, each distribution has a unique normal formal sum representation.
	
	A weighted probabilistic tree over $\Bool^n$ is an inductively defined set of diagrams of type $I \to \Bool^n$:
	\[
	t,r \quad := \quad \langle b_1 \rangle \otimes \dots \otimes \langle b_n \rangle \mid (t \otimes r) ; \phi^p_{\Bool^n}
	\]
	where each $b_i \in \{0,1\}$. We represent a normal formal sum as a weighted probabilistic tree, by representing $1 \cdot (b_1,\dots,b_n)$ as $\langle b_1 \rangle \otimes \dots \otimes \langle b_n \rangle$, and $p \cdot (b_1,\dots,b_n) + (1-p) \cdot s$ as $(\langle b_1 \rangle \otimes \dots \otimes \langle b_n \rangle \otimes t) ; \phi^p_{\Bool^n}$, where $t$ is the representation of the normal formal sum $s$.
	It can be shown that any weighted probabilistic tree over $\Bool^n$ is equal to the tree representation of a normal formal sum. This is done by using the axioms of (iv) to get rid of zero-probability branches, sorting the order, and re-associating branches, as well as determinism of $\langle 1 \rangle$, $\langle 0 \rangle$ combined with the second equation of (v) to combine branches with the same result.
	
	The normal forms of type $\Bool^0 \to \Bool^n$ are the tree representations of normal formal sums, which are unique for each distribution. We generalize to arbitrary $\Bool^k \to \Bool^n$ by doing a case distinction on inputs, where a normal form of type $\Bool^{k+1} \to \Bool^n$ is of the form: $(\copier{\Bool^k} \otimes id_{\Bool}) ; (id_{\Bool^k} \otimes \chi_{\Bool^k , \Bool}) ; (f_1 \otimes id_\Bool \otimes f_0) ; \phi_{\Bool^n}$, where $f_1$ and $f_0$ are normal forms of morphisms of type $\Bool^k \to \Bool^n$.
	Each morphism $f : \Bool^k \multimap \Bool^n$ of $\FStoch_{\{1,0\}}$ has a unique normal form of this kind, with $f : \Bool^0 \to \Bool^n$ represented by the tree representation of the normal formal sum for the distribution given by $f(*)$, and $f : \Bool^{k+1} \to \Bool^n$ represented by taking $f_1$ and $f_0$ to be the normal form of the functions $(b_1,\dots,b_k) \mapsto f(b_1,\dots,b_k,1)$ and $(b_1,\dots,b_k) \mapsto f(b_1,\dots,b_k,0)$ respectively.
	This normal form is unique, as this specification is necessitated by the $\den{-}$ functor.
	
	See below examples of normal forms, with a normal form of a morphism of type $\Bool^0 \to \Bool^3$ on the left in the form of a sorted weighted tree (using lexicographic ordering on $\Bool^3$), and a normal form of a morphism of type $\Bool^3 \to \Bool^k$ with a weighted tree at each $t$ (some liberties were taken regarding the order of swaps and copies).
	\begin{center}
		\vcen{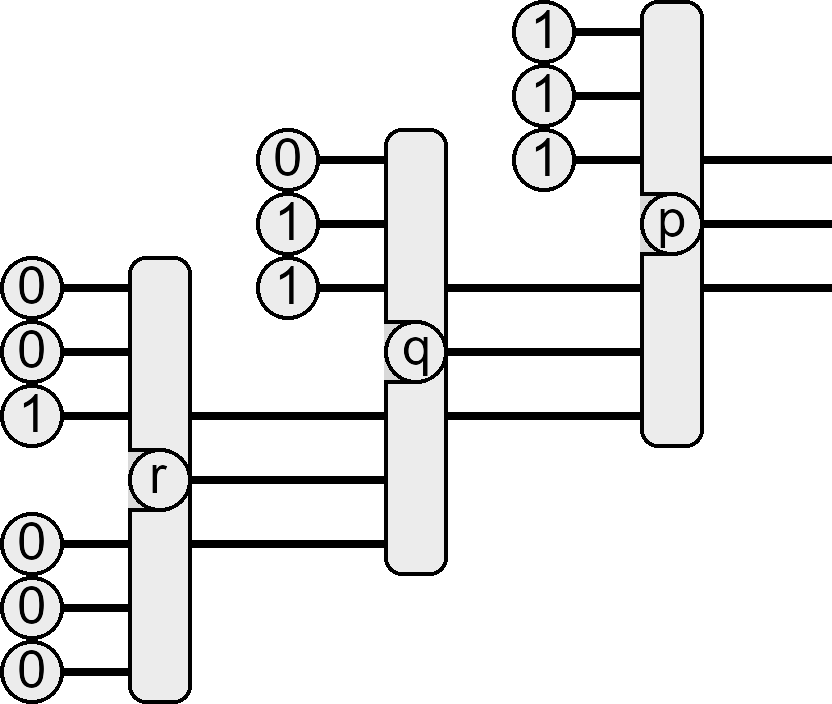}
		\qquad 
		\qquad 
		\vcen{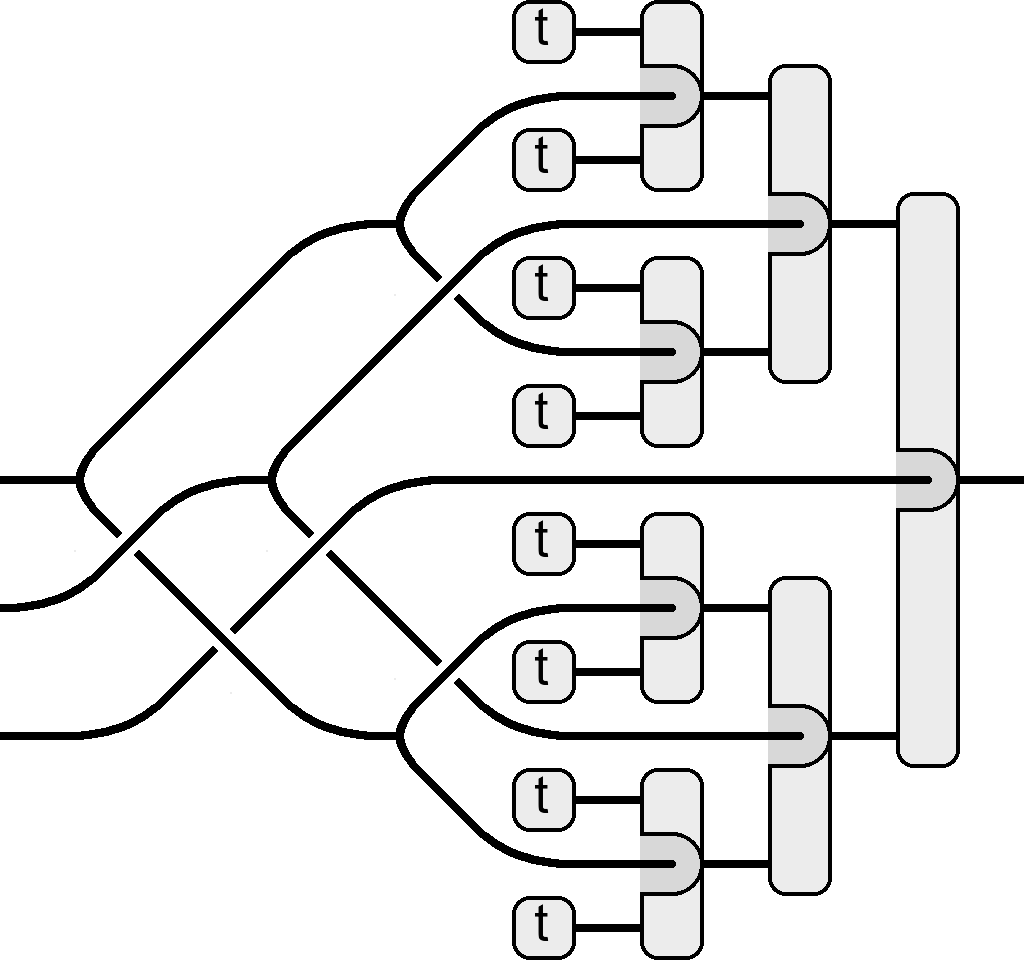}
	\end{center}
	
	We should prove that each diagram is equivalent to one in normal form. This is done in two steps.
	\begin{itemize}
		\item 
		First we show that the identity diagrams $id_{\Bool^n} : \Bool^n \to \Bool^n$ can be put in normal form. This can be shown by using the equations from (v) together with monoidal coherence and naturality of $\phi$.
		\item
		Second, note that each morphism in a freely generated monoidal category can be described as a composition of morphisms of the form $(id_A \otimes s \otimes id_B)$, with $s$ a generator, which in our case is of $\{\chi, \copier{}, \discharger{}, \langle p \rangle,\phi\}$. Then show that $f ; (id_A \otimes s \otimes id_B)$ where $f$ in normal form, can be reduced to a normal form.
		\begin{itemize}
			\item Pull $(id_A \otimes s \otimes id_B)$ through each $\pi$ using naturality, both the if-gates used to distinguish between different inputs, as well as the if-gates used in the weighted probabilistic trees, until we get to leaves of the trees.
			\item Reduce $(\langle b_1 \rangle \otimes \dots \otimes \langle b_n \rangle) ; (id_A \otimes s \otimes id_B)$ to a weighted probabilistic tree. 
			\item The result is $f$, but with different weighted probabilistic trees at each case in its case distinction. Put all those trees into normal form.
		\end{itemize}
	\end{itemize}
	Hence we can normalize each $f : \Bool^k \to \Bool^n$ by first giving it a specification as composition of $(id_A \otimes s \otimes id_B)$ morphisms. Then normalize the identity function on $\Bool^k$ and inductively normalize composing with each additional layer.

\end{proof}

\begin{lemma}\label{lemma:mp is metric enriched}
	$\MP$ is a metric enriched symmetric monoidal category.
\end{lemma}
\begin{proof}[Proof of Lemma~\ref{lemma:mp is metric enriched}]
We take as distance between arrows \[d_{\Bool^n \to \Bool^m}(c,d) := d_{\{1,0\}^n \to \{1,0\}^m}(\den{c}, \den{d})\]
for every $c,d \colon \Bool^n \to \Bool^m$ in $\MP$.

First, we verify that $d_{\Bool^n \to \Bool^m}$ is indeed a distance—that is, it satisfies the conditions in~\eqref{eq:distance}. Most properties are straightforward to check; the only non-trivial one is $d_{\Bool^n \to \Bool^m}(c, d) = 0 \iff c = d$, which we prove below:
\begin{align*}
	d_{\Bool^n \to \Bool^m}(c,d) = 0
	&\iff d_{\{1,0\}^n \to \{1,0\}^m}(\den{c}, \den{d}) = 0 \tag{Definition of $d_{\Bool^n \to \Bool^m}$} \\
	&\iff \den{c} = \den{d} \tag{Lemma~\ref{lemma:fstoch is metric enriched}} \\
	&\iff c = d \tag{Theorem~\ref*{thm:complete}}
\end{align*}

To conclude that $\MP$ is metric enriched, we need to verify also that composition and tensoring are non-expansive, that is~\eqref{eq:composition non expansive} and~\eqref{eq:tensor non expansive}. We prove~\eqref{eq:composition non expansive} below, the proof for the other is analogous.
\begin{align*}
	d_{\Bool^n \to \Bool^m}(c;d, c';d')
	&= d_{\{1,0\}^n \to \{1,0\}^m}(\den{c;d}, \den{c';d'}) \tag{Definition of $d_{\Bool^n \to \Bool^m}$} \\
	&= d_{\{1,0\}^n \to \{1,0\}^m}(\den{c};\den{d}, \den{c'};\den{d'}) \tag{$\den{-}$ is a functor} \\
	&\leq d_{\{1,0\}^n \to \{1,0\}^l}(\den{c}, \den{c'}) + d_{\{1,0\}^l \to \{1,0\}^m}(\den{d},\den{d'}) \tag{\eqref{eq:composition non expansive}} \\
	&= d_{\Bool^n \to \Bool^l}(c, c') + d_{\Bool^l \to \Bool^m}(d, d') \tag{Definition of $d_{\Bool^n \to \Bool^m}$}
\end{align*}
\end{proof}

\subsection{Other results}

\begin{proof}[Extended proof of Lemma \ref{lemma:PI on markov with p choice}]
	We define copy, discard and if-gate on objects as done in the main body of the paper.
	We then show that in $\C$, $\tau^k_{1,(A),(A,A)}(\copier{A}) = \copier{A^k}$, $\tau^k_{1,(A),()}(\discharger{A}) = \discharger{A^k}$, and
	
	\begin{center}
		$\vcen{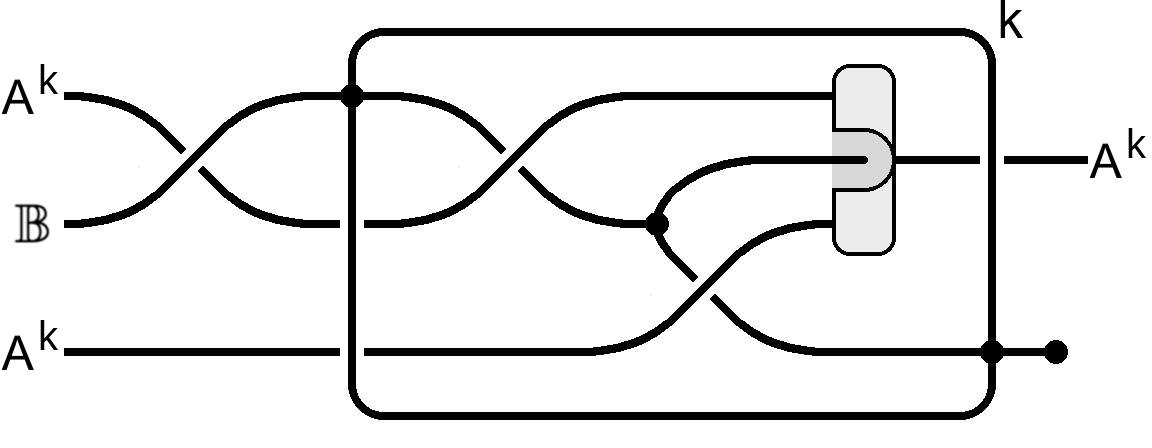} = \vcen{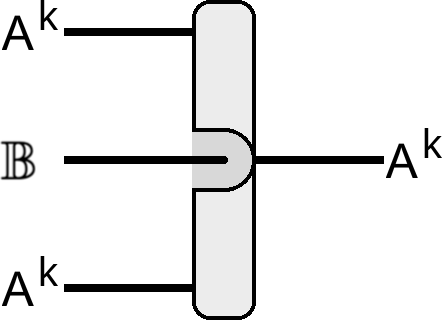}$
	\end{center}
	
	Consequently, by induction on $A$ in $\PI{\C}$, we can show that with the chosen definition, $S_k(\copier{A}) = \copier{S_k(A)}$, $S_k(\discharger{A}) = \discharger{S_k(A)}$, and $S_k(\phi_{A}) = \phi_{S_k(A)}$.
	This allows us to lift the equations for Markov categories and probabilistic choice to $\PI{\C}$ via $S_k$.

	We prove the relevant property for $\phi$ by induction on $k$. The induction basis is simple. The induction step goes as follows:
	
	\begin{center}
		$\vcen{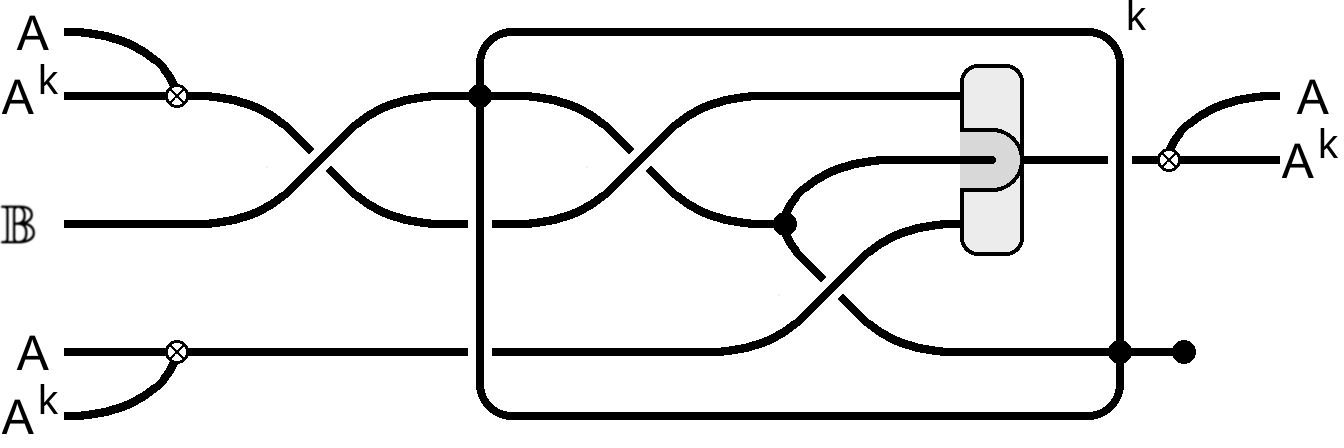} = \vcen{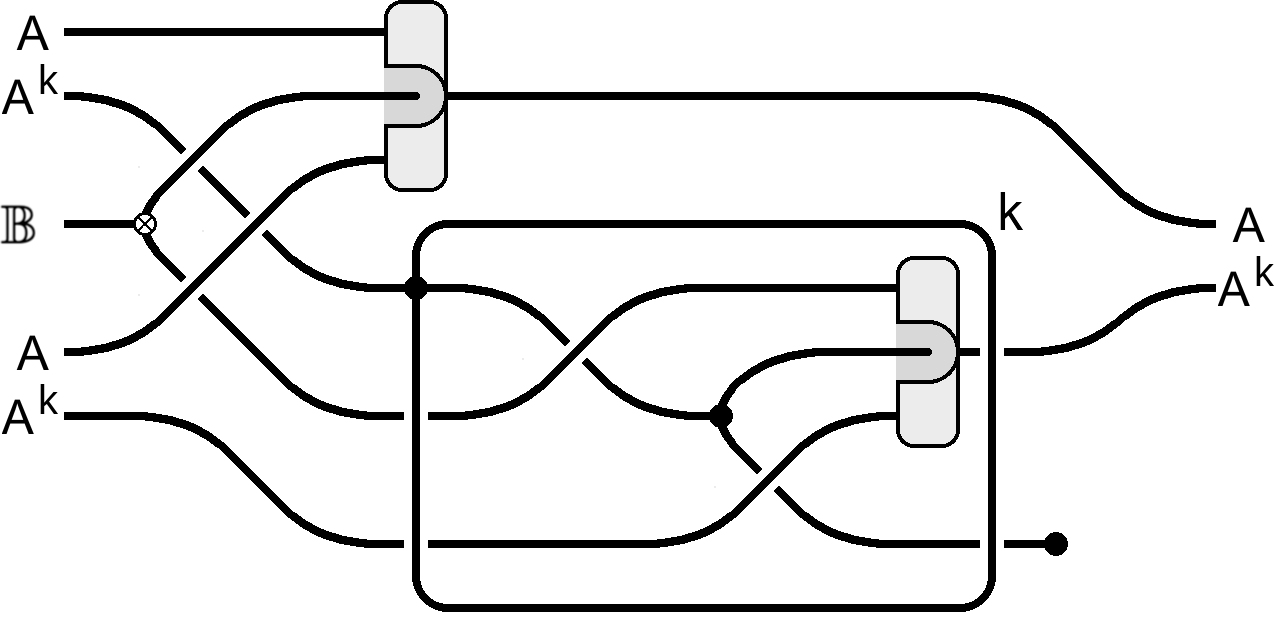} = \vcen{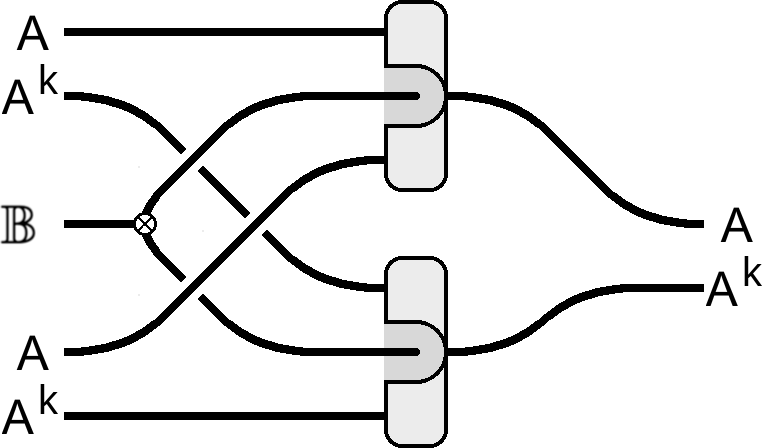}= \vcen{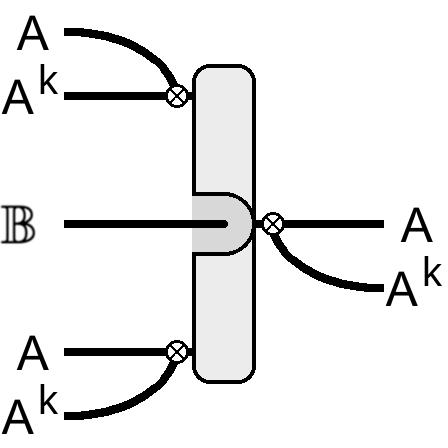}$
	\end{center}
	
	Using the above property, we can show by induction on objects $A$ that $S_k(\phi_{A}) = \phi_{S_k(A)}$. This allows us to prove the relevant equations by lifting them from $\C$.
\end{proof}

\begin{lemma}\label{lemma:tv alternative}
    $\sum_{a \in \texttt{sup}(v) \cup \texttt{sup}(w)} \texttt{abs}(v(a) - w(a))/2 = 1- \sum_{a \in \texttt{sup}(v) \cap \texttt{sup}(w)} \texttt{min}(v(a) , w(a))$.
\end{lemma}
\begin{proof}
    Note that 
    \[ \sum_{a \in \texttt{sup}(v) \cap \texttt{sup}(w)} \texttt{min}(v(a) , w(a)) = \sum_{a \in \texttt{sup}(v) \cup \texttt{sup}(w)} \texttt{min}(v(a) , w(a)) \]
    and
    \begin{align*}
        & 2 \cdot \sum_{a \in \texttt{sup}(v) \cup \texttt{sup}(w)} \texttt{min}(v(a) , w(a)) + \sum_{a \in \texttt{sup}(v) \cup \texttt{sup}(w)} \texttt{abs}(v(a) - w(a))
        \\
        = & \sum_{a \in \texttt{sup}(v) \cup \texttt{sup}(w)} \texttt{min}(v(a) , w(a)) + \sum_{a \in \texttt{sup}(v) \cup \texttt{sup}(w)} \texttt{max}(v(a) , w(a))
        \\
        = & \sum_{a \in \texttt{sup}(v) \cup \texttt{sup}(w)} (v(a) + w(a)) 
        \\
        = & \sum_{a \in \texttt{sup}(v)} (v(a)) + \sum_{\texttt{sup}(w)} w(a) 
        \\
        = & 2
    \end{align*}
\end{proof}

\begin{proof}[Proof of Lemma \ref{lem:split}]
	For the first diagram: \vcen{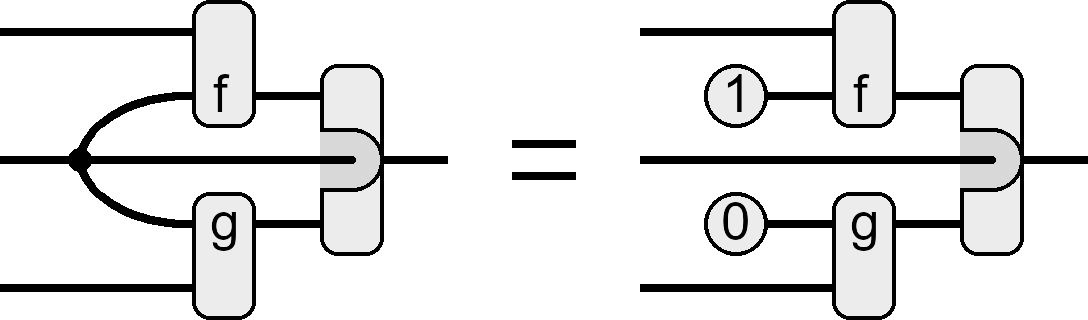}
	
	First we apply the (v) axioms to each of the incoming wires, and combining those with monoidal coherence.
	Then we pull the main morphism through the $\phi$-gate using naturality. 
	As a third step, we use determinism of $1$ and $0$, as well as the $\phi$ reduction steps.
	Lastly, we use naturality of the discard operation to get the right result.
		
	\noindent
	\vcen{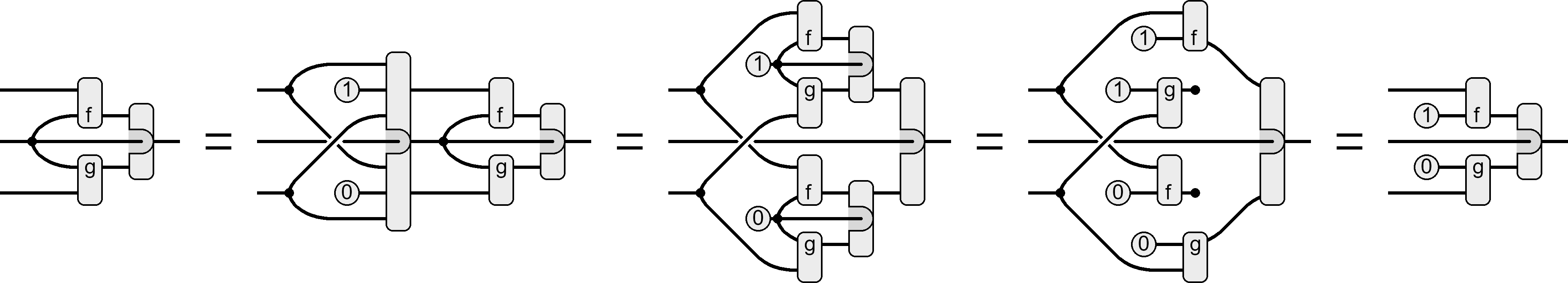}
	
	For the second diagram:	\vcen{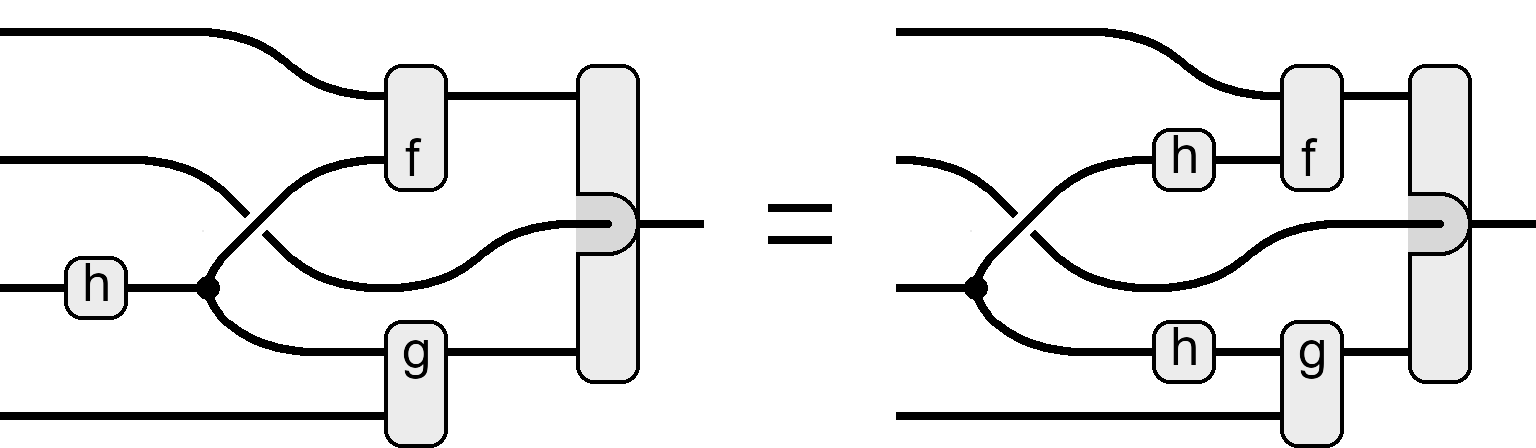}
	
	For simplicity, we define: \vcen{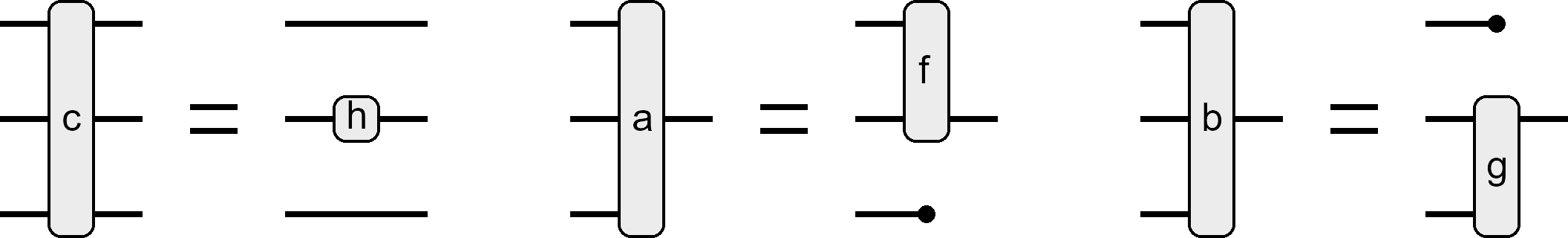}, and swap the first two wires.
	We then prove as follows, starting with the equations from (v) with monoidality of if, then use naturality of if, and finally simplifying using equations from (iv):
		
	\vcen{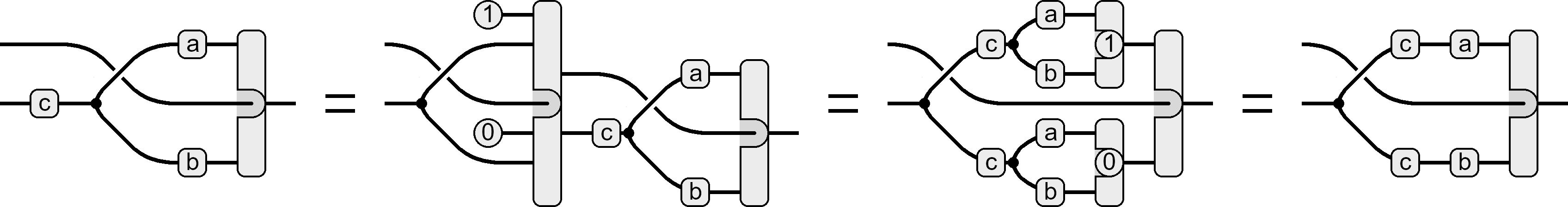}
\end{proof}

\begin{proof}[Proof of Lemma~\ref{lemma:quantitative soundness and completeness}]
    For soundness we proceed by induction on the rules in~\eqref{eq:metric}. Observe that the laws that do not involve the probabilistic choice are trivially verified from the fact that $\FStoch_{\{1,0\}}$ is metric-enriched.

    For the top-last rule, given maps $c,c' \colon \{1,0\}^n \to \{1,0\}^m$ and $d,d' \colon \{1,0\}^l \to \{1,0\}^m$ such that $c \equiv_\delta c'$ and $d \equiv_\delta d'$, by inductive hypothesis we have that $d_{\{1,0\}^n \to \{1,0\}^m}(c,c') \leq \delta$ and $d_{\{1,0\}^l \to \{1,0\}^m}(d,d') \leq \delta$.
    Now let $h = (c \otimes id_\Bool \otimes d);\phi_{\{1,0\}^m}$ and $h' = (c' \otimes id_\Bool \otimes d');\phi_{\{1,0\}^m}$ and observe that for any $a \in \{1,0\}^n, b \in \{1,0\}^l$, we have the following cases
    \[
    \begin{array}{cc}
        h(a,0,b) = {(c \otimes \discharger{\{1,0\}^l})}(a,b)
        &
        h'(a,0,b) = {(c' \otimes \discharger{\{1,0\}^l})}(a,b)
        \\
        h(a,1,b) = {(\discharger{\{1,0\}^n} \otimes d)}(a,b)
        &
        h'(a,1,b) = {(\discharger{\{1,0\}^n} \otimes d')}(a,b)
    \end{array}
    \]
    and thus, $d_{\{1,0\}^m}(h(a,0,b), h'(a,0,b)) = d_{\{1,0\}^m}({(c \otimes \discharger{\{1,0\}^l})}(a,b), {(c' \otimes \discharger{\{1,0\}^l})}(a,b)) \leq \delta$
    and $d_{\{1,0\}^m}({h}(a,1,b), {h'}(a,1,b)) = d_{\{1,0\}^m}({(\discharger{\{1,0\}^n} \otimes d)}(a,b), {(\discharger{\{1,0\}^n} \otimes d')}(a,b)) \leq \delta$.

    To conclude we have that
    \begin{align*}
        d_{\{1,0\}^{n+1+l} \to \{1,0\}^m}(h,h')
        &= \texttt{max}_{(a,i,b) \in \{1,0\}^{n+i+l}} \tv_{\{1,0\}^m}({h}(a,i,b), {h'}(a,i,b)) \\
        &\leq \texttt{max}_{(a,i,b) \in \{1,0\}^{n+i+l}} \delta \\
        &= \delta.
    \end{align*}

    For the bottom-last rule, given maps $c,c' \colon \{1,0\}^n \to \{1,0\}^m$ and $d,d' \colon \{1,0\}^l \to \{1,0\}^m$ such that $c \equiv_\delta c'$ and $d \equiv_{\delta'} d'$, by inductive hypothesis we have that $d_{\{1,0\}^n \to \{1,0\}^m}(c,c') \leq \delta$ and $d_{\{1,0\}^l \to \{1,0\}^m}(d,d') \leq \gamma$.
    Now let $h = (c \otimes d);\phi^p_{\{1,0\}^m}$ and $h' = (c' \otimes d');\phi^p_{\{1,0\}^m}$ and observe that for any $a \in \{1,0\}^n, b \in \{1,0\}^l$, we have that
    \[
    \begin{array}{cc}
        {h}(a,b) = p\ket{{c}(a)} + (1-p)\ket{{d}(b)}
        &
        {h'}(a,b) = p\ket{{c'}(a)} + (1-p)\ket{{d'}(b)}
    \end{array}
    \]
    and thus,
    \begin{align*}
        &d_{\{1,0\}^{n+l} \to \{1,0\}^m}(h,h') \\
        = &\texttt{max}_{(a,b) \in \{1,0\}^{n+l}} \tv_{\{1,0\}^m}({h}(a,b), {h'}(a,b)) \\
        = &\texttt{max}_{(a,b) \in \{1,0\}^{n+l}} \sum_x \frac{| {h}(a,b)(x) - {h'}(a,b)(x) |}{2} \\
        = &\texttt{max}_{(a,b) \in \{1,0\}^{n+l}} \sum_x \frac{| p\ket{{c}(a)(x) - {c'}(a)(x)} + (1-p)\ket{{d}(b)(x) - {d'}(b)(x)} |}{2} \\
        \leq &p \cdot \texttt{max}_{a \in \{1,0\}^n} \sum_x \frac{| {c}(a)(x) - {c'}(a)(x) |}{2} + (1-p) \cdot \texttt{max}_{b \in \{1,0\}^l} \sum_x \frac{| {d}(b)(x) - {d'}(b)(x) |}{2} \\
        = &p \cdot d_{\{1,0\}^{n} \to \{1,0\}^m}({c},{c'}) + (1-p) \cdot d_{\{1,0\}^{l} \to \{1,0\}^m}({d},{d'}) \\
        \leq &p\cdot \delta + (1-p) \cdot \gamma.
    \end{align*}

    \bigskip

    For completeness, we rely on the normal form Theorem~\ref{thm:complete} and use the rules from \eqref{eq:metric}.
    Take two morphisms $f, g$ from $n \to m$ in normal form.
    We prove that $f \equiv_{d(f,g)} g$ by induction on $n$ and $m$.
    \begin{itemize}
        \item Base case: $n = m = 0$, then both morphisms are $id_0$ and hence equal. So $f \equiv_0 g$ and $d(f,g) = 0$.

        \item Suppose $n = 0$ and $m = m' + 1$. Then $f = (\langle 1 \rangle \otimes f_1 \otimes \langle 0 \rangle \otimes f_0) ; \phi^{p}$ and  $g = (\langle 1 \rangle \otimes g_1 \otimes \langle 0 \rangle \otimes g_0) ; \phi^{q}$.
        Now, $d(f,g) = \min(p,q) \cdot d(f_1,g_1) + \min(1-p,1-q) \cdot d(f_0,g_0) + (1-\min(p,q)-\min(1-p,1-q))$, since $d(\langle i \rangle \otimes f_i , \langle i \rangle \otimes g_i) = d(f_i , g_i)$, and $d(\langle i \rangle \otimes f_i , \langle 1-i \rangle \otimes g_{i-1}) = 1$.

        Assume w.l.o.g. that $p \leq q$, so $\min(p,q) = p$ and $\min(1-p,1-q) = 1-q$, and $d(f,g) = p \cdot d(f_1,g_1) + (1-q) \cdot d(f_0 , g_0) + (q-p)$. Define $k = \frac{q-p}{1-p}$. We can then prove that $f = (\langle 0 \rangle \otimes f_0 \otimes \langle 0 \rangle \otimes f_0) ; (\langle 1 \rangle \otimes f_1 \otimes \phi^k) ; \phi^p$ using \ref{ax:phiIdemp} \ref{ax:phidistr}, and $g = (\langle 1 \rangle \otimes g_1 \otimes \langle 0 \rangle \otimes g_0) ; (\langle 1 \rangle \otimes g_1 \otimes \phi^k) ; \phi^p$ using the same rules and \ref{ax:phias}.
        By induction hypothesis, $\langle 1 \rangle \otimes f_1 \equiv_{d(f_1, g_1)} \langle 1 \rangle \otimes g_1$ and $\langle 0 \rangle \otimes f_0 \equiv_{d(f_0, g_0)} \langle 0 \rangle \otimes g_0$, and trivially  $\langle 0 \rangle \otimes f_0 \equiv_1 \langle 1 \rangle \otimes g_1$.
        Applying the second quantitative $\phi$-rule twice, we get that $f \equiv_\alpha g$ for: $\alpha = p \cdot d(f_1,g_1) + (1-p) \cdot (k \cdot 1 + (1-k) \cdot d(f_0,g_0)) = p \cdot d(f_1,g_1) + (p-q) + (1-q) \cdot d(f_0,g_0)$ as required.

        \item Suppose $n = n'+1$. Then $f = (id_{1} \otimes \delta_{n'}) ; (\gamma_{1,n'} \otimes id_{n'}) ; (f_1 \otimes id_1 \otimes f_0) ; \phi_{n'}$ and $g = (id_{1} \otimes \delta_{n'}) ; (\gamma_{1,n'} \otimes id_{n'}) ; (g_1 \otimes id_1 \otimes g_0) ; \phi_{n'}$.
        Note that 
        \begin{align*}
            d(f,g) &= \max_{(i,x) \in \{1,0\}^{n}}(d(f(i,x) , g(i,x))) \\
            &= \max_{(i,x) \in \{1,0\}^{n}}(d(f_i(x) , g_i(x))) \\
             &= \max(\max_{x \in \{1,0\}^{n'}}(d(f_1(x) , g_1(x))) , \max_{x \in \{1,0\}^{n'}}(d(f_0(x) , g_0(x)))) \\
             &= \max(d(f_1,g_1) , d(f_0,g_0)).
        \end{align*}

        By induction hypothesis, $f_1 \equiv_{d(f_1, g_1)} g_1$ and $f_0 \equiv_{d(f_0,g_0)} g_0$. So with weakening,
        by induction hypothesis, $f_1 \equiv_{d(f, g)} g_1$ and $f_0 \equiv_{d(f,g)} g_0$ (since $d(f,g)$ is at least as large as the  other distances).
        Applying the first quantitative $\phi$-rule, we get:
        $(f_1 \otimes id_1 \otimes f_0) ; \phi_{n'} \equiv_{d(f, g)} (g_1 \otimes id_1 \otimes g_0) ; \phi_{n'}$, so with compositionality, $f \equiv_{d(f, g)} g$ as required.
    \end{itemize}
\end{proof}

\section{Appendix to Section~\ref{sec:asyeq}}

\begin{proof}[Proof of Proposition~\ref{prop:asyeq is a congruence}]
    Suppose $f \equiv f'$ and $g \equiv g'$. Take $k \in \Nat$, and $\varepsilon > 0$, then by assumption on $\varepsilon/2$ there are $N$ and $M$ such that if $n \geq N$ and $n \geq M$, then respectively $n^k \cdot d(S_n(f) , S_n(f')) < \varepsilon/2$ and $n^k \cdot d(S_n(g) , S_n(g')) < \varepsilon/2$.
    Taking $K = \max(N,M)$, then for $n \geq K$ we get $n^k \cdot d(S_n(f ; g) , S_n(f' ; g')) = n^k \cdot d(S_n(f) ; S_n(g) , S_n(f') ; S_n(g')) \leq n^k \cdot (d(S_n(f) , S_n(f')) + d(S_n(g) , S_n(g'))) = n^k \cdot d(S_n(f) , S_n(f')) + n^k \cdot d(S_n(g) , S_n(g')) < \varepsilon/2 + \varepsilon/2 = \varepsilon$.
    Same reasoning may be applied to the monoidal product operation.
\end{proof}

	We give a symbolic proof of the fact that the $\tau^\star$ operation preserves asymptotic equivalence. This a corollary of Lemma \ref{lem:Newton2}.

	Suppose $f \equiv f'$ of appropriate type to apply $\tau^\star_S$. Remember that $d(S_n(\tau^\star_S(f)) , S_n(\tau^\star_S(g))) = d(\tau^n(S_n(f)) , \tau^n(S_n(g))) \leq n \cdot d(S_n(f) , S_n(g))$.
	Take $k \in \Nat$ and $\varepsilon > 0$, and apply assumption to $k + 1$ to get $N > 0$ such that for all $n \geq N$, $n^{k+1} \cdot d(S_n(f), S_n(g)) < \varepsilon$.
	So $n^k \cdot d(S_n(\tau^\star(f)) , S_n(\tau^\star(g))) = n^k \cdot d(\tau^n(S_n(f)) , \tau^n(S_n(g))) \leq n^{k+1} \cdot d(S_n(f) , S_n(g)) < \varepsilon$.

\begin{proof}[Proof of Lemma \ref{lem:all1} in symbolic notation]
	
	First, observe that by functoriality of $\tau$, we can combine the whole morphism into on $\langle 1 \rangle ; \tau^n_{\Bool , () , (\Bool)}(f)$, where $f = ((\langle p \rangle ; \copier{\Bool}) \otimes \textit{id}_{\Bool}) ; (\textit{id}_{\Bool} \otimes \textsf{\&})$.
	
	By induction on $n$, we can prove that $\langle 1 \rangle ; \tau^n_{\Bool , () , (\Bool)}(f) \equiv_{p^n} \tau^n_{1 , () , (\Bool)}(\langle p \rangle) \otimes \langle 0 \rangle$:
		\begin{itemize}
			\item For $n = 0$, the required distance $\equiv_1$ is trivially true.
			\item For $n = m+1$, 
			$\langle 1 \rangle ; \tau^n_{\Bool , () , (\Bool)}(f) = 
			\langle 1 \rangle ; \tau^m_{\Bool , () , (\Bool)}(f) ; (\textit{id}_{\Bool^m} \otimes f)$.
			Now the last $f$ has a single weighted choice, and we can rewrite the morphism as a weighted choice between two options, using that generally $g ; (h \oplus_p k) = (g ; h) \oplus_p (g ; k)$. We analyze the two options.
			\begin{itemize}
				\item If the last $f$ creates a $1$, then the final result will be $\langle 1 \rangle ; \tau^m_{\Bool , () , (\Bool)}(f) ; (\textit{id}_{\Bool^m } \otimes \langle 1 \rangle \otimes \textit{id}_{\Bool})$, since we add a 1 at the end and keep the all-1 value from the initial $m$ steps. By induction hypothesis, this is $p^m$ away from $\tau^m_{1 , () , (\Bool)}(\langle p \rangle) \otimes \langle 1 \rangle \otimes \langle 0 \rangle$
				\item If the last $f$ creates a $0$, then the final result will be $\langle 1 \rangle ; \tau^m_{\Bool , () , (\Bool)}(f) ; (\textit{id}_{\Bool^n} \otimes \discharger{\Bool} \otimes \langle 0 \rangle \otimes \langle 0 \rangle)$, since we add a 0 at the end and replace the all-1 value to 0. This is equivalent to $\tau^m_{1 , () , (\Bool)}(\langle p \rangle) \otimes \langle 0 \rangle \otimes \langle 0 \rangle$.
			\end{itemize}
			Hence the morphism is equal to:
			$(\langle 1 \rangle ; \tau^m_{\Bool , () , (\Bool)}(f) ; (\textit{id}_{\Bool^m } \otimes \langle 1 \rangle \otimes \textit{id}_{\Bool})) \oplus_p (\tau^m_{1 , () , (\Bool)}(\langle p \rangle) \otimes \langle 0 \rangle \otimes \langle 0 \rangle)$, which is at most $p \cdot p^m 0 p^n$ away from $(\tau^m_{1 , () , (\Bool)}(\langle p \rangle) \otimes \langle 1 \rangle \otimes \langle 0 \rangle) \oplus_p (\tau^m_{1 , () , (\Bool)}(\langle p \rangle) \otimes \langle 0 \rangle \otimes \langle 0 \rangle)$ which can be shown using the induction hypothesis.
			We can pull out the commonalities to see this is equal to $\tau^m_{1 , () , (\Bool)}(\langle p \rangle) \otimes (\langle 1 \rangle \oplus_p \langle 0 \rangle) \langle 0 \rangle$, where $\langle 1 \rangle \oplus_p \langle 0 \rangle = \langle p \rangle$. Hence this is equal to $\tau^n_{1 , () , (\Bool)}(\langle p \rangle) \otimes \langle 0 \rangle$.
		\end{itemize}
\end{proof}

Equality checking on $\star$-tuples given equality check on bits $\textsf{eq}$, symbolically:
\[
\textsf{eq}^\star \quad := \quad \tau^\star_{1 , (\Bool , \Bool) , (\Bool)}(\textsf{eq}) ; \textsf{all-1} \quad = \quad (\langle 1 \rangle \otimes \textit{id}_{\Bool^\star \otimes \Bool^\star}) ; \tau^\star_{\Bool , (\Bool , \Bool) ,  ()}((\textit{id}_\Bool \otimes \textsf{eq}) ; \textsf{\&})
\]

\begin{proof}[Proof of Lemma \ref{lem:keyguess} in symbolic notation]
	
	By compositional properties, 
	
	$(\tau^{\star}_{1 , () , (\Bool)}(\langle 1/2 \rangle) \otimes \textit{id}_{\Bool^\star}) ; (\copier{\Bool^\star} \otimes \textit{id}_{\Bool^\star}) ; (\textit{id}_{\Bool^\star} \otimes \textsf{eq}^\star) = 
	(\langle 1 \rangle \otimes \textit{id}_{\Bool^\star}) ; \tau^\star_{\Bool , () , (\Bool)}((\textit{id}_\Bool \otimes (\langle 1/2 \rangle ; \copier\Bool) \otimes \textit{id}_\Bool) ; (\chi_{\Bool,\Bool} \otimes \textsf{eq}^\star) ; (\textit{id}_\Bool \otimes \textsf{\&}))$.
	
	By comparing distributions we can see that $(\langle 1/2 \rangle \otimes \textit{id}_\Bool) ; (\copier\Bool \otimes \textit{id}_\Bool) ; (\textit{id}_\Bool \otimes \textsf{eq}^\star) = (\textit{id}_\Bool \otimes \langle 1/2 \rangle) ; (\textit{id}_\Bool \otimes \copier\Bool) ; (\textsf{eq}^\star \otimes \textit{id}_\Bool)$, so the content of the iteration can be rewritten to $(\textit{id}_\Bool \otimes (\langle 1/2 \rangle ; \copier\Bool) \otimes \textit{id}_\Bool) ; (\chi_{\Bool,\Bool} \otimes \textsf{eq}^\star) ; (\textit{id}_\Bool \otimes \textsf{\&}) = ((\langle 1/2 \rangle ; \copier\Bool) \otimes \textit{id}_{\Bool^2}) ; (\textit{id}_{\Bool} \otimes \textsf{\&} \otimes \textit{id}_{\Bool}) ; (\textit{id}_{\Bool} \otimes \textsf{eq})$.
	
	Using compositionality again, we get 
	
	$(\textit{id}_{\Bool^\star} \otimes (\tau^{\star}_{1 , () , (\Bool)}(\langle 1/2 \rangle); \copier{\Bool^\star} ; (\textit{id}_{\Bool^\star} \otimes \textsf{all-1}))) ; (\tau^\star_{1 , (\Bool) , (\Bool)}(\textsf{eq}) \otimes \textit{id}_{\Bool})$, which by the previous lemma is asymptotically equivalent to: $(\textit{id}_{\Bool^\star} \otimes \tau^{\star}_{1 , () , (\Bool)}(\langle 1/2 \rangle) \otimes \langle 0 \rangle) ; (\tau^\star_{1 , (\Bool) , (\Bool)}(\textsf{eq}) \otimes \textit{id}_{\Bool}) = \tau^{\star}_{1 , (\Bool) , (\Bool)}((\langle 1/2 \rangle \otimes \textit{id}_{\Bool}) ; \textsf{eq}) \otimes \langle 0 \rangle = \tau^{\star}_{1 , (\Bool) , (\Bool)}(\discharger\Bool ; \langle 1/2 \rangle) \otimes \langle 0 \rangle = \discharger{\Bool^\star} ; (\tau^\star_{1 , (\Bool) , (\Bool)}(\langle 1/2 \rangle) \otimes \langle 0 \rangle)$.
\end{proof}

\end{document}